\documentclass[11pt]{article}
\pdftrailerid{}
\usepackage[margin=1in]{geometry}
\usepackage{amsmath,amssymb,amsthm,mathtools,bm,mathrsfs}
\usepackage{enumitem}
\usepackage{booktabs}
\usepackage{array}
\usepackage{tabularx}
\usepackage{graphicx}
\usepackage{subcaption}
\usepackage{placeins}
\usepackage{float}
\graphicspath{{figures/}{./}}
\usepackage{microtype}
\usepackage{url}
\usepackage[hidelinks]{hyperref}
\hypersetup{
  pdftitle={Bayesian Confidence Recalibration and Research-Equilibrium Criticality: Temporal Support in Robust Portfolios},
  pdfauthor={Han Yan\c{c}},
  pdfsubject={Mathematical Finance manuscript},
  pdfkeywords={confidence recalibration, Bayesian ambiguity, model release, validation latency, endogenous information, equilibrium criticality, Volterra causality}
}
\usepackage[nameinlink,capitalize,noabbrev]{cleveref}
\usepackage[round,authoryear]{natbib}
\allowdisplaybreaks

\usepackage{aliascnt}
\newtheorem{theorem}{Theorem}[section]
\newaliascnt{proposition}{theorem}
\newtheorem{proposition}[proposition]{Proposition}
\aliascntresetthe{proposition}
\newaliascnt{lemma}{theorem}
\newtheorem{lemma}[lemma]{Lemma}
\aliascntresetthe{lemma}
\newaliascnt{corollary}{theorem}
\newtheorem{corollary}[corollary]{Corollary}
\aliascntresetthe{corollary}
\theoremstyle{definition}
\newaliascnt{definition}{theorem}
\newtheorem{definition}[definition]{Definition}
\aliascntresetthe{definition}
\newaliascnt{assumption}{theorem}
\newtheorem{assumption}[assumption]{Assumption}
\aliascntresetthe{assumption}
\newaliascnt{example}{theorem}

\aliascntresetthe{example}
\theoremstyle{remark}
\newaliascnt{remark}{theorem}

\aliascntresetthe{remark}

\crefname{theorem}{Theorem}{Theorems}
\Crefname{theorem}{Theorem}{Theorems}
\crefname{proposition}{Proposition}{Propositions}
\Crefname{proposition}{Proposition}{Propositions}
\crefname{lemma}{Lemma}{Lemmas}
\Crefname{lemma}{Lemma}{Lemmas}
\crefname{corollary}{Corollary}{Corollaries}
\Crefname{corollary}{Corollary}{Corollaries}
\crefname{definition}{Definition}{Definitions}
\Crefname{definition}{Definition}{Definitions}
\crefname{assumption}{Assumption}{Assumptions}
\Crefname{assumption}{Assumption}{Assumptions}
\crefname{example}{Example}{Examples}
\Crefname{example}{Example}{Examples}
\crefname{remark}{Remark}{Remarks}
\Crefname{remark}{Remark}{Remarks}

\newcommand{\R}{\mathbb{R}}
\newcommand{\E}{\mathbb{E}}

\newcommand{\Hh}{\mathcal{H}}

\newcommand{\CE}{\operatorname{CE}}
\newcommand{\Tax}{\mathcal{T}^{\mathrm{rec}}}
\newcommand{\argmax}{\operatorname*{arg\,max}}

\newcommand{\sgn}{\operatorname{sgn}}
\newcommand{\diag}{\operatorname{diag}}
\newcommand{\pc}{\mathrm{PC}}
\newcommand{\spe}{\mathrm{SPE}}
\newcommand{\fresh}{\mathrm{F}}
\newcommand{\inh}{\mathrm{I}}
\newcommand{\dd}{\,\mathrm{d}}

\title{\textbf{Bayesian Confidence Recalibration and Research-Equilibrium Criticality}\\[0.25em]
\large Temporal Support in Robust Portfolios}
\author{Han Yan\c{c}\\[0.35em]
\small School of Physical and Mathematical Sciences,\\
\small Nanyang Technological University, Singapore}
\date{}

\begin{document}
\maketitle
\vspace{-1.5em}

\begin{abstract}
Robust portfolio rules that reconstruct confidence sets after learning need not preserve the
evaluator obtained by prior-by-prior Bayesian transport. In the Gaussian model, this discrepancy
is summarized by natural-coordinate displacement: inherited transport preserves it whereas fresh
reconstruction can replace it. We price evaluator replacement and trace the resulting optimized
curvature through endogenous research.
Optimized robust value represents protocol regret as a functional Bregman divergence, while a
within-vintage rectangular Gaussian benchmark with constant absolute risk aversion (CARA) yields a stopped recalibration tax. In a
versioned model-release economy, validated history propagates through a strictly causal network
and a same-cycle share of current optimized marginal value feeds back into research supply. The
capacity-constrained equilibrium reduces to a scalar equation with protocol-indexed gain
\(\mathfrak g_I^P=\lambda\beta_{R,I}(W_R^P)''\). Purely causal validation cannot create a
same-cycle unit mode; provenance changes criticality through optimized curvature. For a
scalar primitive supplier-score shock \(z\) in direction \(h_I\) and financial outcome \(\mathcal O\), sensitivity factors as
\(\omega_{\mathcal O,h,I}/(1-\mathfrak g_I^P)\). Conditional on a smooth equilibrium
state and active cell, completion-time information sharply bounds this multiplier when all
compatible timing laws are subcritical; no finite uniform bound exists when the timing set
reaches the pole.
\end{abstract}

\noindent\textbf{Keywords:} confidence recalibration; Bayesian ambiguity; quant model release; validation latency; endogenous information; equilibrium criticality; Volterra causality.

\medskip
\noindent\textbf{JEL:} G11, D81, D83, C61.

\section{Introduction}\label{sec:intro}

A posterior confidence set is a statistical object; the ambiguity evaluator
governing a portfolio is an intertemporal preference object. The two coincide
only when the reconstructed set is the Bayesian continuation of the earlier
one. In the Gaussian model this distinction has a natural-coordinate
representation: Bayesian transport preserves the displacement of every
alternative, whereas fresh reconstruction generally substitutes a new
displacement. This paper asks one financial question: \emph{how is protocol
replacement priced, how does release timing transmit its value through
endogenous research, and when does the resulting criticality become visible in
financial outcomes?} Throughout, ``equilibrium'' refers to the
research-supply/information-acquisition fixed point rather than asset-price
market clearing.

The construction begins with the evaluator. A common Gaussian likelihood
preserves every natural displacement \(q\):
\begin{equation}\label{eq:intro-natural-transport}
 V_t^{-1}(m_t^q-m_t)=q.
\end{equation}
An inherited vintage is \(m_t+V_t\mathcal Q_s\); fresh reconstruction uses
\(m_t+V_t\mathcal Q_t\). If \(h,H\) are support
functions and
\(\mathfrak V(h)=\sup_{a\in\mathcal A}\{L(a)-\ell_{Ba}(h)\}\), the loss from
forcing evaluator \(h\) to use the action selected by \(H\) is
\begin{equation}\label{eq:intro-bregman-price}
 \boxed{\mathfrak r(h,H)=\mathfrak V(h)-\mathfrak V(H)
 +\ell_{Ba_H}(h-H)\ge0.}
\end{equation}
The divergence prices evaluator replacement. In the within-vintage rectangular Gaussian specialization with constant absolute risk aversion (CARA), the dynamic performance identity yields a stopped recalibration tax. These welfare results
are derived before specifying research supply; the release economy then takes
the resulting protocol-indexed optimized value \(W_R^P\) as its endogenous
feedback object.

A versioned quant-model release cycle supplies the financial mechanism.
Research becomes reusable only after validation, so completed earlier-stage
output enters through a strictly lower-triangular pipeline \(Q\). Current model
value can matter before the batch closes only through a same-cycle settlement
share \(\lambda\); with stochastic completion time this share becomes the
value-weighted statistic \(\lambda_\rho(\Delta)\) defined below. On a smooth
active cell \(I\), the causal pipeline resolves into the exposure
\(\beta_{R,I}\), and protocol \(P\) has same-cycle gain
\begin{equation}\label{eq:intro-release-critical-gain}
 \boxed{\mathfrak g_I^P(y)=
 \lambda\beta_{R,I}(W_R^P)''(y).}
\end{equation}
Causal validation is therefore nonsingular at \(\lambda=0\) but exposure
amplifying: holding supplier-cost coefficients and release loadings fixed, an entrywise increase in the
nonnegative validation network \(Q\) cannot reduce the full-interior \(\beta_R\). For a primitive score
shock \(z\) and differentiable release-dependent outcome \(\mathcal O\), the
local equilibrium response takes the form
\begin{equation}\label{eq:intro-financial-multiplier}
 \boxed{
 \frac{\dd\mathcal O}{\dd z}
 =\frac{\omega_{\mathcal O,h,I}(y)}{1-\mathfrak g_I^P(y)}.}
\end{equation}
The numerator is the open-loop financial loading; the denominator is the
distance to the protocol-indexed unit-gain boundary. Comparing fresh and
inherited protocols therefore amounts to comparing their optimized curvatures,
not to assigning a stability sign to provenance itself.

Temporal support explains why the denominator separates in this way. At a
smooth research equilibrium write the derivative of the cost-normalized
response as \(D\mathcal G=N+\mathcal V\), where \(\mathcal V\) is strictly
causal and \(N\) contains same-date or otherwise noncausal feedback. Under the
fractional--Volterra estimate used below, define the resolved noncausal operator
\(K:=(I-\mathcal V)^{-1}N\). Then
\begin{equation}\label{eq:intro-temporal-support-factorization}
 \boxed{D_E\mathcal F=(I-\mathcal V)(I-K),\qquad
 D_E\mathcal F\text{ singular}\iff1\in\sigma(K).}
\end{equation}
The release primitive is the scalar active-cell realization of this operator:
its only possibly nonzero resolved eigenvalue is \(\mathfrak g_I^P\). Thus
delayed validation can magnify a same-cycle channel through its causal
resolvent but cannot create a unit mode by itself. The nonlinear hierarchy is
strict: a statewise unit mode need not lie on a branch fold, a fold need not be
feasible under capacity constraints, and a feasible fold can be financially
silent when the relevant loading vanishes.

Three neighboring literatures define the comparison. First, robust portfolio
choice, confidence sets, Bayesian ambiguity, and costly information acquisition
study the evaluator under learning
\citep{GuanJiaLiang2024,GuanLiangXia2024,LiangLiuMa2026,LiangWangXia2025,LiangYe2024};
our first step instead tracks whether a reconstructed confidence set preserves
or replaces the inherited evaluator. Second, endogenous information and timing
can create complementarity or multiplicity
\citep{HellwigVeldkamp2009,MyattWallace2012,GanguliYang2009}; here temporal
support is tied to a versioned release protocol in which validated history and
same-cycle settlement enter separately. Third, nonlocal equilibrium and
feedback methods provide the operator language
\citep{LeiPun2024,WangYuZhangZhou2026,Wu2026FeedbackCycles}; the release economy
maps the resolved noncausal operator into protocol-indexed portfolio curvature
and a financially loaded sensitivity multiplier. Machine-learning portfolios,
implementability, factor vintages, post-publication return decay, and model
validation motivate the institutional setting
\citep{GuKellyXiu2020,JensenKellyMalamudPedersen2026,AkeyRobertsonSimutin2026,McLeanPontiff2016,CumminsEtAl2023,FederalReserve2026ModelRisk}.
The 2026 interagency guidance generally places validation before first use but
allows constrained use before completion under urgent business need.
\emph{The Price of Relearning} \citep{YancPriceRelearning2026} develops
ambiguity provenance, evaluator vintages, and their directed welfare consequences
in general dynamic decisions. The Gaussian natural-coordinate transport identity is
a shared specialization. The analysis here begins from its protocol-specific
optimized financial curvature and traces that curvature through a versioned
research--validation--release equilibrium, with validation timing determining
financial sensitivity.

Sections~\ref{sec:inheritance}--\ref{sec:model} identify evaluator replacement
and its welfare price. Section~\ref{sec:endogenous} isolates the role of temporal
support in endogenous research. Section~\ref{sec:model-release} develops the
versioned release economy, its protocol-indexed critical boundary, and the
financial multiplier. The Electronic Companion (EC) contains proofs, hard-response
and vanishing-regularization boundaries, nonlinear-cost and finite-factor
extensions, and the worked scalar benchmarks.

\section{Preference transport and statistical coverage}\label{sec:inheritance}

\subsection{Gaussian learning and the compatibility coordinate}\label{subsec:financial-setup}

The money-market account is the numeraire. One risky asset satisfies
\begin{equation}\label{eq:stock}
 \frac{\dd S_t}{S_t}=\theta\,\dd t+\sigma\,\dd W_t^S,
 \qquad \sigma>0,
\end{equation}
where the constant excess drift \(\theta\) is unknown. In the exogenous-information benchmark the investor also observes
\begin{equation}\label{eq:private-signal}
 \dd Y_t=\sqrt{\iota(t)}\,\theta\,\dd t+\dd W_t^Y,
\end{equation}
with deterministic \(\iota(t)\ge0\). For a Gaussian prior,
\begin{align}
 \dd m_t
 &=v_t\left(\sigma^{-1}\dd\widehat W_t^S+\sqrt{\iota(t)}\dd\widehat W_t^Y\right),
 \label{eq:filter-m}\\
 \dd v_t
 &=-v_t^2\bigl(\sigma^{-2}+\iota(t)\bigr)\dd t.
 \label{eq:filter-v}
\end{align}
A dollar risky position \(\pi_t\) generates
\begin{equation}\label{eq:wealth-nominal}
 \dd X_t=\pi_tm_t\,\dd t+\pi_t\sigma\,\dd\widehat W_t^S,
\end{equation}
and utility is CARA, \(U(x)=-e^{-\gamma x}\).  The filter in \eqref{eq:filter-m}--\eqref{eq:filter-v} is the linear--Gaussian specialization of standard nonlinear filtering; see \citet{BainCrisan2009}.  Portfolio choice with an unknown drift under partial information is classical; see \citet{Lakner1998}.

\begin{definition}[Bayesian inheritance and fresh recalibration]\label{def:protocols}
At date \(s\), write a common-covariance ambiguity set of posterior means as
\(\mathcal M_s=m_s+V_s\mathcal Q_s\). Its \emph{Bayesian inheritance} at \(t>s\) is obtained by updating every member through the same realized likelihood. A \emph{freshly recalibrated} set is reconstructed from the date-\(t\) posterior. The two protocols are dynamically compatible when the two sets agree after every common update.
\end{definition}

Let the unknown drift be \(d\)-dimensional and suppose the likelihood increment from \(s\) to \(t\) is
\begin{equation}\label{eq:gaussian-likelihood}
 L_{s,t}(\theta)\propto
 \exp\!\left\{\theta^\top b_{s,t}-\frac12\theta^\top A_{s,t}\theta\right\},
 \qquad A_{s,t}\succeq0.
\end{equation}
Then
\begin{equation}\label{eq:matrix-posterior}
 V_t^{-1}=V_s^{-1}+A_{s,t},
 \qquad
 V_t^{-1}m_t=V_s^{-1}m_s+b_{s,t}.
\end{equation}
For an alternative mean \(m_s^q=m_s+V_sq\), the following identity is pathwise in the realized sufficient statistic.

\begin{proposition}[Gaussian natural-coordinate invariance]\label{thm:matrix-inheritance}
Prior-by-prior Bayesian updating preserves the natural displacement:
\begin{align}
 \boxed{m_t^q=m_t+V_tq}
 &\label{eq:matrix-q-invariance}\\
 \boxed{V_t^{-1}(m_t^q-m_t)=V_s^{-1}(m_s^q-m_s)=q}
 &\label{eq:q-invariance}
\end{align}
Consequently,
\begin{equation}\label{eq:matrix-inherited-set}
 \mathcal M^I_{s\to t}=m_t+V_t\mathcal Q_s,
\end{equation}
so a fresh set \(m_t+V_t\mathcal Q_t\) is compatible with the earlier one exactly when \(\mathcal Q_t=\mathcal Q_s\).
\end{proposition}

\begin{proof}
The alternative prior has natural mean parameter \(V_s^{-1}m_s+q\). The common likelihood adds the same precision \(A_{s,t}\) and natural statistic \(b_{s,t}\) to both priors, so
\(V_t^{-1}m_t^q=V_t^{-1}m_t+q\). Multiplication by \(V_t\) proves the result.
\end{proof}

\subsection{What fresh confidence reconstruction changes}\label{subsec:fresh-finance-scope}

A nonempty compact convex body is uniquely determined by its support function
\citep{Schneider2013}. Equation~\eqref{eq:q-invariance} therefore reduces the
financial protocol choice to the natural-coordinate body: Bayesian inheritance
carries the earlier body $\mathcal Q_s$, whereas fresh reconstruction substitutes the
current body $\mathcal Q_t$. The remainder of the paper prices that replacement and
studies its propagation through research supply.

The distinction is already nontrivial for an ordinary Gaussian credible
interval.  In one dimension, with posterior precision $p_t=v_t^{-1}$, a
fixed-level fresh interval has natural-coordinate half-width
\begin{equation}\label{eq:main-fresh-natural-budget}
 k_t^{\fresh}=z_\alpha\sqrt{p_t},
\end{equation}
where $z_\alpha>0$ denotes the standard-normal critical value specified by the institution's confidence rule, whereas Bayesian inheritance preserves the earlier natural budget
$k_s$.  Thus nontrivial learning implies
\begin{equation}\label{eq:main-fresh-vs-inherited-budget}
 \boxed{k_t^{\fresh}\ne k_s^{\inh}}
\end{equation}
unless the confidence rule is deliberately changed to respect transport.
This is the primitive intervention priced in the next section.

\section{Welfare price of protocol replacement}\label{sec:model}

\subsection{Portfolio problem and global selector}

At a state $(m,v)$, a Gaussian alternative inherited from an earlier date has mean $m+vq$. Under its predictive law,
\begin{align}
\dd X_t&=\pi_t(m_t+v_tq_t)\,\dd t+\pi_t\sigma\,\dd W_t^{S,q},
\label{eq:wealth-q}\\
\dd m_t&=v_t^2I_tq_t\,\dd t+v_t\left(\sigma^{-1}\dd W_t^{S,q}+\sqrt{\iota(t)}\dd W_t^{Y,q}\right),
\label{eq:m-q}\\
\dd v_t&=-v_t^2I_t\,\dd t.
\label{eq:v-q}
\end{align}
For a single inherited prior, $q_t\equiv q$ by \cref{thm:matrix-inheritance}. To obtain a recursive value for one evaluator, we use a maxmin multiple-prior criterion \citep{GilboaSchmeidler1989} and impose a within-vintage rectangular distortion class, in the recursive-rectangularity sense of \citet{EpsteinSchneider2003}:
\[
\mathcal Q_k:=\{q:\ q\text{ is progressively measurable and }|q_u|\le k\text{ a.e.}\}.
\]
The particular progressively measurable drift-distortion class $\mathcal Q_k$ is our continuous-time model specialization; the cited multiple-priors work supplies the rectangularity principle rather than this Brownian implementation.
The dynamic implementation below is therefore within-vintage rectangular. The budget $k$ is frozen for that evaluator, while a future self acting at state $(m_u,v_u)$ uses the newly prescribed budget $\kappa(v_u)$. The generally nonrectangular prior-by-prior transported family is not rectangularized.

\begin{assumption}[Localized admissibility]\label{ass:admissibility}
Fix a bounded open state cylinder $\mathcal D\subset[0,T)\times\R\times(0,\infty)$ and localize at its first exit. Admissible portfolio controls are progressively measurable, locally bounded after localization, square integrable under every $q\in\mathcal Q_k$, and satisfy the uniform exponential-moment condition needed for the CARA criterion. Candidate feedbacks and bounded spike deviations are taken from this class. All equilibrium statements below are local to $\mathcal D$.
\end{assumption}

For an evaluator with frozen budget $k$, define the robust certainty equivalent of an admissible continuation control $\pi$ by
\begin{equation}\label{eq:frozen-ce-definition}
\mathfrak J_t^k(\pi)
:=-\frac1\gamma\log
\sup_{q\in\mathcal Q_k}
\E_t^q\!\left[
\exp\{-\gamma(X_T^\pi-X_t)\}
\right],
\end{equation}
with the localized continuation understood as in \cref{ass:admissibility}.

\begin{definition}[Current-vintage commitment]\label{def:commitment}
At state $(t,m,v)$, the current evaluator freezes $k=\kappa(v)$ and chooses the entire admissible continuation control to maximize \eqref{eq:frozen-ce-definition}.  A maximizer is denoted $\pi^{\pc,k}$ and the optimized commitment certainty equivalent is
\[
 F^k(t,m,v):=\sup_{\pi}\mathfrak J_t^k(\pi)
 =\mathfrak J_t^k(\pi^{\pc,k}).
\]
The superscript $\pc$ refers to commitment of the control and of the evaluator's ambiguity budget, not to a different statistical experiment.  We reserve $f^k$ below for the evaluator-$k$ policy value generated by following the future sophisticated feedback; $F^k$ and $f^k$ are therefore not identified off the current diagonal.
\end{definition}

\begin{definition}[Sophisticated recalibration equilibrium (SPE)]\label{def:spe}
Write $\beta^*=\pi^*$ in the exogenous-information model and
$\beta^*=(\pi^*,e^*)$ when research is controlled.  For $h>0$, let
$\eta=(\eta^\pi,\eta^e)$ be a bounded admissible observation-feedback spike on $[t,t+h)$, with $\eta^e\in[0,\bar e]$; in the exogenous model the second component is omitted.  Let $\beta^{h,\eta}$ use $\eta$ on $[t,t+h)$ and $\beta^*$ thereafter.  The feedback $\beta^*$ is a weak sophisticated equilibrium if, at every interior state and for the current budget $k=\kappa(v)$,
\begin{equation}\label{eq:spike-definition}
\limsup_{h\downarrow0}
\frac{\mathfrak J_t^k(\beta^{h,\eta})-\mathfrak J_t^k(\beta^*)}{h}
\le0
\qquad\text{for every bounded admissible spike }\eta.
\end{equation}
Here $\mathfrak J_t^k(\pi)$ is \eqref{eq:frozen-ce-definition} in the exogenous model; in the endogenous model it is the same robust certainty equivalent with the wealth drift reduced by $c(e)$ and with the controlled observation law described in \cref{subsec:controlled-filter} below.  A family $f^k$ is the associated frozen-vintage continuation family when
$V^k(t,x,m,v)=-\exp\{-\gamma[x+f^k(t,m,v)]\}$
represents the evaluator-$k$ utility from following the future equilibrium feedback.
\end{definition}

For a smooth continuation $f$, the CARA transform gives
\begin{align}
\Hh(k;\pi;f)
={}&\pi m-v^2I_tf_v+\frac12v^2I_tf_{mm}
-\frac\gamma2\left[\sigma^2\pi^2+2v\pi f_m+v^2I_tf_m^2\right]
\nonumber\\
&-kv\left|\pi+vI_tf_m\right|.
\label{eq:Hamiltonian}
\end{align}
The equilibrium PDE system corresponding to \cref{def:spe} is
\begin{align}
\partial_tf^k+\Hh(k;\pi^*;f^k)&=0,
&f^k(T,m,v)&=0,
\label{eq:frozen-pde}\\
\pi^*(t,m,v)&\in\argmax_\pi\Hh\!\left(\kappa(v);\pi;f^{\kappa(v)}\right).
\label{eq:diag-selector}
\end{align}
The derivative in \eqref{eq:diag-selector} is taken with the current evaluator's $k$ frozen. It is not the total derivative of the diagonal map $v\mapsto f^{\kappa(v)}$.

The Electronic Companion gives the stopped-spike verification showing that,
under its stated local regularity conditions, the displayed diagonal and
joint selectors coincide with the selectors of the dynamic spike problem.

\paragraph{Current selector.}
For the stock-only/no-extra-signal benchmark, set $\iota\equiv0$ and
suppress endogenous observation effort, $e\equiv0$, so
$I_t=\sigma^{-2}$.  Completing the robust quadratic then gives the exact
soft-threshold exposure.  For $a\ge0$, write
$\mathsf S_a(x):=\operatorname{sgn}(x)(|x|-a)_+$.  With the current diagonal
budget $k^\circ:=\kappa(v)$,
\[
 \pi^*+\frac v{\sigma^2}f_m^{k^\circ}
 =\frac{\mathsf S_{k^\circ v}(m)}{\gamma\sigma^2}.
\]
Thus the switching surfaces are \(m=\pm\kappa(v)v\).  This formula is used only
to interpret active faces; its coefficient derivation is given in the Electronic
Companion.

The local curvature of protocol regret is useful for asymptotics but is not
needed for the global comparison.  We proceed directly to the value-envelope
identity, which remains exact at active-set changes.

\subsection{Protocol regret beyond scalar ambiguity budgets}\label{subsec:full-shape-regret}

The scalar ambiguity budget is a ray inside a support-function space.  Let
\(\mathbb E=C(S^{d-1})\) with the supremum norm, let \(\mathscr H_d\subset
\mathbb E\) be the cone of support functions of nonempty compact convex
subsets of \(\R^d\), and define the homogeneous evaluation functional
\begin{equation}\label{eq:homogeneous-evaluation-main}
 \ell_z(h)=
 \begin{cases}
  \|z\|h(z/\|z\|),&z\ne0,\\
  0,&z=0.
 \end{cases}
\end{equation}
For a nonempty compact convex set $Q\subset\mathbb R^d$, write its support function as
$h_Q(u):=\sup_{q\in Q}q^\top u$ \citep{Schneider2013}.  Thus
\(\ell_z(h_Q)=h_Q(z)\).  In the dual measure representation,
\(\ell_z=\|z\|\delta_{z/\|z\|}\) for \(z\ne0\); the weight \(\|z\|\) is
part of the functional and cannot be suppressed.  Let
\(\mathcal A\subset\R^n\) be compact and
convex, let \(B:\R^n\to\R^d\) be linear, and let \(L\) be continuous and
strongly concave.  The ambient extension
\begin{equation}\label{eq:shape-value-functional-main}
 \mathfrak V(h)=\max_{a\in\mathcal A}\{L(a)-\ell_{Ba}(h)\},
 \qquad h\in\mathbb E,
\end{equation}
is useful because it permits ordinary Banach-space subgradients.  When
\(h\in\mathscr H_d\), the objective is strongly concave and its maximizer
\(a_h\) is unique.

\begin{theorem}[Protocol regret and Bregman representation]\label{thm:global-budget-regret}
Put \(R_B=\sup_{a\in\mathcal A}\|Ba\|\).  The functional
\(\mathfrak V:\mathbb E\to\R\) is convex and \(R_B\)-Lipschitz.  At every
\(H\in\mathbb E\) with a unique optimizer,
\begin{equation}\label{eq:shape-subgradient-main}
 \xi_H:=-\ell_{Ba_H}\in\partial\mathfrak V(H),
 \qquad
 \mathfrak V'(H;g)=-\ell_{Ba_H}(g),
\end{equation}
where the directional derivative is Hadamard.  For support functions
\(h,H\in\mathscr H_d\), the welfare loss when evaluator \(h\) is forced to
use the action selected by \(H\) is the generalized functional
Bregman divergence associated with \(\mathfrak V\) \citep{Bregman1967}
\begin{equation}\label{eq:functional-bregman-regret-main}
 \boxed{
 \mathfrak r(h,H)
 :=\mathfrak V(h)-[L(a_H)-\ell_{Ba_H}(h)]
 =\mathfrak V(h)-\mathfrak V(H)-\langle\xi_H,h-H\rangle
 \ge0.}
\end{equation}
It vanishes if and only if \(a_H\) is also optimal under \(h\), hence exactly
when \(a_H=a_h\) under strong concavity.

Localized noncompact extensions, the quadratic-finance face decomposition, and the
smooth-stratum differential geometry are collected in the Electronic Companion.
\end{theorem}

The scalar ambiguity-budget model is the ray restriction of this representation.  For
$h=kh_Q$ and $H=\bar k h_Q$, define
\begin{equation}\label{eq:budget-value-envelope}
 \varphi_Q(k)=\mathfrak V(kh_Q).
\end{equation}
Then
\begin{equation}\label{eq:global-regret-bregman}
 \boxed{\mathfrak r(k,\bar k)=D_{\varphi_Q}(k,\bar k).}
\end{equation}
The Electronic Companion gives the corresponding action/face decomposition, the
smooth-stratum Hessian and its decision-visible null space, and the scalar
soft-threshold formulas used at active-face changes.

\subsection{Dynamic performance identity and recalibration tax}\label{sec:exogenous}

Let \(\iota(t)\) be bounded and deterministic and let \(v_t\) solve
\eqref{eq:filter-v}.  Put
\begin{equation}\label{eq:alpha-def}
 \tau=T-t,\qquad D(t,v)=\sigma^2+\tau v,\qquad
 \alpha(t,v)=\sigma^2/D(t,v).
\end{equation}

\paragraph{Gaussian--CARA coefficient closure.}\hypertarget{thm:no-hedge}{}
On any stopped region where the relevant evaluators remain on the positive
common exposed face, the sophisticated auxiliary value \(f^k\) has
\begin{equation}\label{eq:quadratic-f-v3}
 f^k(t,m)=A(t)m^2+k\beta(t)m+C^k(t),
\end{equation}
where
\begin{equation}\label{eq:universal-coeffs}
 \boxed{A(t)=\frac{T-t}{2\gamma[\sigma^2+(T-t)v_t]},\qquad
 1+\gamma\beta(t)=\alpha(t,v_t).}
\end{equation}
The optimized commitment value \(F^k\) has the same \(m\)-dependent
coefficients \(A\) and \(k\beta\) on this face, although its constant term is
generally different.  Hence \(F_m^k=f_m^k\), and at a fresh state
\begin{equation}\label{eq:pi-universal}
 \boxed{\pi_t^{\spe}=\frac{\alpha(t,v_t)}{\gamma\sigma^2}
 [m_t-\kappa(v_t)v_t]=\pi_t^{\pc}.}
\end{equation}
Moreover,
\begin{equation}\label{eq:fresh-sign-general-e}
 \sgn(\pi_t^{\spe}+v_tI_tf_m^{\kappa(v_t)})
 =\sgn[m_t-\kappa(v_t)v_t].
\end{equation}
Thus recalibration changes continuation welfare but not the current interior
risky position while the common face is maintained.

\paragraph{Common-face finance specialization.}\hypertarget{cor:cara-protocol-regret}{}
For evaluator budget \(k\) and selector budget \(\bar k\),
\begin{equation}\label{eq:portfolio-budget-response}
 \pi_u^{\bar k}-\pi_u^k
 =-\frac{\alpha(u,v_u)v_u}{\gamma\sigma^2}(\bar k-k),
\end{equation}
and
\begin{equation}\label{eq:instantaneous-cara-regret}
 \boxed{\mathfrak r_u(k,\bar k)
 =\frac{\gamma\sigma^2}{2}(\pi_u^{\bar k}-\pi_u^k)^2
 =\frac{\alpha(u,v_u)^2v_u^2}{2\gamma\sigma^2}(\bar k-k)^2.}
\end{equation}
The corresponding \(n\)-asset fixed-face identity is recorded in the EC; it
has the same quadratic protocol-gap structure on the free portfolio
coordinates.

For any bounded admissible feedback \(\widehat\pi\), let
\(g^{k,\widehat\pi}(t,m,v):=\mathfrak J_t^k(\widehat\pi)\) be the evaluator-\(k\)
certainty-equivalent policy value and define its dynamic performance gap
\begin{equation}\label{eq:dynamic-performance-gap-definition}
 \mathcal P_t^k[\widehat\pi]
 :=F^k(t,m,v)-g^{k,\widehat\pi}(t,m,v).
\end{equation}
The recalibration tax is the specialization to the current vintage and the
future sophisticated policy,
\begin{equation}\label{eq:recalibration-tax-definition}
 \boxed{\Tax_t=\mathcal P_t^{k_t}[\pi^{\spe}]
 =F^{k_t}(t,m_t,v_t)-g^{k_t,\pi^{\spe}}(t,m_t,v_t),
 \qquad k_t=\kappa(v_t).}
\end{equation}
Both terms are evaluated by the current vintage. Within-vintage
rectangularity isolates evaluator replacement from the distinct
commitment-versus-sophistication wedge of a nonrectangular prior family. At a
fixed state, minimizing over one \(q\in[-k_t,k_t]\) and using the pointwise
selector of \(\mathcal Q_{k_t}\) generate the same Hamiltonian support term
\(-k_tv\lvert\pi+vI_tF_m^{k_t}\rvert\); their difference is intertemporal.
Thus \(\Tax_t\) is the dynamic price of replacement in the
recursive-rectangular Gaussian implementation indexed by the statistical
natural budget \(k_t\). The Electronic Companion records the one-step
coincidence.

\begin{theorem}[Dynamic performance and common-face recalibration identities]
\label{thm:global-performance-identity}\label{cor:exact-por}
Let $F^k$ be the optimized current-vintage commitment certainty equivalent and
let $g^{k,\widehat\pi}$ be the evaluator-$k$ value of a bounded admissible
feedback $\widehat\pi$.  On every bounded localization cylinder assume that
both functions are classical $C^{1,2,1}$ solutions in $(t,m,v)$, with the
derivatives entering the Hamiltonian continuous, and that the comparison SDE
specified in the Electronic Companion is well posed and its stopped local
martingale is a true martingale.  Define
\begin{equation}\label{eq:hamiltonian-gap-definition-main}
 \mathfrak R^{k,\widehat\pi}
 :=\mathcal H(k;\pi^{\pc,k};F^k)
   -\mathcal H(k;\widehat\pi;F^k)\ge0,
\end{equation}
and put $\Delta=F^k-g^{k,\widehat\pi}$.

\textup{(i) Global identity.}
If an exhausting localization $\tau_n\uparrow T$ has a uniformly integrable
boundary family and the two evaluations share the terminal condition, then
\begin{equation}\label{eq:global-performance-FK}
 \boxed{F^k-g^{k,\widehat\pi}
 =\mathbb E_{t,m,v}^{k,\widehat\pi}\!\left[
 \int_t^T\mathfrak R_s^{k,\widehat\pi}\,\dd s\right]\ge0.}
\end{equation}
The growth conditions used to guarantee these hypotheses are recorded in the
Electronic Companion.

\textup{(ii) Stopped identity and switching source.}
For every admissible localization time $\tau$, the comparison law satisfies
\begin{equation}\label{eq:stopped-performance-FK}
 \boxed{\Delta(t,m,v)
 =\mathbb E_{t,m,v}^{k,\widehat\pi}\!\left[
 \Delta(\tau,M_\tau,V_\tau)
 +\int_t^\tau\mathfrak R_s^{k,\widehat\pi}\,\dd s\right].}
\end{equation}
Thus stopping at an exit boundary never removes its continuation value.  With
\[
 x_F=m+\gamma v(\sigma^2I_t-1)F_m^k,\qquad
 y_k=\mathsf S_{kv}(x_F)/(\gamma\sigma^2),\qquad
 \widehat y=\widehat\pi+vI_tF_m^k,
\]
and $\zeta_k\in\partial|y_k|$ chosen at the optimizer,
\begin{equation}\label{eq:dynamic-face-switching-source}
 \boxed{\mathfrak R^{k,\widehat\pi}
 =\frac{\gamma\sigma^2}{2}(\widehat y-y_k)^2
 +kv[|\widehat y|-|y_k|-\zeta_k(\widehat y-y_k)].}
\end{equation}

\textup{(iii) Common-face recalibration tax.}
Fix $k_t=\kappa(v_t)$, set $\widehat\pi=\pi^{\spe}$, and let
$\tau_{\mathrm{cf}}$ be the first exit from the common positive face, capped
at $T$.  Then
\begin{align}
 \Tax_t
 =\mathbb E_{t,m_t,v_t}^{k_t,\pi^{\spe}}\!\Bigg[
 &\Delta(\tau_{\mathrm{cf}},M_{\tau_{\mathrm{cf}}},
 V_{\tau_{\mathrm{cf}}}) \nonumber\\
 &+\int_t^{\tau_{\mathrm{cf}}}
 \frac{\alpha(u,v_u)^2v_u^2}{2\gamma\sigma^2}
 [\kappa(v_u)-\kappa(v_t)]^2\,\dd u\Bigg].
 \label{eq:stopped-common-face-tax}
\end{align}
If the boundary gap vanishes, only the first term disappears: the source is
still integrated to $\tau_{\mathrm{cf}}$.  The common-face source can be
integrated to $T$, and the closed form below can be used, only under the
no-exit condition $\tau_{\mathrm{cf}}=T$ almost surely. Accordingly, the
following formula is a facewise no-exit benchmark. In that case
\begin{equation}\label{eq:exact-tax-v3}
 \boxed{\Tax_t=\int_t^T
 \frac{\alpha(u,v_u)^2v_u^2}{2\gamma\sigma^2}
 [\kappa(v_u)-\kappa(v_t)]^2\,\dd u,}
\end{equation}
which is nonnegative and vanishes exactly when the natural ambiguity budget is
constant almost everywhere on $[t,T]$.  For the stock-only fixed-level credible rule
$\kappa(v)=z_\alpha/\sqrt v=z_\alpha\sqrt p$, the EC evaluates this no-exit
integral in closed form.
\end{theorem}

\section{Temporal support in endogenous research}\label{sec:endogenous}

The preceding section compresses protocol replacement into optimized financial
value. We now ask which temporal components of research supply transmit that curvature
into equilibrium feedback and when the critical mode lies on a fold. Endogenous
research can also alter estimation or predictive risk, so its total value need not be
a literal support-only composition. Let
$E$ denote the research state and write the optimized financial value as
\begin{equation}\label{eq:parameterized-optimized-envelope-main}
 W(E):=\max_{a\in\mathcal A}J(a,E).
\end{equation}

\begin{theorem}[Optimized financial value and temporal decomposition]
\label{thm:welfare-curvature-transmission}
\textup{(i) Optimized-value curvature.}
Let $X$ be a real Hilbert space, $U\subset X$ open,
$\mathcal A\subset\mathbb R^n$ open, and
$J:\mathcal A\times U\to\mathbb R$ be $C^2$. Suppose there are a
neighborhood $U_0$ of $E_*$ and a relatively compact open
$A_0\Subset\mathcal A$ such that, for every $E\in U_0$, the unique global
maximizer $a(E)$ is interior and lies in $A_0$. If
$D^2_{aa}J(a(E_*),E_*)$ is negative definite, then, after shrinking $U_0$,
$a(\cdot)$ is $C^1$ and $W(E)=J(a(E),E)$ is twice Fr\'echet differentiable.
For $u,v\in X$,
\begin{align}
 DW(E)[u]&=D_EJ(a(E),E)[u],
 \label{eq:parameterized-envelope-first-main}\\
 D^2W(E)[u,v]
 &=D^2_{EE}J[u,v]
 -D^2_{Ea}J[u]\,[D^2_{aa}J]^{-1}D^2_{aE}J[v],
 \label{eq:parameterized-envelope-second-main}
\end{align}
with the derivatives on the second line evaluated at $(a(E),E)$. On a smooth
support-function stratum, if $h:U\to\mathbb E$ is $C^2$ and
$J(a,E)=L(a)-\ell_{Ba}(h(E))$, then
\begin{align}
 D^2W(E)[u,v]
 ={}&D^2\mathfrak V(h(E))[Dh(E)u,Dh(E)v]\nonumber\\
 &+D\mathfrak V(h(E))[D^2h(E)[u,v]].
 \label{eq:welfare-curvature-transmission}
\end{align}

\textup{(ii) Temporal decomposition on $L^\infty$.}
Let $\mathbb X=L^\infty(0,S;\mathbb R^d)$ and write the datewise marginal
research value as
\begin{equation}\label{eq:datewise-marginal-research-value-main}
 \Gamma(E)(s)=D_eW_s\!\left(E(s);\Xi(E)(s)\right).
\end{equation}
Suppose $\Gamma$ is Fr\'echet differentiable at an equilibrium $E_*$ and
\begin{equation}\label{eq:datewise-marginal-value-derivative-main}
 D\Gamma(E_*)=H_*^{\mathrm{sd}}+\mathcal L_*,
\end{equation}
where $H_*^{\mathrm{sd}}$ is bounded pointwise multiplication and
$\mathcal L_*$ is bounded and strictly causal. The Electronic Companion gives
primitive Carath\'eodory--Nemytskii conditions under which
\eqref{eq:datewise-marginal-value-derivative-main} follows from a $C^1$
continuation-state map $\Xi$.

Let $c_s$ be $C^2$ and suppose the pointwise inverse marginal-cost map defines
a $C^1$ response $\mathcal G$ near $E_*$, with $\mathcal G(E_*)=E_*$. If
$C_*(s)=D^2c_s(E_*(s))$ is uniformly positive definite, then
\begin{equation}\label{eq:financial-response-decomposition-main}
 \boxed{
 D\mathcal G(E_*)=M_*^{\mathrm{sd}}+\mathcal V_*,
 \qquad
 M_*^{\mathrm{sd}}=C_*^{-1}H_*^{\mathrm{sd}},
 \qquad
 \mathcal V_*=C_*^{-1}\mathcal L_*.}
\end{equation}
If $\mathcal L_*$ obeys the fractional Volterra bound below, so does
$\mathcal V_*$. Hence
\begin{equation}\label{eq:financial-master-operator-main}
 \boxed{K_*=(I-\mathcal V_*)^{-1}M_*^{\mathrm{sd}}.}
\end{equation}
For a static same-date problem $\mathcal L_*=0$. In the scalar timing family,
the fixed-state critical equality is
\begin{equation}\label{eq:scalar-curvature-balance}
 \boxed{\lambda W''(e)=c''(e).}
\end{equation}
\end{theorem}

In the Gaussian timing benchmark below, \(p=e+1/2\) enters both the ambiguity
term and the quadratic
predictive-risk coefficient, so \(H_*^{\mathrm{sd}}\) is the same-date
curvature of the full optimized financial value rather than ambiguity
curvature alone.  The support-space Bregman result remains exact for evaluator
replacement with the remaining financial primitives fixed.

Across timing regimes the relevant change is the temporal support of marginal
information value.  Compatible governance admits an exact precision-time
reparametrization and a hard current selector.  Fresh recalibration instead
requires research supply to reproduce the continuation value it helps create,
which is the fixed point analyzed below.

\subsection{Controlled information, joint selection, and the precision clock}
\label{subsec:controlled-filter}\label{subsec:compatible-hard-control}

Research is a bounded predictable observation functional,
\[
 e_t=\varepsilon(t,R_{\cdot\wedge t},Y_{\cdot\wedge t})\in[0,\bar e],
 \qquad \dd Y_t=\sqrt{e_t}\,\theta\,\dd t+\dd W_t^Y.
\]
A Gaussian prior remains Gaussian and its precision satisfies
\begin{equation}\label{eq:controlled-precision}
 p_t=v_0^{-1}+t/\sigma^2+\int_0^te_s\,\dd s.
\end{equation}
The likelihood derivation, inherited-distortion law, and posterior-mean filter
are in the Electronic Companion; predictability is imposed in the canonical
observation filtration rather than by treating research as an exogenous
schedule.

\begin{assumption}[Uniformly convex research cost]\label{ass:strong-research-cost}
The cost \(c:[0,\bar e]\to\mathbb R_+\) is \(C^3\), increasing, satisfies
\(c(0)=0\), and \(c''(e)\ge\chi_0>0\).
\end{assumption}
For a frozen evaluator, joint concavity gives a unique projected research KKT
rule.  Precision is also an exact control clock: with inverse time \(t(p)\),
\begin{equation}\label{eq:precision-time-calendar}
 \dd\bar t_p=(\sigma^{-2}+\bar e_p)^{-1}\dd p,
 \qquad \dd\langle\bar m\rangle_p=p^{-2}\dd p.
\end{equation}
Thus research changes learning speed but not posterior-mean quadratic variation
per unit precision.  The same EC derivation records the full Hamiltonian, KKT
system, reachable domain, and precision-time Bellman equation.

\subsection{Fresh recalibration and the causal--noncausal separator}\label{subsec:fresh-bifurcation}

For the stopped Gaussian responses satisfying the EC differentiability and
kernel hypotheses, current research enters only later continuation
coefficients, so the derivative is strictly Volterra. Within that class pure
causal feedback cannot generate a smooth saddle-node even when a horizon-wide
small-gain inequality fails. The separator below isolates this temporal-support
argument from any global contraction requirement; the EC gives the corresponding
vanishing-regularization transfer theorem.

On $L^\infty$, ``strictly causal'' is understood modulo null sets: for almost
every date $s$, the output up to $s$ depends only on the input history before
$s$.  Every quantitative use below is through the displayed fractional
Volterra bound, which is representative-independent and makes the temporal
support precise. For spectral statements, real response spaces and operators
are understood through their standard complexifications.

\begin{theorem}[Causal--noncausal equilibrium separator]
\label{thm:causal-noncausal-separator}
Let $H$ be a Banach space, $\mathbb X=L^\infty(0,S;H)$, and let
$\mathcal G_\vartheta:\mathbb X\to\mathbb X$ have residual
$\mathcal F(E,\vartheta)=E-\mathcal G_\vartheta(E)$.  Put
$\eta=1-\alpha$ for $0\le\alpha<1$.

\textup{(i) Nonlinear causal uniqueness with a contractive diagonal.}
Suppose that, for some $0\le\mu<1$,
\begin{equation}\label{eq:abstract-causal-lipschitz-main}
 \|\mathcal G_\vartheta(E_1)(s)-\mathcal G_\vartheta(E_2)(s)\|_H
 \le \mu\|E_1(s)-E_2(s)\|_H
 +C\int_0^s(s-r)^{-\alpha}
 \|E_1(r)-E_2(r)\|_H\,\dd r.
\end{equation}
Then $\mathcal G_\vartheta$ has at most one fixed point.  Pure strict
causality is the case $\mu=0$; no horizon-wide small-gain restriction on $C$
is needed.

\textup{(ii) Smooth causal--noncausal factorization.}
At an equilibrium suppose
\begin{equation}\label{eq:noncausal-volterra-decomposition-main}
 D_E\mathcal G_\vartheta(E)=N+\mathcal V,
\end{equation}
where $N$ is bounded and $\mathcal V$ satisfies the linear fractional-Volterra
bound obtained from part~\textup{(i)} with $\mu=0$.  Then
\begin{equation}\label{eq:abstract-volterra-power-main}
 \|\mathcal V^n\|
 \le\frac{[C\Gamma(\eta)S^\eta]^n}{\Gamma(n\eta+1)},
 \qquad r(\mathcal V)=0.
\end{equation}
Thus $R=(I-\mathcal V)^{-1}$ exists and
\begin{equation}\label{eq:causal-noncausal-factorization-main}
 \boxed{D_E\mathcal F
 =(I-\mathcal V)[I-\mathcal K],
 \qquad \mathcal K=RN.}
\end{equation}
Consequently
\begin{equation}\label{eq:separator-spectrum-main}
 D_E\mathcal F\text{ is singular}
 \quad\Longleftrightarrow\quad 1\in\sigma(\mathcal K).
\end{equation}
Moreover,
\begin{equation}\label{eq:mittag-leffler-resolvent-main}
 \|R\|\le E_\eta(C\Gamma(\eta)S^\eta),
 \qquad
 E_\eta(x)=\sum_{n=0}^\infty\frac{x^n}{\Gamma(n\eta+1)},
\end{equation}
so $\|N\|E_\eta(C\Gamma(\eta)S^\eta)<1$ is a sufficient
causal-adjusted nonsingularity condition. In the local $C^1$ setting, the
implicit-function theorem then excludes a smooth local bifurcation. In
particular, when $N=0$ and the stated Volterra bound holds, the classical
linearization is invertible; in the local $C^1$ setting no smooth local
bifurcation occurs.

The pure same-date finite-grid spectrum boundary, its
finite-dimensional equality case, and an infinite-dimensional counterexample
are stated and proved in the Electronic Companion.

\textup{(iii) Finite-rank causal compression.}
If $N=U_0L$ has rank at most $r$, with bounded $U_0:\mathbb C^r\to\mathbb X$
and $L:\mathbb X\to\mathbb C^r$, define
$B:=L(I-\mathcal V)^{-1}U_0\in\mathbb C^{r\times r}$. Then
\begin{align}
 I-\lambda K\text{ is invertible}
 &\quad\Longleftrightarrow\quad I_r-\lambda B\text{ is invertible},
 \label{eq:finite-rank-invertibility-main}\\
 \sigma(K)\setminus\{0\}&=\sigma(B)\setminus\{0\},\qquad
 1\in\sigma(K)\Longleftrightarrow\det(I_r-B)=0.
 \label{eq:finite-rank-spectrum-main}
\end{align}
Thus strictly causal history is resolved before the noncausal factor dimension
is compressed.
\end{theorem}

The logical hierarchy used below is
\begin{equation}\label{eq:criticality-hierarchy-main}
 \text{statewise unit mode}
 \;\not\Rightarrow\;
 \text{branch criticality}
 \;\not\Rightarrow\;
 \text{feasible fold}
 \;\not\Rightarrow\;
 \text{financial materiality}.
\end{equation}
\paragraph{Statewise criticality and fold certification.}
For a reference equilibrium that persists along a linear same-cycle path
\(D_E\mathcal G=\mathcal V+\lambda N\), put
\(R=(I-\mathcal V)^{-1}\) and \(K=RN\). If \(K\) is compact with an algebraically simple positive eigenvalue
\(\rho=r(K)\), the statewise
spectral threshold is \(\lambda_{\mathrm{spec}}=\rho^{-1}\). Along the \(C^1\) primitive-parameter branches specified in the EC, the simple
critical mode has a nonzero inverse-linear sensitivity residue precisely when the
limiting primitive perturbation loads on the left critical mode; a financial
observable inherits that leading term only when it also loads on the right critical
mode. A branch saddle-node requires
more: at an equilibrium with one-dimensional kernel and cokernel and Fredholm
index zero (automatic here for identity minus compact after resolving the
Volterra factor), a scalar unfolding parameter must satisfy the transversality condition
\(\langle\psi,\mathcal F_\vartheta\rangle\ne0\), and the full quadratic
coefficient
\(\frac12\langle\psi,D^2_{EE}\mathcal F[\phi,\phi]\rangle\) must be nonzero.
The Electronic Companion gives the corresponding simple-mode resolvent expansion
and Lyapunov--Schmidt reduction.

\paragraph{Finite-factor financial realization.}
For \(Nx=\sum_{j=1}^r u_j\ell_j(x)\), the compression in
\cref{thm:causal-noncausal-separator} has entries
\begin{equation}\label{eq:finite-factor-matrix}
 \boxed{B_{ij}=\ell_i((I-\mathcal V)^{-1}u_j).}
\end{equation}
It preserves the entire causal history and compresses only the noncausal factor
dimension.  The Electronic Companion gives a two-date mechanism check and the
precise finite-grid boundary.

For a nonempty closed convex set $C$ and $x\in C$, use the convex-analytic normal cone
$N_C(x):=\{\nu:\langle\nu,z-x\rangle\le0\text{ for every }z\in C\}$.

\paragraph{Hard-causal boundary.}
The abstract separator has a constructive counterpart.  On a reverse-time
finite grid, the componentwise projected response is a unique triangular
recursion with identity diagonal in every differentiable residual.  In
continuous time, a Lipschitz selector composed with a stopped response obeying
the stated fractional--Volterra estimate on a closed invariant class has a
unique fixed point reached by Picard iteration; when the two maps are
Fr\'echet differentiable and the selector derivative is bounded and
pointwise, the residual derivative is invertible.  The Electronic Companion
states the Banach-lattice domain, proves the Mittag--Leffler estimate, and gives
stopped, regularized, and local analytic certificates. This is the pure-causal baseline
for the same-cycle settlement channel introduced next.

\section{Versioned model release and financial criticality}
\label{sec:model-release}\label{subsec:model-release}

A released model version carries the confidence evaluator constructed in
\cref{sec:inheritance,sec:model}. For a protocol
\(P\in\{\fresh,\inh\}\), precision \(y\) indexes a posterior covariance
\(\Sigma(y)\), a protocol-specific evaluator \(\mathcal M^P(y)\), and optimized
portfolio value
\begin{equation}\label{eq:model-release-value-bridge}
 W_R^P(y):=\max_a J\!\left(a;\mathcal M^P(y),\Sigma(y)\right).
\end{equation}
The reduction below is stated for a generic \(W_R\); setting
\(W_R=W_R^P\) gives the corresponding protocol. The Gaussian fresh and
inherited specializations are compared below. Staged effort produces validated
information capital; settlement converts its marginal portfolio value into suppliers'
private return; latency determines how much of that return is available within the
same decision cycle.

Consider one quant-model release cycle with stages
\(n=0,\ldots,J-1\).  A unit mass of atomless, price-taking research suppliers
chooses \(x_{i,n}\in[0,\bar e_n]\), with aggregate
\(e_n=\int_0^1x_{i,n}\,\dd i\).  Validated research capital evolves as
\begin{equation}\label{eq:model-release-capital}
 q_{n+1}=\delta_nq_n+\kappa_ne_n,
 \qquad q_0\ge0,\quad 0\le\delta_n\le1,\quad \kappa_n\ge0.
\end{equation}
At stage \(n\), inherited validated capital raises marginal productivity by
\(a_n^{\mathrm{val}}q_n\), where \(a_n^{\mathrm{val}}\ge0\).  Let
\(\ell_n\ge0\) be stage \(n\)'s contribution to released precision,
\begin{equation}\label{eq:model-release-precision}
 y=\ell^\top e,\qquad e=(e_0,\ldots,e_{J-1})^\top.
\end{equation}
Before the research batch closes, a fraction \(\lambda\in[0,1]\) of the marginal
release-value settlement uses the current price \(W_R'(y)\); the remaining fraction
uses a predetermined inherited shadow price \(\bar\Gamma\ge0\). The derivative
\(W_R'(y)\) is the competitive marginal value of an infinitesimal contribution
to released precision; atomless suppliers take this price as given. After absorbing the
nonnegative predetermined contribution of \(q_0\) into
\(\widetilde b^0\in\mathbb R_+^J\), supplier \(i\) at stage \(n\) solves
\begin{align}\label{eq:model-release-supplier-payoff}
 \max_{0\le x\le\bar e_n}\quad
 x\Big[\widetilde b_n^0-\vartheta d_n+a_n^{\mathrm{val}}
       (q_n-q_n^{(0)})
 +\ell_n\{(1-\lambda)\bar\Gamma+\lambda W_R'(y)\}\Big]
 -\frac{c_n}{2}x^2,
\end{align}
where \(q_n^{(0)}=(\prod_{j=0}^{n-1}\delta_j)q_0\),
\(c_n>0\), and \(d_n\ge0\).  Equivalently, solving the capital recursion
forward gives the strictly lower triangular matrix
\begin{equation}\label{eq:model-release-Q}
 Q_{nr}=a_n^{\mathrm{val}}\kappa_r
 \prod_{j=r+1}^{n-1}\delta_j\quad(r<n),
 \qquad Q_{nr}=0\quad(r\ge n).
\end{equation}
Strict lower triangularity in \(Q\) records staged validation: stage \(n\) may use
completed output from earlier stages, while output from the current stage cannot
return through the validated-capital channel before that stage closes.
Same-cycle settlement therefore enters separately through the release-value
term. The atomless formulation closes aggregate release value as a competitive
rational-expectations fixed point.

Put \(C=\diag(c_0,\ldots,c_{J-1})\),
\(A=C-Q\), \(\bar y=\ell^\top\bar e\), and
\[
 u_\lambda(y):=(1-\lambda)\bar\Gamma+\lambda W_R'(y).
\]
For \(u\in\mathbb R\), define the stagewise release response recursively by
\begin{align}
 \mathscr E_{\vartheta,0}(u)
 &=\Pi_{[0,\bar e_0]}
   \frac{\widetilde b_0^0-\vartheta d_0+\ell_0u}{c_0},
 \nonumber\\
 \mathscr E_{\vartheta,n}(u)
 &=\Pi_{[0,\bar e_n]}
   \frac{\widetilde b_n^0-\vartheta d_n
   +\sum_{r<n}Q_{nr}\mathscr E_{\vartheta,r}(u)+\ell_nu}{c_n},
 \qquad n\ge1,
 \label{eq:model-release-stagewise-map}\\
 \mathscr Y_\vartheta(u)&:=\ell^\top\mathscr E_\vartheta(u).
 \label{eq:model-release-output-map}
\end{align}
For a reconstructed vector \(e=\mathscr E_\vartheta(u)\), write
\begin{equation}\label{eq:model-release-raw-score}
 r_n(e,u,\vartheta):=
 \frac{\widetilde b_n^0-\vartheta d_n+\sum_{j<n}Q_{nj}e_j+\ell_nu}{c_n}.
\end{equation}
A projected-response cell is \emph{strictly complementary} at \((e,u,\vartheta)\)
when every free coordinate satisfies \(0<e_n<\bar e_n\), every lower-bound
coordinate satisfies \(r_n(e,u,\vartheta)<0\), and every upper-bound coordinate
satisfies \(r_n(e,u,\vartheta)>\bar e_n\).  These strict inequalities keep the
active set fixed under small perturbations.

\begin{theorem}[Global release reduction and protocol-indexed criticality]
\label{thm:model-release-reduction}
Let \(W_R\in C^2\) on a neighborhood of \([0,\bar y]\) and
\(\ell\in\mathbb R_+^J\setminus\{0\}\).

\textup{(i)} The full capacity-constrained finite-stage equilibrium is in
one-to-one correspondence with roots on \([0,\bar y]\) of
\begin{equation}\label{eq:model-release-global-scalar}
 \boxed{F^{\mathrm{box}}(y,\vartheta;\lambda)
 :=y-\mathscr Y_\vartheta(u_\lambda(y))=0.}
\end{equation}
Every root reconstructs one and only one equilibrium as
\(e=\mathscr E_\vartheta(u_\lambda(y))\).  At least one equilibrium exists for
every \((\vartheta,\lambda)\).

\textup{(ii)} The map \(\mathscr Y_\vartheta\) is continuous, nondecreasing, and
piecewise affine.  With
\begin{equation}\label{eq:model-release-beta}
 \boxed{\beta_R:=\ell^\top(C-Q)^{-1}\ell>0,}
\end{equation}
it satisfies, for \(u_2\ge u_1\),
\begin{equation}\label{eq:model-release-global-lipschitz}
 0\le\mathscr Y_\vartheta(u_2)-\mathscr Y_\vartheta(u_1)
 \le\beta_R(u_2-u_1).
\end{equation}
Consequently the equilibrium is unique whenever
\begin{equation}\label{eq:model-release-global-uniqueness}
 \lambda\beta_R
 \sup_{0\le y\le\bar y}[W_R''(y)]_+<1.
\end{equation}
For the full-interior face,
\begin{equation}\label{eq:model-release-path-sum}
 \boxed{\beta_R=\sum_{k=0}^{J-1}
 \ell^\top(C^{-1}Q)^kC^{-1}\ell.}
\end{equation}
Thus causal validation cannot create singularity when \(\lambda=0\), but its
nonnegative paths alter exposure to same-cycle model value.

\textup{(iii)} On a fixed projected-response cell with free coordinates
\(I\), \(\mathscr Y_\vartheta\) is affine with release-price slope
\begin{equation}\label{eq:model-release-face-beta}
 \boxed{\beta_{R,I}
 :=\ell_I^\top(C_I-Q_{II})^{-1}\ell_I\ge0}
\end{equation}
and cost-shifter slope
\begin{equation}\label{eq:model-release-face-kappa}
 \boxed{\kappa_{\vartheta,I}
 :=\ell_I^\top(C_I-Q_{II})^{-1}d_I\ge0.}
\end{equation}
The smooth cell residual satisfies
\begin{equation}\label{eq:model-release-face-derivative}
 \partial_yF_I=1-\lambda\beta_{R,I}W_R''(y),
 \qquad \partial_\vartheta F_I=\kappa_{\vartheta,I}.
\end{equation}
Equivalently, its response derivative separates into a nilpotent causal block
and a rank-one same-cycle block; after resolving the causal block, the single
possibly nonzero master eigenvalue is
\(\lambda\beta_{R,I}W_R''(y)\).  If \(W_R\) is \(C^3\) near a scalar turning point, that point lifts to a
classical smooth fold of the full equilibrium when the reconstructed point is
strictly complementary, \(\kappa_{\vartheta,I}>0\), and the scalar quadratic
coefficient is nonzero.  At a loss of strict complementarity the global
residual is generally only piecewise \(C^1\); the critical point is then a
boundary contact or kink unless additional matching conditions are proved.

\textup{(iv)} Protocol comparative statics. Fix the validation network, the
smooth cell \(I\), and all nonprotocol primitives. For protocols
\(P\in\{\fresh,\inh\}\) with optimized values \(W_R^P\in C^2\), define
\begin{equation}\label{eq:protocol-critical-gain-main}
 \mathfrak g_I^P(y):=\lambda\beta_{R,I}(W_R^P)''(y).
\end{equation}
Then
\begin{equation}\label{eq:protocol-critical-gain-difference-main}
 \boxed{\mathfrak g_I^{\fresh}(y)-\mathfrak g_I^{\inh}(y)
 =\lambda\beta_{R,I}
 \bigl[(W_R^{\fresh})''(y)-(W_R^{\inh})''(y)\bigr].}
\end{equation}
Whenever \(\beta_{R,I}>0\) and \((W_R^P)''(y)>0\), the formal statewise
same-cycle share at unit gain is
\begin{equation}\label{eq:protocol-statewise-threshold-main}
 \lambda_{\mathrm{spec}}^P(y)
 =\frac{1}{\beta_{R,I}(W_R^P)''(y)}.
\end{equation}
It lies in the admissible timing range \([0,1]\) exactly when this reciprocal is
at most one; otherwise no admissible same-cycle share reaches unit gain at that
state. Thus protocol replacement moves the unit-gain boundary through optimized
financial curvature; temporal support alone supplies no fresh-versus-inherited
ordering.
\end{theorem}

The scalar released-quality index makes the same-cycle block rank one. The
Electronic Companion gives the finite-factor extension, whose criticality is
governed by the corresponding resolved exposure matrix.

\paragraph{Robustness to nonquadratic research costs.}
Quadratic costs are not required for the scalar reduction or the
temporal-support mechanism; they make the global response piecewise affine,
the cell exposure constant, and the fold coefficient especially simple. With
\(\mathcal C_n''\ge\underline c_n>0\), the exact scalar reduction persists,
and on a smooth cell
\[
 \beta_{R,I}(e)=\ell_I^\top
 [\diag(\mathcal C_n''(e_n))_{n\in I}-Q_{II}]^{-1}\ell_I,
 \qquad
 1=\lambda\beta_{R,I}(e)W_R''(y)
\]
remains the local unit-gain condition. Strict complementarity is understood
relative to the marginal-cost thresholds \(\mathcal C_n'(0)\) and
\(\mathcal C_n'(\bar e_n)\). The fold nondegeneracy coefficient is the full
scalar second derivative and includes the state variation of
\(\beta_{R,I}(e)\). The Electronic Companion gives its closed form, the global
uniqueness certificate, and the fold details.

The theorem strengthens the interior reduction by retaining every capacity
face in one global scalar equation.  Under an exogenous perturbation of the
predetermined shadow value, \(\beta_{R,I}\) is locally identified on a smooth
face by
\begin{equation}\label{eq:model-release-beta-identification}
 \left.\frac{\partial y}{\partial\bar\Gamma}\right|_{\lambda=0,I}
 =\beta_{R,I}.
\end{equation}
Together with local optimized-value curvature and release-time records, this
provides the operational inputs to the timing diagnostic below.

Only total protocol-indexed optimized curvature enters the release gain. The
Electronic Companion gives a baseline accounting identity that separates
evaluator geometry from predictive-risk and reoptimization terms under a
chosen reference specification.

\begin{corollary}[Gaussian protocol benchmark and fold lift]
\label{cor:gaussian-protocol-criticality}\label{cor:model-release-gaussian-lift}
Specialize \eqref{eq:model-release-value-bridge} to the fresh Gaussian value
\begin{equation}\label{eq:model-release-gaussian-value}
 W_R(y)=W_R^{\fresh}(y)
 =\frac{(\frac32\sqrt{y+1/2}-1)^2}{2(y+3/2)}.
\end{equation}
Its positive curvature is maximized at zero,
\[
 W_R''(0)=\frac{39\sqrt2-20}{54},
\]
and decreases strictly to zero on its positive-curvature interval.  The
full-interior reduced positive-branch turning-point onset is
\begin{equation}\label{eq:model-release-onset}
 \boxed{\lambda_{\mathrm{on}}^{\fresh}
 =\frac{1}{\beta_RW_R''(0)}
 =\frac{54}{\beta_R(39\sqrt2-20)}.}
\end{equation}
For comparison, preserve the initial natural budget \(k_0=1/\sqrt2\). The
inherited positive-position value is
\begin{equation}\label{eq:gaussian-protocol-inherited-value-main}
 W_R^{\inh}(p)=\frac{(3p-\sqrt2)^2}{8p(p+1)},\qquad p=y+\frac12.
\end{equation}
Under the common normalization \(\beta_R=32\), both protocols attain their
maximal positive curvature at \(p=1/2\), giving
\begin{equation}\label{eq:gaussian-protocol-onsets-main}
 \boxed{\lambda_{\mathrm{on}}^{\fresh}\approx0.0480026,
 \qquad \lambda_{\mathrm{on}}^{\inh}\approx0.0122230.}
\end{equation}
The comparison is benchmark-specific: across Gaussian primitives the
fresh-minus-inherited curvature can have either sign. The Electronic Companion
gives the corresponding closed-form signal-to-ambiguity boundary. The protocol
implication is the transmission formula in part~\textup{(iv)} of
\cref{thm:model-release-reduction}, not a universal stability ordering.
If \(\lambda\le\lambda_{\mathrm{on}}^{\fresh}\), the
capacity-constrained equilibrium is unique for every \(\vartheta\). At equality,
the full-interior reduced derivative first vanishes at \(y=0\); at a cost
shifter that places \(y=0\) on the reduced branch, this is a boundary tangency,
not an interior fold. If
\(\lambda>\lambda_{\mathrm{on}}^{\fresh}\), the unconstrained
full-interior reduced equation has a unique \(y_f>0\) satisfying
\begin{equation}\label{eq:model-release-fold-condition}
 \lambda\beta_RW_R''(y_f)=1.
\end{equation}
Let
\[
 a_0(\lambda):=\ell^\top A^{-1}
 [\widetilde b^0+(1-\lambda)\bar\Gamma\ell],
 \qquad
 \kappa_\vartheta:=\ell^\top A^{-1}d,
\]
and assume the transversality condition \(\kappa_\vartheta>0\).  Since
\(A^{-1}\ge0\), this condition is equivalent to requiring that some
cost-exposed stage can reach a release-relevant stage through the nonnegative
validation graph; it is not implied by \(d\ge0\) alone.  Define
\begin{equation}\label{eq:model-release-fold-cost}
 \vartheta_f=\frac{a_0(\lambda)-y_f
 +\lambda\beta_RW_R'(y_f)}{\kappa_\vartheta}>0.
\end{equation}
The reduced turning point is a full interior nondegenerate saddle-node if and
only if
\begin{equation}\label{eq:model-release-fold-lift}
 e_f=A^{-1}\!
 [\widetilde b^0-\vartheta_fd+(1-\lambda)\bar\Gamma\ell
 +\lambda\ell W_R'(y_f)]
 \in\prod_{n=0}^{J-1}(0,\bar e_n).
\end{equation}
When this lift condition holds,
\begin{equation}\label{eq:model-release-branch-expansion}
 y_\pm(\vartheta)=y_f\pm
 \sqrt{\frac{2\kappa_\vartheta(\vartheta_f-\vartheta)}
 {-\lambda\beta_RW_R'''(y_f)}}
 +O(\vartheta_f-\vartheta).
\end{equation}
The singular research direction is \(A^{-1}\ell\), and every differentiable
portfolio observable with nonzero derivative at \(y_f\) inherits the
inverse-square-root pole.  On another projected-response cell the same
statements hold with its affine baseline, \(\beta_{R,I}\), and
\(\kappa_{\vartheta,I}>0\), provided strict complementarity keeps that cell
locally fixed.
\end{corollary}

A reduced turning point lifts only when the reconstructed effort lies in the
stated cell. The Electronic Companion gives a two-stage nonlift example. It
also proves that \(0\le Q_1\le Q_2\) entrywise, with \(C\) and \(\ell\)
fixed, implies \(\beta_R(Q_1)\le\beta_R(Q_2)\); adding stages without a nested
embedding does not imply this ordering.

\begin{theorem}[Financial multiplier and validation-latency identification]
\label{cor:model-release-deadline}\label{thm:financial-release-multiplier}
Let \(D\ge0\) be an exogenous validation time, independent of current research
\(e\) and released precision \(y\), and not controlled by the research
suppliers. Let \(\Delta\) be the release deadline and let unreleased alpha
decay at rate \(\rho\ge0\). A marginal unit of released precision has
realized settlement value
\begin{equation}\label{eq:model-release-realized-settlement}
 u_{\Delta,\rho}(D,y)
 :=\bar\Gamma+e^{-\rho D}\mathbf1_{\{D\le\Delta\}}
 [W_R'(y)-\bar\Gamma].
\end{equation}
Risk-neutral suppliers value the exogenous validation clock in expectation;
the linear settlement term in \eqref{eq:model-release-supplier-payoff} therefore has ex ante value
\begin{equation}\label{eq:model-release-expected-settlement}
 \mathbb E[u_{\Delta,\rho}(D,y)]
 =\bar\Gamma+\lambda_\rho(\Delta)[W_R'(y)-\bar\Gamma]
 =u_{\lambda_\rho(\Delta)}(y),
\end{equation}
where
\begin{equation}\label{eq:model-release-effective-share}
 \lambda_\rho(\Delta)
 :=\mathbb E[e^{-\rho D}\mathbf1_{\{D\le\Delta\}}].
\end{equation}
Thus \(\lambda_\rho(\Delta)\) is the value weight on validation completed before
the release cutoff; at \(\rho=0\) it is the same-cycle completion fraction.

\textup{(i) Financial materiality.}
Fix a strictly complementary quadratic-cost cell with free set \(I\), put
\(A_I=C_I-Q_{II}\), and perturb the free supplier scores by \(zh_I\). For a
release-dependent financial outcome \(\mathcal O(y)\in C^1\), define
\begin{equation}\label{eq:model-release-loading-definitions}
 \chi_I(y):=\beta_{R,I}W_R''(y),\qquad
 \gamma_{h,I}:=\ell_I^\top A_I^{-1}h_I,\qquad
 \omega_{\mathcal O,h,I}(y):=\mathcal O'(y)\gamma_{h,I}.
\end{equation}
At every nonsingular equilibrium on that cell,
\begin{equation}\label{eq:model-release-financial-multiplier}
 \boxed{
 \frac{\dd\mathcal O}{\dd z}
 =\frac{\omega_{\mathcal O,h,I}(y)}
 {1-\lambda_\rho(\Delta)\chi_I(y)}.}
\end{equation}
Structural criticality is financially silent at first order when either the
primitive does not reach released precision \((\gamma_{h,I}=0)\) or the outcome
does not load on it \((\mathcal O'(y)=0)\).

\textup{(ii) Sharp timing identification.}
Write \(p_\Delta=\mathbb P(D\le\Delta)\). For every validation-time law,
\begin{equation}\label{eq:model-release-distribution-free-bounds}
 \boxed{\underline\lambda_\Delta
 :=e^{-\rho\Delta}p_\Delta
 \le\lambda_\rho(\Delta)\le
 p_\Delta=: \overline\lambda_\Delta.}
\end{equation}
The interval is sharp over all nonnegative validation-time laws with the same
completion probability (allowing atoms); within atomless classes the endpoints
are generally limiting rather than attained. Holding the equilibrium state and
smooth cell fixed, timing information alone therefore gives the sharp
conditional set
\begin{equation}\label{eq:model-release-partial-id-interval}
 \boxed{\underline\lambda_\Delta\beta_{R,I}[W_R''(y)]_+
 \le\mathfrak C_I(y;\Delta)
 \le\overline\lambda_\Delta\beta_{R,I}[W_R''(y)]_+,}
 \qquad
 \mathfrak C_I(y;\Delta):=\beta_{R,I}[W_R''(y)]_+\lambda_\rho(\Delta).
\end{equation}
Let \(L_I\) and \(U_I\) denote the endpoints. Then \(U_I<1\) certifies strict
subcriticality for every compatible timing law, whereas \(L_I>1\) certifies
strict supercriticality. If the interval contains one, timing information alone
does not identify a strict side of unit gain.

If \(\chi_I(y)>0\) and
\(\overline\lambda_\Delta\chi_I(y)<1\), the conditional timing-identified
multiplier set is
\begin{equation}\label{eq:model-release-multiplier-id}
 \boxed{
 \frac{1}{1-\underline\lambda_\Delta\chi_I(y)}
 \le\frac{1}{1-\lambda_\rho(\Delta)\chi_I(y)}\le
 \frac{1}{1-\overline\lambda_\Delta\chi_I(y)}.}
\end{equation}
The corresponding sensitivity set is its image under multiplication by
\(\omega_{\mathcal O,h,I}(y)\). If the nondegenerate timing interval contains
\(1/\chi_I(y)\) and \(\omega_{\mathcal O,h,I}(y)\ne0\), no finite uniform bound
on absolute financial sensitivity is valid over nonsingular compatible timing
laws.

\textup{(iii) Exponential clock under the Gaussian release value.}
If \(D\sim\mathrm{Exp}(\nu)\), then
\begin{equation}\label{eq:model-release-exponential-share}
 \lambda_\rho(\Delta)
 =\frac{\nu}{\nu+\rho}
 \left(1-e^{-(\nu+\rho)\Delta}\right).
\end{equation}
For either Gaussian protocol value \(W_R^P\), \(P\in\{\fresh,\inh\}\), put
\(G_P=\beta_R(W_R^P)''(0)\). More generally, for a curvature index
\(G>0\), write \(r=\rho/\nu\) and \(d=\nu\Delta\). At zero released
precision \(y=0\), the normalized gain is
\begin{equation}\label{eq:model-release-dimensionless-index}
 \boxed{\mathfrak C(G,r,d)
 =\frac{G}{1+r}\left(1-e^{-(1+r)d}\right).}
\end{equation}
A finite unit-gain deadline exists if and only if \(G>1+r\), and then
\begin{equation}\label{eq:model-release-critical-deadline}
 \boxed{d_*=-\frac{1}{1+r}
 \log\!\left(1-\frac{1+r}{G}\right),\qquad
 \Delta_*=\frac{d_*}{\nu}.}
\end{equation}
For \(d<d_*\), the normalized sensitivity multiplier is
\(\mathfrak A(G,r,d)=[1-\mathfrak C(G,r,d)]^{-1}\), and
\begin{equation}\label{eq:model-release-deadline-pole}
 \boxed{\mathfrak A(G,r,d)
 \sim\frac{1}{(G-1-r)(d_*-d)}}
 \qquad(d\uparrow d_*).
\end{equation}
Whenever \(G_P>1+r\), the protocol-specific critical deadline is the
same formula with \(G=G_P\), and
\begin{equation}\label{eq:protocol-deadline-comparative-main}
 \frac{\partial d_*}{\partial G}
 =-\frac{1}{G(G-1-r)}<0.
\end{equation}
Thus the protocol acts on the admissible validation frontier through optimized
curvature. In the Gaussian normalization of
\cref{cor:gaussian-protocol-criticality}, \(G_{\inh}>G_{\fresh}\): when both
critical deadlines are finite, \(d_*^{\inh}<d_*^{\fresh}\); if
\(G_{\fresh}\le1+r<G_{\inh}\), only inheritance has a finite critical deadline;
if \(G_{\inh}\le1+r\), neither protocol reaches unit gain at \(y=0\)
by any finite deadline. At \(d=d_*\), the \(y=0\) contact is a boundary
tangency when a cost shifter places it on the reduced branch. A later deadline
creates a unique positive unconstrained unit-gain location on the stated
Gaussian value domain. Within the smooth active-cell class studied here, that
location is certified as a full saddle-node when its lift is attainable,
strict complementarity holds, and cost transversality and quadratic
nondegeneracy are satisfied.
\end{theorem}

\paragraph{Portfolio specialization.}
The generic loading in \eqref{eq:model-release-financial-multiplier} is directly
financial in the fresh Gaussian positive-position benchmark. With
\(p=y+1/2\),
\[
 \pi^{\fresh}(y)=\frac{\frac32p-\sqrt p}{p+1},\qquad
 (\pi^{\fresh})'(y)=\frac{p+3\sqrt p-1}{2\sqrt p\,(p+1)^2}>0,
\]
so a primitive score shock satisfies
\begin{equation}\label{eq:model-release-portfolio-multiplier}
 \boxed{
 \frac{\dd\pi^{\fresh}}{\dd z}
 =\frac{(\pi^{\fresh})'(y)\gamma_{h,I}}
 {1-\lambda_\rho(\Delta)\beta_{R,I}(W_R^{\fresh})''(y)}.}
\end{equation}
Taking \(\mathcal O=W_R^P\) gives the same denominator for optimized portfolio
value. Hence the critical mode amplifies a risky position and its optimized
financial value through the same distance-to-unit-gain term.

\begin{figure}[H]
\centering
\begin{minipage}[t]{.36\textwidth}
\centering
\includegraphics[width=\linewidth]{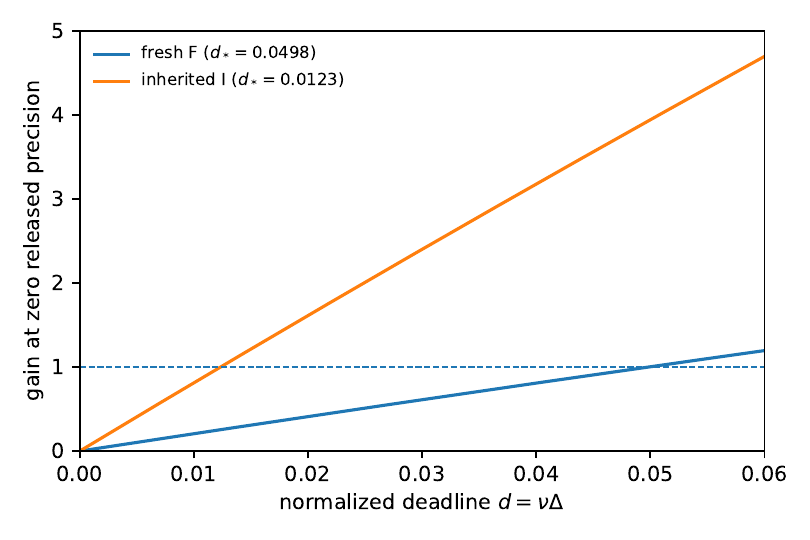}
\end{minipage}\hspace{.025\textwidth}
\begin{minipage}[t]{.36\textwidth}
\centering
\includegraphics[width=\linewidth]{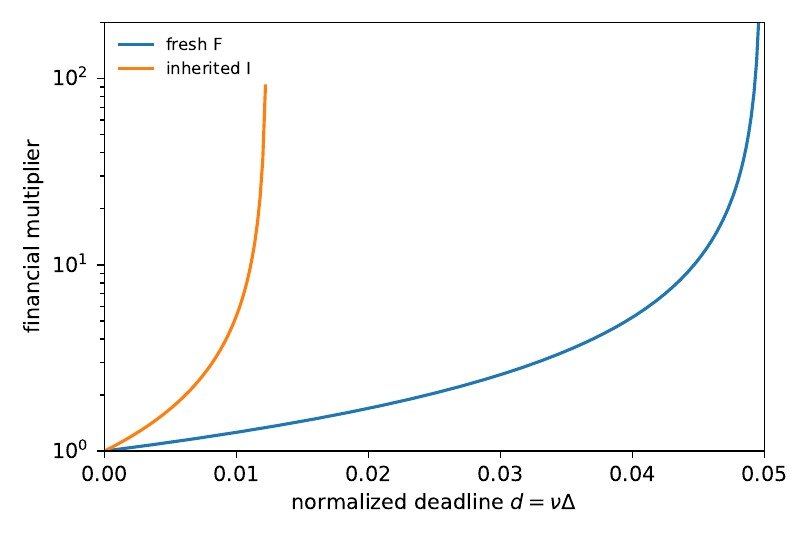}
\end{minipage}
\caption{Protocol-specific validation timing and financial materiality.}
\label{fig:endogenous-research}
\end{figure}

\noindent\textit{Notes to Figure~\ref{fig:endogenous-research}.}
Gaussian benchmark: resolved validation exposure \(\beta_R=32\); exponential-clock
decay/completion-rate ratio \(r=\rho/\nu=0.5\). Left: zero-released-precision gain,
with unit crossings at the benchmark critical deadlines. Right: subcritical financial
multiplier, diverging at those deadlines. The reduced frontiers require an attainable
smooth-cell lift and active-set matching for fold relevance.

\FloatBarrier

\section{Conclusion}\label{sec:conclusion}

Bayesian transport determines whether a confidence update preserves the
portfolio evaluator. Optimized value prices replacement through support-space
Bregman geometry. Within the recursive-rectangular Gaussian implementation, the
dynamic performance identity evaluates an inherited vintage under any prescribed
admissible policy, and the stopped recalibration tax is its fresh-reconstruction
specialization.

In the versioned release economy, causal validation is spectrally
nonsingular but economically amplifying: it accumulates resolved exposure
without itself closing a same-cycle unit mode. On a smooth
cell the protocol-indexed gain is
\(\mathfrak g_I^P=\lambda\beta_{R,I}(W_R^P)''\). The capacity-constrained
equilibrium reduces to a scalar equation. Within this smooth-cell class,
strict complementarity, attainability, transversality, and the full second
derivative certify whether a unit-gain state lifts to the classical
research-equilibrium saddle-node studied here.

The financially relevant object is the loaded resolvent:
\[
 \frac{\dd\mathcal O}{\dd z}
 =\frac{\omega_{\mathcal O,h,I}}
 {1-\lambda_\rho(\Delta)\beta_{R,I}W_R''}.
\]
It separates the four objects used throughout: protocol-indexed optimized
curvature, resolved validation exposure, same-cycle timing, and financial
loading. Completion-time information yields sharp conditional bounds while all compatible
timing laws remain subcritical; if the timing set reaches the pole, no finite uniform
sensitivity bound follows without further information. Together they link evaluator
provenance to measurable financial criticality.

\clearpage
\appendix
\section*{Electronic Companion}
\addcontentsline{toc}{section}{Electronic Companion}

\noindent
The Electronic Companion is organized by dependence on the Main Paper.
\textbf{Core proofs} establish the welfare, filtering, and stopped-control results.
\textbf{Causal-support certificates} verify the separator, spectral loading, and
fold conditions used in the equilibrium argument. \textbf{Robustness, extensions,
and worked benchmarks} collect hard-response limits, nonlinear costs, finite-factor
compression, and scalar certificates. The parabolic, vanishing-regularization,
and local analytic results in this last layer are sufficient-condition
certifications: regularity, invariant-domain, and selector assumptions are inputs,
not consequences of Gaussianity alone.

\section{Global protocol regret and governance loss}\label{app:cara-governance}

\subsection{Functional support-space regret and its decision metric}\label{app:functional-shape-regret}

Let \(\mathbb E=C(S^{d-1})\), let \(\mathscr H_d\) denote the support-function
cone, and let \(\ell_z\) be \eqref{eq:homogeneous-evaluation-main}.  Its
operator norm is \(\|\ell_z\|_{\mathbb E^*}=\|z\|\).

\begin{proof}[Proof of \cref{thm:global-budget-regret}]
For fixed \(a\), the map
\(h\mapsto L(a)-\ell_{Ba}(h)\) is continuous and affine.  Its pointwise
supremum \(\mathfrak V\) is therefore convex.  With
\(R_B=\sup_{a\in\mathcal A}\|Ba\|\),
\begin{align*}
 \mathfrak V(h)-\mathfrak V(H)
 &\le \sup_{a\in\mathcal A}\ell_{Ba}(H-h)
 \le R_B\|h-H\|_\infty,
\end{align*}
and interchanging \(h,H\) proves the Lipschitz bound.

Let \(a_H\) be the unique optimizer at \(H\).  For every \(h\in\mathbb E\),
\[
 \mathfrak V(h)
 \ge L(a_H)-\ell_{Ba_H}(h)
 =\mathfrak V(H)-\ell_{Ba_H}(h-H),
\]
so \(-\ell_{Ba_H}\in\partial\mathfrak V(H)\).  To obtain the Hadamard
derivative, take \(t_j\downarrow0\), \(g_j\to g\) in \(\mathbb E\), and an
optimizer \(a_j\) at \(H+t_jg_j\).  Compactness gives subsequential limits,
and continuity plus the preceding Lipschitz bound makes every such limit an
optimizer at \(H\), hence \(a_H\).  The inequalities
\begin{align*}
 -\ell_{Ba_H}(g_j)
 &\le \frac{\mathfrak V(H+t_jg_j)-\mathfrak V(H)}{t_j}
 \le-\ell_{Ba_j}(g_j)
\end{align*}
therefore have the same limit, proving
\eqref{eq:shape-subgradient-main}.

Direct subtraction using
\(\mathfrak V(H)=L(a_H)-\ell_{Ba_H}(H)\) gives
\[
 \mathfrak r(h,H)
 =\mathfrak V(h)-\mathfrak V(H)
 +\ell_{Ba_H}(h-H)
 =D_{\mathfrak V}^{-\ell_{Ba_H}}(h,H).
\]
Nonnegativity is the subgradient inequality.  Equality holds exactly when the
affine minorant generated by \(a_H\) also supports \(\mathfrak V\) at \(h\),
which is equivalent to optimality of \(a_H\) under \(h\).

For the noncompact extension, assume that the action menu $\mathcal A$ is
closed and convex and fix $H$.  Let
$F_H(a)=L(a)-\ell_{Ba}(H)$ and let $a_H$ be its unique maximizer.  We use the
following localized conditions, which avoid any circular appeal to local
strong concavity.  There is a neighborhood $\mathcal U_H$ of $H$ in
$\mathbb E$ and a set $A_H\subset\mathcal A$ such that every maximizer of
$F_h$ for $h\in\mathcal U_H$ is attained and lies in $A_H$,
$\sup_{a\in A_H}\|Ba\|\le R_H$, and, for some $\mu>0$,
\begin{equation}\label{eq:noncompact-quadratic-growth-v206}
 F_H(a_H)-F_H(a)\ge\frac\mu2\|a-a_H\|^2,
 \qquad a\in A_H.
\end{equation}
Uniform $\mu$-strong concavity of $F_H$ on the convex hull of $A_H$ is a
sufficient condition for this quadratic-growth inequality.  Let $a_j$
maximize $F_{H+t_jg_j}$, where $t_j\downarrow0$ and $g_j\to g$.  For all
large $j$, $a_j\in A_H$, and optimality under the perturbed objective gives
\begin{align}
 \frac{\mu}{2}\|a_j-a_H\|^2
 &\le F_H(a_H)-F_H(a_j)\nonumber\\
 &\le t_j\{\ell_{Ba_H}(g_j)-\ell_{Ba_j}(g_j)\}
 \le 2t_jR_H\|g_j\|_\infty.
 \label{eq:noncompact-optimizer-stability-v206}
\end{align}
Hence $a_j\to a_H$, and the same squeeze argument proves the local Hadamard
derivative and Bregman identity.  The localization assumption is essential for
this extension; a global $R_B$-Lipschitz estimate additionally requires a
global exposure bound.  In the quadratic finance specialization, $M\succ0$
supplies global strong concavity and therefore the localized conditions
automatically.

\paragraph{Quadratic finance decomposition.}
In the unconstrained quadratic case let $z_h=Ba_h$, $z_H=Ba_H$, and
$\widetilde h(z)=\ell_z(h)$.  Fermat's rule and the convex chain rule give
\[
 0\in Ma_h-r+B^\top\partial\widetilde h(z_h),
\]
so some $q_h\in\partial\widetilde h(z_h)$ satisfies
$B^\top q_h=r-Ma_h$.  Because $\widetilde h$ is one-homogeneous,
$q_h^\top z_h=\widetilde h(z_h)$.  Expanding the quadratic objective around
$a_h$ gives
\begin{equation}\label{eq:shape-face-decomposition-ec}
 \mathfrak r(h,H)
 =\frac12\|a_H-a_h\|_M^2
   +\widetilde h(z_H)-q_h^\top z_H.
\end{equation}
For $h=kh_Q$, $H=\bar k h_Q$, choose
$q_k\in\partial h_Q(Ba_k)$ with $kB^\top q_k=r-Ma_k$ and put
$\zeta_k=B^\top q_k$.  Then, in addition to the ray identity
\eqref{eq:global-regret-bregman},
\begin{align}
 \mathfrak r(k,\bar k)
 &=\tfrac12\|a_{\bar k}-a_k\|_M^2
 +k[h_Q(Ba_{\bar k})-q_k^\top Ba_{\bar k}],
 \label{eq:global-face-decomposition-ec}\\
 kD_{h_Q\circ B}^{\zeta_k}(a_{\bar k},a_k)
 &=k[h_Q(Ba_{\bar k})-q_k^\top Ba_{\bar k}].
 \label{eq:face-switch-support-gap-ec}
\end{align}
\end{proof}

\begin{lemma}[Smooth-stratum decision metric]\label{lem:functional-shape-hessian}
On an interior $C^2$ support-function stratum, put
$\widetilde h(z)=\ell_z(h)$, $z_h=Ba_h\ne0$, and
\begin{equation}\label{eq:shape-action-curvature-ec}
 C_h:=-D^2L(a_h)+B^\top D^2\widetilde h(z_h)B\succ0.
\end{equation}
Then the optimizer map is $C^1$ in a local $C^2$ chart and
\begin{align}
 Da_h[g]&=-C_h^{-1}B^\top\nabla\widetilde g(z_h),
 \label{eq:shape-action-derivative-ec}\\
 D^2\mathfrak V(h)[g_1,g_2]
 &=\nabla\widetilde g_1(z_h)^\top
 B C_h^{-1}B^\top\nabla\widetilde g_2(z_h).
 \label{eq:shape-value-hessian-ec}
\end{align}
Moreover,
\begin{equation}\label{eq:shape-metric-action-speed-app}
 D^2\mathfrak V(h)[g,g]
 =Da_h[g]^\top C_hDa_h[g]\ge0,
\end{equation}
and
\begin{equation}\label{eq:shape-metric-kernel-ec}
 D^2\mathfrak V(h)[g,g]=0
 \quad\Longleftrightarrow\quad
 B^\top\nabla\widetilde g(z_h)=0.
\end{equation}
Thus the local decision-visible Hessian has rank at most $n$.
\end{lemma}

\begin{proof}
The interior first-order equation is
\begin{equation}\label{eq:shape-optimizer-foc-app}
 \nabla L(a)-B^\top\nabla\widetilde h(Ba)=0.
\end{equation}
Its derivative in \(a\) at \((a_h,h)\) is \(-C_h\).  The implicit-function
theorem applies in any local \(C^2\) support-function chart on which
\(C_h\) remains positive definite.  Differentiating
\eqref{eq:shape-optimizer-foc-app} in direction \(g\) gives
\[
 -C_hDa_h[g]-B^\top\nabla\widetilde g(z_h)=0,
\]
which is \eqref{eq:shape-action-derivative-ec}.  The envelope derivative is
\(D\mathfrak V(h)[g_1]=-\widetilde g_1(z_h)\).  Differentiating in direction
\(g_2\) and substituting the optimizer derivative yields
\begin{align*}
 D^2\mathfrak V(h)[g_1,g_2]
 &=-\nabla\widetilde g_1(z_h)^\top BDa_h[g_2]\\
 &=\nabla\widetilde g_1(z_h)^\top
 BC_h^{-1}B^\top\nabla\widetilde g_2(z_h).
\end{align*}
Substituting
\(B^\top\nabla\widetilde g=-C_hDa_h[g]\) gives
\eqref{eq:shape-metric-action-speed-app}.  Positive definiteness of \(C_h\)
then proves the null-space characterization.  The factorization through the
\(n\)-dimensional action derivative gives rank at most \(n\).
\end{proof}

\paragraph{Scalar soft-threshold specialization.}
For $a\ge0$, let $\mathsf S_a(x):=\operatorname{sgn}(x)(|x|-a)_+$.  For
\(\ell_k(y)=xy-\frac{\gamma\sigma^2}{2}y^2-kv|y|\), put
\(\theta=|x|/v\), \(c=v^2/(2\gamma\sigma^2)\), and
\(y_k=\mathsf S_{kv}(x)/(\gamma\sigma^2)\).  The scalar restriction of the
functional theorem gives
\begin{align}
 \Phi_x(k)&=c(\theta-k)_+^2,\label{eq:soft-budget-envelope}\\
 \mathfrak r_x(k,\bar k)
 &=c[(\theta-k)_+^2-(\theta-\bar k)_+^2
 +2(\theta-\bar k)_+(k-\bar k)]
 \label{eq:soft-global-regret-compact}\\
 &=c\{[(\theta-\bar k)_+-(\theta-k)_+]^2
 +2(k-\theta)_+(\theta-\bar k)_+\}.
 \label{eq:soft-global-regret-decomposition}
\end{align}
The first term is action displacement and the second is face switching.  This
formula remains exact when either budget crosses the no-trade boundary.

\subsection{Protocol-regret curvature}\label{subsec:protocol-regret}

The compatibility criterion is geometric.  The next theorem supplies a
welfare metric for a small incompatibility in any smooth strongly concave
decision problem.  Its implicit-function and second-order sensitivity steps
are standard in parametric optimization; see \citet{BonnansShapiro2000}.  The
role of the theorem here is to identify the compatibility defect as the
perturbation and to isolate the curvature that later becomes an exact
financial tax.  The term \emph{protocol regret} refers here to evaluation
under an earlier parameter of an action selected after reconstructing that
parameter; it is not an alternative updating rule.

Let $\mathcal A\subset\mathbb R^n$ and $\mathcal Z\subset\mathbb R^r$ be
open, and let $J:\mathcal A\times\mathcal Z\to\mathbb R$.  For
$\zeta\in\mathcal Z$, write $a(\zeta)$ for the optimizer of
$a\mapsto J(a;\zeta)$.

\begin{theorem}[Protocol-regret curvature]\label{thm:protocol-regret-curvature}
Fix $\zeta_0\in\mathcal Z$.  Suppose there are convex neighborhoods
$\mathcal U\Subset\mathcal A$ of $a(\zeta_0)$ and
$\mathcal V\Subset\mathcal Z$ of $\zeta_0$ such that $J$ is $C^3$ on a
neighborhood of $\overline{\mathcal U}\times\overline{\mathcal V}$,
$a(\zeta)\in\mathcal U$ is the unique interior optimizer for every
$\zeta\in\mathcal V$, and
\begin{equation}\label{eq:protocol-strong-concavity}
 J_{aa}(a;\zeta)\preceq-\mu I_n
 \quad\text{for all }(a,\zeta)\in\mathcal U\times\mathcal V
 \quad\text{and some }\mu>0.
\end{equation}
Define the evaluator-$\zeta_0$ protocol regret
\begin{equation}\label{eq:protocol-regret}
 \mathfrak r(\zeta_0,\zeta)
 :=J(a(\zeta_0);\zeta_0)-J(a(\zeta);\zeta_0).
\end{equation}
Then:
\begin{enumerate}[label=(\roman*)]
\item
\begin{equation}\label{eq:regret-action-bound}
 \mathfrak r(\zeta_0,\zeta)
 \ge\frac\mu2\|a(\zeta)-a(\zeta_0)\|^2\ge0;
\end{equation}
\item $\nabla_\zeta\mathfrak r(\zeta_0,\zeta_0)=0$ and
\begin{equation}\label{eq:regret-curvature}
 \boxed{
 D^2_{\zeta\zeta}\mathfrak r(\zeta_0,\zeta_0)
 =J_{\zeta a}(-J_{aa})^{-1}J_{a\zeta}
 =:\mathscr C(\zeta_0)\succeq0,
 }
\end{equation}
where all derivatives on the right are evaluated at
$(a(\zeta_0),\zeta_0)$;
\item for every $d\in\mathbb R^r$,
\begin{equation}\label{eq:regret-local-expansion}
 \mathfrak r(\zeta_0,\zeta_0+\epsilon d)
 =\frac{\epsilon^2}{2}d^\top\mathscr C(\zeta_0)d+o(\epsilon^2).
\end{equation}
\end{enumerate}
If, in addition,
\begin{equation}\label{eq:exact-quadratic-evaluator}
 J(a;\zeta_0)
 =J(a_0;\zeta_0)-\frac12(a-a_0)^\top M_0(a-a_0),
 \qquad M_0\in\mathbb S_{++}^n,
\end{equation}
with $a_0=a(\zeta_0)$, then the regret is exactly
\begin{equation}\label{eq:exact-quadratic-regret}
 \boxed{
 \mathfrak r(\zeta_0,\zeta)
 =\frac12[a(\zeta)-a_0]^\top M_0[a(\zeta)-a_0].
 }
\end{equation}
\end{theorem}

\begin{proof}
Uniform strong concavity gives
\[
 J(a_0;\zeta_0)-J(a;\zeta_0)
 \ge-\nabla_aJ(a_0;\zeta_0)^\top(a-a_0)
 +\frac\mu2\|a-a_0\|^2.
\]
The first-order condition at the interior optimum makes the linear term zero,
which proves \eqref{eq:regret-action-bound}.  The implicit-function theorem
applied to $J_a(a(\zeta);\zeta)=0$ gives
\begin{equation}\label{eq:optimizer-sensitivity}
 Da(\zeta_0)=-J_{aa}^{-1}J_{a\zeta}.
\end{equation}
Differentiating \eqref{eq:protocol-regret} once gives zero at $\zeta_0$.
Differentiating a second time, all terms containing $J_a(a_0;\zeta_0)$ vanish,
and
\[
 D^2_{\zeta\zeta}\mathfrak r(\zeta_0,\zeta_0)
 =-[Da(\zeta_0)]^\top J_{aa}Da(\zeta_0)
 =J_{\zeta a}(-J_{aa})^{-1}J_{a\zeta}.
\]
The matrix is positive semidefinite because $-J_{aa}$ is positive definite.
Taylor's theorem gives \eqref{eq:regret-local-expansion}.  Finally,
\eqref{eq:exact-quadratic-regret} follows by substituting $a=a(\zeta)$ in
\eqref{eq:exact-quadratic-evaluator}.
\end{proof}

\subsection{Parameterized optimized envelope and research curvature}

\begin{proof}[Proof of \cref{thm:welfare-curvature-transmission}]
For part~\textup{(i)}, first localize the optimizer. On the compact closure $\overline A_0$, continuity of $J$ and uniqueness
of the maximizer imply, by the maximum theorem, that $a(E)\to a(E_*)$ as
$E\to E_*$.  Since $D^2_{aa}J(a(E_*),E_*)$ is invertible, the
implicit-function theorem supplies a unique $C^1$ critical branch through
$(a(E_*),E_*)$; after shrinking $U_0$, optimizer continuity identifies this
branch with the unique global maximizer $a(E)$.  The first-order condition
$D_aJ(a(E),E)=0$ therefore gives
\[
 Da(E)[u]=-[D^2_{aa}J]^{-1}D^2_{aE}J[u].
\]
Negative definiteness persists on the smaller neighborhood.  The envelope
identity follows because $D_aJ\,Da=0$.  Differentiating once more gives
\[
 D^2W(E)[u,v]=D^2_{EE}J[u,v]
 -D^2_{Ea}J[u][D^2_{aa}J]^{-1}D^2_{aE}J[v].
\]
On the stated smooth support-function stratum,
$W=\mathfrak V\circ h$, so the ordinary second-order Fr\'echet chain rule
gives \eqref{eq:welfare-curvature-transmission}.

For part~\textup{(ii)}, the response decomposition in the Main Paper follows
from the chain rule once $D\Gamma(E_*)=H_*^{\mathrm{sd}}+\mathcal L_*$ is
available. We record sufficient primitive conditions for that decomposition.
Let $\Xi:\mathbb X\to L^\infty(0,S;\mathbb R^q)$ be $C^1$ near $E_*$,
with $D\Xi(E_*)$ bounded and strictly causal, and write
$w(s,e,\xi)=D_eW_s(e;\xi)$. Suppose $w$ is Carath\'eodory and, for almost
every $s$, $C^1$ in $(e,\xi)$ on a common neighborhood of the essential
range of $(E_*,\Xi(E_*))$, with $D_ew$ and $D_\xi w$ essentially uniformly
bounded and uniformly continuous there. Assume also
$w(\cdot,E_*(\cdot),\Xi(E_*)(\cdot))\in L^\infty$. Put
$z_E=(E,\Xi(E))$ and let $\omega(\delta)$ be a common essential modulus of
continuity for $(D_ew,D_\xi w)$ on that neighborhood.  The Carath\'eodory
property makes $s\mapsto w(s,z_E(s))$ measurable.  Essential boundedness is not
implied by derivative bounds alone; it starts from the assumed base value
$w(\cdot,z_{E_*}(\cdot))\in L^\infty$.  For $E$ in a sufficiently small $L^\infty$ neighborhood,
the uniform derivative bound and local Lipschitz continuity of $\Xi$ give
\[
 \|\Gamma(E)-\Gamma(E_*)\|_\infty
 \le L_w\bigl(\|E-E_*\|_\infty+\|\Xi(E)-\Xi(E_*)\|_\infty\bigr),
\]
so $\Gamma(E)\in L^\infty$ on a neighborhood of $E_*$.  On that
neighborhood the Nemytskii map is well defined independently of the chosen
$L^\infty$ representatives.

For $h\to0$ in $\mathbb X$, write
\[
 \delta z_h=
 \bigl(h,\Xi(E_*+h)-\Xi(E_*)\bigr).
\]
The $C^1$ property of $\Xi$ gives
\[
 \delta z_h=\bigl(h,D\Xi(E_*)h\bigr)+o(\|h\|_\infty)
 \quad\text{in }L^\infty.
\]
For almost every $s$, the integral mean-value formula yields
\begin{align*}
 &w(s,z_{E_*}(s)+\delta z_h(s))-w(s,z_{E_*}(s))
   -Dw(s,z_{E_*}(s))\delta z_h(s)\\
 &\qquad=\int_0^1
 [Dw(s,z_{E_*}(s)+r\delta z_h(s))-Dw(s,z_{E_*}(s))]
 \delta z_h(s)\,\dd r.
\end{align*}
Taking the essential supremum bounds this remainder by
$\omega(\|\delta z_h\|_\infty)\|\delta z_h\|_\infty=o(\|h\|_\infty)$.
Define almost everywhere
\begin{align}
 H_*^{\mathrm{sd}}(s)
 &:=D_ew\left(s,E_*(s),\Xi(E_*)(s)\right),
 \label{eq:samedate-financial-hessian-app}\\
 (\mathcal L_*h)(s)
 &:=D_\xi w\left(s,E_*(s),\Xi(E_*)(s)\right)
       (D\Xi(E_*)h)(s).
 \label{eq:causal-financial-hessian-app}
\end{align}
Substituting the derivative of $\Xi$ proves
\begin{align*}
 (D\Gamma(E_*)h)(s)
 ={}&D_ew(s,E_*(s),\Xi(E_*)(s))h(s)\\
 &+D_\xi w(s,E_*(s),\Xi(E_*)(s))(D\Xi(E_*)h)(s),
\end{align*}
which is \eqref{eq:datewise-marginal-value-derivative-main}.  The first term is
a bounded multiplication operator.  The second is strictly causal because
$D\Xi(E_*)$ is strictly causal and its left multiplier is pointwise.

For the cost response, suppose $\psi_s=(\nabla c_s)^{-1}$ is defined on a
common neighborhood of $\Gamma(E_*)(s)$, is Carath\'eodory, and has an
essentially bounded, uniformly continuous derivative there. The same
Nemytskii argument shows that $\mathcal G(E)=\psi(\Gamma(E))$ is Fr\'echet
differentiable.  Since
$D\psi_s(\Gamma(E_*)(s))=C_*(s)^{-1}$ almost everywhere,
\[
 D\mathcal G(E_*)
 =C_*^{-1}(H_*^{\mathrm{sd}}+\mathcal L_*)
 =M_*^{\mathrm{sd}}+\mathcal V_*.
\]
Uniform positive definiteness makes multiplication by $C_*^{-1}$ bounded and
preserves strict temporal support.  If $\mathcal L_*$ obeys the stated
fractional-Volterra estimate, then so does $\mathcal V_*$ with its constant
multiplied by $\|C_*^{-1}\|_\infty$.  The Volterra lemma below therefore makes
$I-\mathcal V_*$ invertible and yields
$K_*=(I-\mathcal V_*)^{-1}M_*^{\mathrm{sd}}$.  The static and scalar
specializations follow immediately.
\end{proof}

\paragraph{Baseline accounting for release curvature.}
Fix a non-ambiguity reference specification \(\bar\zeta\), including
predictive covariance, feasible actions, and coefficients other than the
ambiguity support function, and define
\[
 \mathfrak V_{\bar\zeta}(h)
 :=\sup_{a\in\mathcal A_{\bar\zeta}}
 \{L_{\bar\zeta}(a)-\ell_{B_{\bar\zeta}a}(h)\}.
\]
For a protocol-indexed evaluator path \(h_R^P(y)\), write
\[
 P_R^{P,\bar\zeta}(y)
 :=W_R^P(y)-\mathfrak V_{\bar\zeta}(h_R^P(y)).
\]
Whenever these maps are twice differentiable on the relevant stratum, the
chain rule gives
\begin{align}
 (W_R^P)''(y)
 ={}&D^2\mathfrak V_{\bar\zeta}(h_R^P(y))
 [ (h_R^P)'(y),(h_R^P)'(y)]\nonumber\\
 &+D\mathfrak V_{\bar\zeta}(h_R^P(y))[(h_R^P)''(y)]
 +(P_R^{P,\bar\zeta})''(y).
 \label{eq:evaluator-criticality-transmission-app}
\end{align}
Only the left side is invariant to the reference specification. The first two
terms price evaluator geometry under the stated freeze; the remainder contains
predictive-risk, reoptimization, and other non-evaluator primitive changes.
Equation~\eqref{eq:evaluator-criticality-transmission-app} is an accounting
identity rather than a causal decomposition.

\subsection{Linear comparison equation and the global performance identity}

The argument in this subsection is classical, not viscosity-based.  On each
localization cylinder the two policy values are assumed to have the
$C^{1,2,1}$ regularity stated in the Main Paper, the induced comparison SDE is
well posed, and the stopped local martingale is a true martingale.  The global
formula additionally uses an exhausting localization, uniform integrability,
and the common terminal condition.  These hypotheses are part of the theorem,
not consequences of the formal Hamiltonian subtraction.

To avoid overloading the sophisticated auxiliary value $f^k$, let $F^k$ denote
the optimized current-vintage commitment certainty equivalent from
\cref{def:commitment}.  For a bounded prescribed feedback $\widehat\pi$, let
$g^{k,\widehat\pi}=\mathfrak J^k(\widehat\pi)$ be its evaluator-$k$ policy
certainty equivalent.  Put $F=F^k$, $g=g^{k,\widehat\pi}$,
$\Delta=F-g$, and let $\pi^{\pc,k}$ be the commitment optimizer.  The two
policy-evaluation equations are
\[
F_t+\Hh(k;\pi^{\pc,k};F)=0,
\qquad
g_t+\Hh(k;\widehat\pi;g)=0.
\]
Define the Hamiltonian gap
\[
 \mathfrak R^{k,\widehat\pi}
 :=\Hh(k;\pi^{\pc,k};F)-\Hh(k;\widehat\pi;F)\ge0.
\]
Add and subtract $\Hh(k;\widehat\pi;F)$ to obtain
\[
\Delta_t+
\{\Hh(k;\widehat\pi;F)-\Hh(k;\widehat\pi;g)\}
+\mathfrak R^{k,\widehat\pi}=0.
\]
The terms linear in $F_v,F_{mm}$ give
$-v^2I\Delta_v+v^2I\Delta_{mm}/2$.  The cross-wealth term gives
$-\gamma v\widehat\pi\Delta_m$, and the quadratic-gradient term gives
\[
-\frac\gamma2v^2I(F_m^2-g_m^2)
=-\frac\gamma2v^2I(F_m+g_m)\Delta_m.
\]
Finally, define
\[
 A_F:=\widehat\pi+vIF_m,\qquad
 A_g:=\widehat\pi+vIg_m.
\]
Then
\[
|A_F|-|A_g|=\chi(A_F-A_g)=\chi vI\Delta_m,
\qquad |\chi|\le1.
\]
Define the measurable secant
\[
 \chi^{k,\widehat\pi}
 :=\begin{cases}
 (|A_F|-|A_g|)/(A_F-A_g),&A_F\ne A_g,\\
 0,&A_F=A_g,
 \end{cases}
 \qquad |\chi^{k,\widehat\pi}|\le1,
\]
and therefore
\begin{equation}\label{eq:app-comparison-induced-drift}
 b^{k,\widehat\pi}
 =-\gamma v\widehat\pi
 -\frac{\gamma}{2}v^2I(F_m+g_m)
 -kv^2I\chi^{k,\widehat\pi}.
\end{equation}
Combining these terms gives
\begin{equation}\label{eq:app-comparison-linear-pde}
 \Delta_t-v^2I\Delta_v+\frac12v^2I\Delta_{mm}
 +b^{k,\widehat\pi}\Delta_m+\mathfrak R^{k,\widehat\pi}=0.
\end{equation}
Let $\mathbb P_{t,m,v}^{k,\widehat\pi}$ denote the law of
\[
 \dd M_s=b^{k,\widehat\pi}(s,M_s,V_s)\,\dd s
 +V_s\sqrt{I_s}\,\dd B_s,
 \qquad
 \dd V_s=-V_s^2I_s\,\dd s,
 \qquad (M_t,V_t)=(m,v),
\]
and let $\mathbb E_{t,m,v}^{k,\widehat\pi}$ be expectation under this law.
On each compact localization the coefficients are bounded and the diffusion in
the $m$ direction is uniformly nondegenerate.  For any localization stopping
time $\tau\in[t,T]$, It\^o's formula applied to
\eqref{eq:app-comparison-linear-pde} gives
\[
 \Delta(t,m,v)=\mathbb E_{t,m,v}^{k,\widehat\pi}\!\left[
 \Delta(\tau,M_\tau,V_\tau)+
 \int_t^\tau\mathfrak R_s^{k,\widehat\pi}\,\dd s\right].
\]
The stopped local-martingale term is a true martingale.  This is exactly
\eqref{eq:stopped-performance-FK}; in particular, stopping at the exit of a
fixed cylinder does \emph{not} remove the boundary value $\Delta(\tau)$.

For the global formula, take an increasing localization $\tau_n\uparrow T$
almost surely.  A sufficient condition used here is at-most-quadratic growth of
$F$ and $g$, at-most-linear growth of their $m$-derivatives, and the resulting
linear-growth comparison drift.  Standard finite-horizon SDE moment estimates
then give a uniform $(2+\eta)$-moment bound for the comparison state for some
$\eta>0$, hence
$\{\Delta(\tau_n,M_{\tau_n},V_{\tau_n})\}_n$ is uniformly integrable.
Because the two policy evaluations share the terminal certainty-equivalent
condition $F(T,\cdot,\cdot)=g(T,\cdot,\cdot)=0$, the boundary term converges to
zero in $L^1$, and letting $n\to\infty$ gives
\eqref{eq:global-performance-FK}.  If instead one works on a fixed stopped
cylinder, the boundary term must be retained unless common exit continuation
data are imposed.  The same stopped formula follows by approximation if the
absolute-value secant is selected measurably at points where $A_F=A_g$.

For the source decomposition, set
$y=\pi+vIF_m$ and
$x_F=m+\gamma v(\sigma^2I-1)F_m$.  Direct expansion gives
\[
\Hh(k;\pi;F)
=x_Fy-\frac{\gamma\sigma^2}{2}y^2-kv|y|
+\text{terms independent of }\pi.
\]
The optimizer $y_k$ satisfies
$x_F-\gamma\sigma^2y_k=kv\zeta_k$ with
$\zeta_k\in\partial|y_k|$.  Therefore
\begin{align*}
\mathfrak R^{k,\widehat\pi}
&=x_F(y_k-\widehat y)
-\frac{\gamma\sigma^2}{2}(y_k^2-\widehat y^2)
-kv(|y_k|-|\widehat y|)\\
&=\frac{\gamma\sigma^2}{2}(\widehat y-y_k)^2
+kv[|\widehat y|-|y_k|-\zeta_k(\widehat y-y_k)].
\end{align*}
This proves the global dynamic identity.  Notice that no sign-stability
assumption appears in the comparison equation.  The common-face square formula
is recovered because the support gap vanishes and the universal CARA
coefficients make the remaining source deterministic conditional on the
variance path.

\section{Gaussian updating and inherited distortions}\label{app:gaussian}

\subsection{Natural-parameter transport}

Let the nominal date-$s$ posterior be $N(m_s,V_s)$ and let an alternative posterior with the same covariance have mean $m_s^q=m_s+V_sq$. Under the common likelihood
\[
L_{s,t}(\theta)\propto
\exp\left\{\theta^\top b_{s,t}-\frac12\theta^\top A_{s,t}\theta\right\},
\qquad A_{s,t}\succeq0,
\]
the nominal and alternative natural parameters satisfy
\begin{align*}
V_t^{-1}&=V_s^{-1}+A_{s,t},
&V_t^{-1}m_t&=V_s^{-1}m_s+b_{s,t},\\
(V_t^q)^{-1}&=V_s^{-1}+A_{s,t},
&(V_t^q)^{-1}m_t^q&=V_s^{-1}m_s^q+b_{s,t}.
\end{align*}
The covariances are equal, $V_t^q=V_t$, and
\[
V_t^{-1}(m_t^q-m_t)
=V_s^{-1}(m_s^q-m_s)=q.
\]
This proves \cref{thm:matrix-inheritance}. The calculation is pathwise in $b_{s,t}$ and does not average over the realized signal.

If the date-$s$ set is $m_s+V_s\mathcal Q_s$, updating every member therefore gives $m_t+V_t\mathcal Q_s$. A fresh set $m_t+V_t\mathcal Q_t$ is compatible with that inheritance exactly when $\mathcal Q_t=\mathcal Q_s$.

\subsection{Controlled observations and separated dynamics}\label{app:controlled-filter}
Let
\[
 \dd R_t=\theta\dd t+\sigma\dd W_t^S,
 \qquad \dd Y_t=\sqrt{e_t}\theta\dd t+\dd W_t^Y,
\]
with bounded observation-predictable $e$.  Girsanov's likelihood on the
canonical observation space is
\[
 L_t(\theta)=\exp\!\left\{\theta\left(
 \frac{R_t}{\sigma^2}+\int_0^t\sqrt{e_s}\dd Y_s\right)
 -\frac{\theta^2}{2}\left(
 \frac{t}{\sigma^2}+\int_0^te_s\dd s\right)\right\}.
\]
Multiplication by the Gaussian prior gives the posterior
$N(m_t,p_t^{-1})$ with
\[
 p_t=v_0^{-1}+t/\sigma^2+\int_0^te_s\dd s,
 \qquad
 p_tm_t=v_0^{-1}m_0+R_t/\sigma^2+\int_0^t\sqrt{e_s}\dd Y_s.
\]
The normalized innovations have unit quadratic variations and zero
cross-variation, hence are Brownian by L\'evy's characterization.  We use here
the standard Girsanov theorem, L\'evy characterization, and Novikov criterion;
see \citet{RevuzYor1999}.  Product calculus gives the separated state equations in
the controlled Gaussian filter closure in the Main Paper.  For a predictable natural distortion $q$,
the stopped exponential martingale with integrand
$(v_sq_s/\sigma,\sqrt{e_s}v_sq_s)$ is valid by bounded-domain Novikov.
Substitution under the resulting measure gives the inherited predictive drift
$m_t+v_tq_t$ and the separated wealth equation.  This proves the proposition.

\section{CARA transform and stopped spike verification}\label{app:verification}

\subsection{Robust Hamiltonian}

Let
\[
V^k(t,x,m,v)=-\exp\{-\gamma[x+f^k(t,m,v)]\}.
\]
For fixed $(\pi,q)$, It\^o's formula and division by the positive factor $-\gamma V^k$ give the certainty-equivalent drift
\begin{align*}
\mathcal G^{\pi,q}f^k
={}&\pi(m+vq)-v^2I_tf_v^k+\frac12v^2I_tf_{mm}^k\\
&-\frac\gamma2\left[
\sigma^2\pi^2+2v\pi f_m^k+v^2I_t(f_m^k)^2
\right]
+qv^2I_tf_m^k.
\end{align*}
The terms involving $q$ are
\[
qv\left(\pi+vI_tf_m^k\right).
\]
For $q\in[-k,k]$, a Borel minimizing selector is
\[
q^*(t,m,v;\pi,k)
=-k\,\sgn\left(\pi+vI_tf_m^k\right),
\]
with any value in $[-k,k]$ at zero. Taking the lower envelope yields the Hamiltonian \eqref{eq:Hamiltonian}.

\paragraph{One-step coincidence of fixed and rectangular distortions.}
At a fixed state, fixed continuation derivative, and fixed portfolio action,
\[
 \inf_{q\in[-k,k]}qv(\pi+vI_tf_m^k)
 =-kv\lvert\pi+vI_tf_m^k\rvert.
\]
Hence a single distortion chosen for one remaining decision and the pointwise
selector induced by the rectangular class have the same current Hamiltonian
support term.  Rectangularity changes the intertemporal coupling of distortions
and continuation values, not the instantaneous robust portfolio geometry.

\subsection{A detailed local verification argument}

For this subsection write $\CE_t^k(\pi):=\mathfrak J_t^k(\pi)$ for the evaluator-$k$ robust certainty equivalent defined in \eqref{eq:frozen-ce-definition}.

We give the proof of the stopped spike-verification statement in the Main Paper. Let $\mathcal D$ be the bounded cylinder in \cref{ass:admissibility}. Choose an increasing localization sequence $\tau_n$ that stops before exit and before the accumulated quadratic variation or absolute control exceeds $n$. The local boundedness assumptions imply, uniformly over bounded spike actions and $q\in\mathcal Q_k$,
\begin{equation}\label{eq:uniform-moment-app}
\sup_{0<h<h_0}
\E_t^q\left[
\sup_{u\in[t,t+h\wedge\tau_n]}
\bigl(|X_u-X_t|^2+|m_u-m|^2+|v_u-v|^2\bigr)
\right]=O(h).
\end{equation}
The exponential-moment assumption makes the stopped exponential value process uniformly integrable.

Fix the current state and budget $k=\kappa(v)$. For a bounded spike $\eta$, let $\pi^{h,\eta}$ equal $\eta$ on $[t,t+h)$ and the candidate equilibrium feedback thereafter. Applying It\^o's formula to $V^k$ on $[t,t+h\wedge\tau_n]$ under a fixed distortion process $q$ gives
\begin{align}
&\E_t^q\!
\left[V^k(t+h\wedge\tau_n,X_{t+h\wedge\tau_n},m_{t+h\wedge\tau_n},v_{t+h\wedge\tau_n})\right]
-V^k(t,x,m,v)\nonumber\\
&\quad=
\E_t^q\int_t^{t+h\wedge\tau_n}
(-\gamma V^k_u)
\left[
\partial_uf^k+\mathcal G^{\eta,q_u}f^k
\right]\dd u.
\label{eq:ito-spike-app}
\end{align}
The local regularity of $f^k$, \eqref{eq:uniform-moment-app}, and continuity of the coefficients imply the expansion
\begin{equation}\label{eq:uniform-spike-app}
\frac{1}{h}
\left\{
\CE_t^k(\pi^{h,\eta})-\CE_t^k(\pi^*)
\right\}
=
\Hh(k;\eta;f^k)(t,m,v)
-\Hh(k;\pi^*;f^k)(t,m,v)
+o(1),
\end{equation}
where the remainder is uniform over spike actions in a fixed bounded set. To justify the robust envelope, use the explicit measurable minimizer above for the upper bound and an $\varepsilon$-minimizing distortion for the reverse bound. Compactness of $[-k,k]$ and the local continuity of the integrand make the two bounds agree to first order.

Equation \eqref{eq:frozen-pde} removes the equilibrium drift. At the current diagonal $k=\kappa(v)$, \eqref{eq:diag-selector} makes the right-hand side of \eqref{eq:uniform-spike-app} nonpositive. Letting first $h\downarrow0$ and then $n\uparrow\infty$ yields \eqref{eq:spike-definition}. Because the exponential utility and the certainty equivalent are strictly increasing transforms of one another, the same first-order condition verifies weak equilibrium in utility units.

For the joint portfolio--research model, take the bounded admissible observation-feedback spike pair $\eta=(\eta^\pi,\eta^e)$ from \cref{def:spe}, replace $I_t$ on the spike interval by $\sigma^{-2}+\eta^e$, add $-c(\eta^e)$ to the certainty-equivalent drift, and maximize the diagonal Hamiltonian over both components.  The controlled-filter representation in the controlled Gaussian filter closure in the Main Paper keeps the spike admissible, and the same localization and uniform-remainder estimate proves the joint condition \eqref{eq:spike-definition}.

\section{Exogenous information: coefficient equations and exact tax}\label{app:exogenous}
On the positive sign-stable branch write
$f^k(t,m,v_t)=A(t)m^2+k\beta(t)m+C^k(t)$.  The portfolio first-order condition
and coefficient matching give
\begin{align*}
 A'&=2\gamma\iota v^2A^2+2vA/\sigma^2-(2\gamma\sigma^2)^{-1},
 &A(T)&=0,\\
 \beta'&=\frac{v}{\gamma\sigma^2}(1+\gamma\beta)
 (1+2\gamma\sigma^2\iota vA),
 &\beta(T)&=0.
\end{align*}
With $\tau=T-t$ and $D=\sigma^2+\tau v$, direct substitution in the
two coefficient equations gives
\[
 A=\frac{\tau}{2\gamma D},\qquad
 1+\gamma\beta=\frac{\sigma^2}{D}=: \alpha.
\]
On the positive exposed face, the portfolio first-order condition is
\[
 m-\gamma\sigma^2\pi-\gamma v f_m^k-kv=0.
\]
Because $f_m^k=2Am+k\beta=\tau(m-kv)/(\gamma D)$, this becomes
\[
 \pi^k=\frac{\alpha}{\gamma\sigma^2}(m-kv),\qquad
 \pi^k+vI f_m^k
 =\frac{1+\tau vI}{\gamma D}(m-kv).
\]
Thus the current selector and current-vintage commitment selector coincide
when they use the same current budget, and the sign identity in
the Gaussian--CARA coefficient closure in the Main Paper follows because $(1+\tau vI)/(\gamma D)>0$.
For two positive-face budgets $k$ and $\bar k$ the same formula gives
\[
 \pi^{\bar k}-\pi^k=-\frac{\alpha v}{\gamma\sigma^2}(\bar k-k).
\]
The evaluator-$k$ Hamiltonian is a concave quadratic in the positive-face
portfolio with curvature $-\gamma\sigma^2$; completing the square therefore
gives
\[
 \mathfrak r_u(k,\bar k)
 =\frac{\gamma\sigma^2}{2}(\pi_u^{\bar k}-\pi_u^k)^2
 =\frac{\alpha(u,v_u)^2v_u^2}{2\gamma\sigma^2}(\bar k-k)^2.
\]
The $n$-asset common-face analogue follows identically on the free coordinates
of a fixed face under $\Sigma\in\mathbb S_{++}^n$:
for $\ell_k(\pi)=\pi^\top r-kb^\top\pi-\gamma\pi^\top\Sigma\pi/2$,
$\pi_k=(\gamma\Sigma)^{-1}(r-kb)$, and substitution yields
$(\bar k-k)^2b^\top\Sigma^{-1}b/(2\gamma)$.

The optimized current-vintage commitment value $F^k$ has the same coefficients
$A$ and $k\beta$ as the sophisticated auxiliary value on this positive common
face; only the constant coefficient differs.  This equality is not assumed.
To see it directly, let $k^\circ=\kappa(v)$ denote the current diagonal budget used
by the future sophisticated selector.  With the quadratic ansatz above, that
selector is
\[
 \pi^*=\frac{m-\gamma v(2Am+k^\circ\beta)-k^\circ v}{\gamma\sigma^2}.
\]
Substituting this feedback into the evaluator-$k$ equation and collecting the
$m^2$ and $km$ coefficients eliminates $k^\circ$ exactly and gives
\begin{align*}
 A'&=2\gamma\iota v^2A^2+2vA/\sigma^2-(2\gamma\sigma^2)^{-1},\\
 \beta'&=\frac{v}{\gamma\sigma^2}(1+\gamma\beta)
 (1+2\gamma\sigma^2\iota vA).
\end{align*}
These are precisely the optimized frozen-$k$ commitment coefficient equations;
hence $F_m^k=f_m^k$ on the common positive face, while their constants may
differ.  By definition of the sophisticated policy value,
$g^{k,\pi^{\spe}}=f^k$.  Subtracting its constant equation from the commitment
constant equation and using the completed-square identity above leaves exactly
\[
 \frac{\alpha(u,v_u)^2v_u^2}{2\gamma\sigma^2}
 [\kappa(v_u)-\kappa(v_t)]^2.
\]
The calculation above identifies the common-positive-face source, but it must
be combined with the stopped comparison identity rather than integrated through
a face change.  Let $\tau_{\mathrm{cf}}$ be the first exit from the common
positive face, capped at $T$.  Substitution into
\eqref{eq:stopped-performance-FK} gives
\[
 \Tax_t=\mathbb E_{t,m_t,v_t}^{k_t,\pi^{\spe}}\!\left[
 \Delta(\tau_{\mathrm{cf}},M_{\tau_{\mathrm{cf}}},V_{\tau_{\mathrm{cf}}})
 +\int_t^{\tau_{\mathrm{cf}}}
 \frac{\alpha(u,v_u)^2v_u^2}{2\gamma\sigma^2}
 [\kappa(v_u)-\kappa(v_t)]^2\,\dd u\right].
\]
If the two policies have common exit continuation data, the displayed boundary
gap vanishes, but the source remains integrated only to the random exit time
$\tau_{\mathrm{cf}}$.  The common-face source may be integrated directly to
$T$ only when $\tau_{\mathrm{cf}}=T$ almost surely.  Under this no-exit
condition the integrand is nonnegative and the tax vanishes exactly when
$\kappa(v_u)=\kappa(v_t)$ for almost every $u\in[t,T]$.

For the stock-only fixed-level credible rule
$\kappa(v)=z_\alpha/\sqrt v=z_\alpha\sqrt p$ and
$p_u=p_t+(u-t)/\sigma^2$, direct integration gives
\[
 \Tax_t=\frac{v_T^2z_\alpha^2}{2\gamma\sigma^2}
 \left[2p_t(T-t)+\frac{(T-t)^2}{2\sigma^2}
 -\frac{4\sigma^2\sqrt{p_t}}3
 \{p_T^{3/2}-p_t^{3/2}\}\right],
\]
which proves part~\textup{(iii)} of \cref{thm:global-performance-identity}. The
evaluator $\mathcal Q_k$ is rectangular within a fixed vintage. Accordingly,
the identity prices this recursive implementation rather than the generally
nonrectangular family obtained by transporting each prior across vintages.

\section{Governed hard control and causal information feedback}\label{app:governed-hard-tariff}

\subsection{Joint portfolio--research selector}
For a frozen evaluator budget $k$, put $I(e)=\sigma^{-2}+e$.  At a fixed
state and a fixed smooth continuation jet, the endogenous CARA Hamiltonian is
\begin{align}
 \mathcal H_k(\pi,e;f)
={}&\pi m-c(e)-v^2I(e)f_v+\frac12v^2I(e)f_{mm}\nonumber\\
&-\frac\gamma2\left[\sigma^2\pi^2+2v\pi f_m+v^2I(e)f_m^2\right]
-kv\left|\pi+vI(e)f_m\right|.\label{eq:ec-joint-hamiltonian-v172}
\end{align}
This is the lower-envelope Hamiltonian obtained from the controlled Gaussian filter closure in the Main Paper.

\begin{proof}[Proof of the joint-selector closure]
Let $\mu=\min\{\gamma\sigma^2,\chi_0\}$.  In
\[
 (\pi,e)\longmapsto
 \mathcal H_k(\pi,e;f)+\frac\mu2(\pi^2+e^2),
\]
the term $-kv|\pi+vI(e)f_m|$ is the negative of a convex function of an
affine map, $-c(e)+\mu e^2/2$ is concave by
\cref{ass:strong-research-cost}, the portfolio quadratic is
$-(\gamma\sigma^2-\mu)\pi^2/2$, and every remaining term is affine in
$(\pi,e)$.  Hence $\mathcal H_k$ is jointly $\mu$-strongly concave.
Coercivity in $\pi$ and compactness of $[0,\bar e]$ give existence; strong
concavity gives uniqueness.

For fixed $e$, set $y=\pi+vI(e)f_m$.  The terms depending on $y$ reduce to
\[
 a_ey-\frac{\gamma\sigma^2}{2}y^2-kv|y|,
 \qquad a_e=m+\gamma\sigma^2vef_m.
\]
The unique maximizer is
$y_e^*=\mathsf S_{kv}(a_e)/(\gamma\sigma^2)$, which gives the portfolio
formula in the Main Paper.

Finally, $\mathcal H_k(\pi,e;f)+\chi_0e^2/2$ is jointly concave by the same
argument.  A partial supremum over $\pi\in\mathbb R$ preserves concavity, so
the reduced research objective plus $\chi_0e^2/2$ is concave.  Partial
optimization therefore preserves at least $\chi_0$ units of strong concavity
in $e$, and the projected research KKT problem has a unique solution.
\end{proof}

\subsection{Precision-time representation and comparison boundary}

\paragraph{Step 1: the exact time change.}
Write $I_t=\sigma^{-2}+e_t$. Since $I_t\ge\sigma^{-2}>0$,
\[
 p_t=p_0+\int_0^t I_s\,\dd s
\]
is strictly increasing up to the finite horizon and admits the inverse clock
$t(p)$. For $i\in\{S,Y\}$ set
\[
 \bar W^i_p=\int_0^{t(p)}\sqrt{I_s}\,\dd W_s^{i,q}.
\]
The two-dimensional continuous local martingale
$(\bar W^S,\bar W^Y)$ has bracket matrix
$(p-p_0)I_2$. By L\'evy's characterization it is a two-dimensional Brownian
motion in the precision-time filtration. Since
$\dd W^{i,q}_{t(p)}=I(\bar e_p)^{-1/2}\dd\bar W^i_p$, substitution into the
separated filter and wealth equations gives
\begin{align*}
 \dd\bar t_p&=I(\bar e_p)^{-1}\dd p,\\
 \dd\bar m_p
 &=p^{-2}\bar q_p\dd p
 +p^{-1}\left[
 \frac{\sigma^{-1}}{\sqrt{I(\bar e_p)}}\dd\bar W^S_p
 +\frac{\sqrt{\bar e_p}}{\sqrt{I(\bar e_p)}}\dd\bar W^Y_p
 \right],\\
 \dd\bar X_p
 &=\left\{u_p(\bar m_p+p^{-1}\bar q_p)
 -\frac{c(\bar e_p)}{I(\bar e_p)}\right\}\dd p
 +\sigma u_p\sqrt{I(\bar e_p)}\dd\bar W^S_p,
\end{align*}
where $u_p=\pi_{t(p)}/I(\bar e_p)$. Hence
\[
 \dd\langle\bar m\rangle_p
 =p^{-2}\frac{\sigma^{-2}+\bar e_p}{I(\bar e_p)}\dd p
 =p^{-2}\dd p,
 \qquad
 \dd\langle\bar X,\bar m\rangle_p=\frac{u_p}{p}\dd p.
\]
Thus posterior-mean quadratic variation per unit precision is independent of
the research rate; only calendar speed changes.

\begin{lemma}[Precision-clock control correspondence and DPP]
\label{lem:precision-time-dpp}
Fix a stopped cylinder and a compatible frozen budget $k$.  Let
$\mathfrak A^{\rm cal}$ be the stopped strong admissible class: $e\in[0,\bar e]$ and
$q\in[-k,k]$ are progressively measurable, the portfolio satisfies the
localized square- and exponential-integrability conditions, state equations
are pathwise unique, and the class is stable under stopping-time
concatenation and measurable selection.  Define the precision-time admissible
class to be the image of $\mathfrak A^{\rm cal}$ under the clock change, not
an independently enlarged class.

Up to the exit
$\zeta=\inf\{p\ge p_0:\bar t_p=T\}$, the map
\begin{equation}\label{eq:calendar-to-precision-controls-app}
 (\pi_t,e_t,q_t)\longmapsto
 \left(u_p,e_p,\bar q_p\right)
 =\left(\frac{\pi_{t(p)}}{I(e_{t(p)})},e_{t(p)},q_{t(p)}\right)
\end{equation}
is a bijection, modulo indistinguishability, between calendar-time and
precision-time investor/nature control triples in the filtrations
$\mathcal F_t$ and $\overline{\mathcal F}_p=\mathcal F_{t(p)}$.  Its inverse
solves $\dd\bar t_p=I(e_p)^{-1}\dd p$, takes the inverse clock, and sets
$\pi_t=I(e_{p(t)})u_{p(t)}$ and $q_t=\bar q_{p(t)}$.  Triple by triple it
preserves the predictive measure, state law, terminal wealth, admissibility,
and stopping-time concatenation.  Moreover,
\[
 p_0+\underline I(T-t_0)\le\zeta
 \le p_0+\overline I(T-t_0),
 \qquad \underline I=\sigma^{-2},\quad
 \overline I=\sigma^{-2}+\bar e.
\]
For every bounded $\overline{\mathcal F}_p$-stopping time $\sigma$, put
$\widehat\sigma=\sigma\wedge\zeta$.  If $\mathcal W$ denotes the utility
value, then
\begin{equation}\label{eq:precision-time-dpp-app}
 \mathcal W(p,t,m,x)
 =\sup_{(u,e)}\inf_{\bar q}
 \E^{\bar q}_{p,t,m,x}\!
 \left[\mathcal W(\widehat\sigma,\bar t_{\widehat\sigma},
 \bar m_{\widehat\sigma},\bar X_{\widehat\sigma})\right].
\end{equation}
Equivalently, under
$\mathcal W=-\exp[-\gamma(x+f)]$,
\begin{equation}\label{eq:precision-time-ce-dpp-app}
 e^{-\gamma f(p,t,m)}
 =\inf_{(u,e)}\sup_{\bar q}
 \E^{\bar q}_{p,t,m,x}\!
 \left[
 e^{-\gamma\{\bar X_{\widehat\sigma}-x+
 f(\widehat\sigma,\bar t_{\widehat\sigma},
 \bar m_{\widehat\sigma})\}}
 \right].
\end{equation}
No new rectangularization is introduced by the time change.
\end{lemma}

\begin{proof}
For every admissible triple, $I(e)\in[\underline I,\overline I]$ makes the
clock pathwise absolutely continuous, continuous, and strictly increasing.
The standard time-change results for continuous martingales and progressively
measurable processes (see, e.g., Chapter~V of \citet{RevuzYor1999}) therefore carry
every control to the filtration
$\overline{\mathcal F}_p=\mathcal F_{t(p)}$; solving the absolutely continuous
calendar-state equation gives the inverse.  The stochastic-integral
substitution in Step~1 proves pathwise equality of the transformed state
processes.  Since $q$ is bounded, the calendar and precision Girsanov
densities are true martingales and are related by the same substitution, so
the predictive measures and robust objectives agree for each control triple.
Consequently the two successive bijections---first for $(\pi,e)$ and then,
conditional on that pair, for $q$---preserve the order
$\sup_{(\pi,e)}\inf_q$.

A calendar stopping time $\tau$ is carried to $p_\tau$, while a precision
stopping time $\sigma$ stopped at $\zeta$ is carried to
$t(\sigma\wedge\zeta)$.  Continuity and strict monotonicity preserve the
stopping-time property.  Pasting controls at either stopping time and applying
the inverse clock gives exactly the pasted triple in the other formulation;
thus the correspondence preserves concatenation and the measurable-selection
class.  We use the standard stopped strong-formulation dynamic-programming
principle; see \citet{FlemingSoner2006}.  Applying the stopped calendar-time DPP at
$t(\widehat\sigma)$ and using triplewise equality of objectives yields
\eqref{eq:precision-time-dpp-app}.  Factoring
$-e^{-\gamma x}$ from that identity reverses the two optimizations inside the
positive exponential loss and gives
\eqref{eq:precision-time-ce-dpp-app}.  The bounds on $\zeta$ follow by
integrating $\underline I\le I(e)\le\overline I$.  Finally, within-vintage
rectangularity is unchanged because the conditional model set and its
conditioning sigma-fields are re-indexed, not enlarged or rebuilt.
\end{proof}

\paragraph{Step 2: equivalence of the Bellman equations.}
In $(t,m,p)$ coordinates, before the precision-time normalization, the
compatible CARA certainty-equivalent equation is
\begin{align*}
0=f_t+\sup_{\pi,e}\Big\{&\pi m-c(e)+I(e)f_p
 +\frac{I(e)}{2p^2}f_{mm}
 -\frac{\gamma}{2}\big[\sigma^2\pi^2+2p^{-1}\pi f_m
 +p^{-2}I(e)f_m^2\big]\\
&-\frac{k}{p}\big|\pi+p^{-1}I(e)f_m\big|\Big\}.
\end{align*}
Put $\pi=I(e)u$ and collect the positive factor $I(e)$.  The normalized
precision-time Bellman equation is
\begin{align}
0=f_p+\sup_{u\in\mathbb R,\,0\le e\le\bar e}\Bigg\{&
 \frac{f_t-c(e)}{I(e)}+um
 -\frac{\gamma\sigma^2}{2}I(e)u^2
 -\frac{\gamma}{p}u f_m
 +\frac{1}{2p^2}f_{mm}
 -\frac{\gamma}{2p^2}f_m^2\nonumber\\
&-\frac{k}{p}\left|u+\frac{f_m}{p}\right|\Bigg\}.
\label{eq:precision-time-hjb-app}
\end{align}
Equivalently the calendar equation is
$\sup_{u,e}I(e)\{f_p+\mathfrak H^P\}=0$.  Because $I(e)$ is strictly positive and bounded
above and below, a family of real numbers has supremum zero after multiplication
by $I(e)$ if and only if the corresponding normalized family has supremum zero;
the zero-level maximizers are also the same. This proves the HJB equivalence
and the feedback equivalence in the precision-time closure in the Main Paper.

\paragraph{Step 3: the comparison frontier and the correct backward sign.}
The precision-time transformation and Hamiltonian algebra above do not by
themselves provide viscosity uniqueness.  In particular, a backward
terminal-value equation written as
\begin{equation}\label{eq:precision-proper-backward-form-app}
 f_p+\mathcal H[f]=0
 \qquad\Longleftrightarrow\qquad
 -f_p-\mathcal H[f]=0
\end{equation}
must use the viscosity convention associated with the proper backward operator
$F=-f_p-\mathcal H$ in the standard viscosity sense of \citet{CrandallIshiiLions1992}.  Thus a subsolution satisfies
\begin{equation}\label{eq:precision-correct-subsuper-sign-app}
 a_U+\mathcal H[U]\ge0,
 \qquad
 a_V+\mathcal H[V]\le0
\end{equation}
for a supersolution $V$, with the usual interpretation at test jets.
The opposite convention cannot support the desired terminal comparison.  The
scalar equation
\[
 u_p+1=0,\qquad u(P)=0,
\]
provides an explicit counterexample: $U(p)=2(P-p)$ obeys
$U_p+1=-1\le0$, $V(p)=0$ obeys $V_p+1=1\ge0$, and
$U(P)=V(P)=0$, yet $U(p)>V(p)$ for every $p<P$.  Thus the opposite sign convention fails terminal comparison; a valid comparison
argument must be formulated from the proper backward operator throughout.

For the present unbounded portfolio dualization, a whole-line quadratic-growth
comparison theorem would require a consistent doubling-of-variables argument,
the $\sup_e\inf_{q,\alpha}$ selection order, a strict coercive barrier, the
terminal/state-constraint boundary treatment, and a whole-line cutoff. The
Da Lio--Ley framework \citep{DaLioLey2006} provides a route under such
hypotheses. Here the precision-time HJB is used as the Bellman characterization
at smooth points and as the viscosity equation associated with the stopped DPP;
the welfare-curvature and temporal-support results use this stopped
characterization rather than a global PDE uniqueness theorem.

\paragraph{Step 3b: exact reachable precision domain without a PDE-comparison inference.}
Fix an initial $(t_0,m_0,p_0)$ and put
$I_-:=\sigma^{-2}$ and $I_+:=\sigma^{-2}+\bar e$.  From
\eqref{eq:precision-time-calendar},
\[
 I_+^{-1}\le\frac{\dd\bar t_p}{\dd p}\le I_-^{-1},
\]
so every admissible trajectory satisfies
\begin{equation}\label{eq:reachable-wedge-app}
 t_0+\frac{p-p_0}{I_+}
 \le \bar t_p\le
 t_0+\frac{p-p_0}{I_-}
 \qquad\text{up to the stopping time }\bar t_p=T.
\end{equation}
The exact interior state wedge reachable before the terminal date is therefore
\begin{equation}\label{eq:reachable-precision-domain-app}
 \boxed{
 \mathcal D_{\mathrm{reach}}
 =\left\{(p,t):
 \begin{array}{l}
 p_0\le p<p_0+I_+(T-t_0),\\[0.2em]
 t_0+\dfrac{p-p_0}{I_+}\le t<T,\\[0.45em]
 t\le t_0+\dfrac{p-p_0}{I_-}
 \end{array}\right\}.}
\end{equation}
Equivalently, for $t<T$,
\[
 t_0+\frac{p-p_0}{I_+}
 \le t\le
 \min\!\left\{T,\,t_0+\frac{p-p_0}{I_-}\right\}.
\]
The upper slow-clock bound need not itself be below $T$: a faster admissible
precision clock may reach the same $p$ before the terminal date.  Accordingly,
requiring the slow-clock upper bound itself to lie below $T$ would be unnecessarily restrictive.

The wedge is not merely necessary.  If $p>p_0$ and $(p,t)$ satisfies
\eqref{eq:reachable-precision-domain-app}, the constant precision rate
$I=(p-p_0)/(t-t_0)$ lies in $[I_-,I_+]$ and reaches $(p,t)$; for $p=p_0$ the
only reachable point is $(p_0,t_0)$.  Thus any
$p\ge p_0+I_+(T-t_0)$ is unreachable in the strict interior before $T$, while
points satisfying the displayed wedge are reachable.  This is a pathwise/DPP
statement only.  It does not erase an artificial $p=P$ boundary when comparing
arbitrary viscosity candidates, and we make no such comparison claim.  Stopped
continuation values can still be nested consistently by the stochastic
dynamic-programming principle under the rectangular formulation
\citep{EpsteinSchneider2003}; exhaustion of the reachable domain is therefore a
stochastic-value/nonexplosion argument, not a boundary-free PDE uniqueness
theorem.

\paragraph{Step 4: feedback and open positive-research regions.}
At a point where the viscosity value is classical, the precision-time and
calendar-time maximizers coincide by Step 2. The joint Hamiltonian is strongly
concave under \cref{ass:strong-research-cost}, so the investor maximizer is
unique. If the one-sided derivative at $e=0$ satisfies
$\Gamma^C-c'(0)\ge\delta>0$, the lower boundary violates the KKT condition and
$e^*>0$. If $c'(\bar e)-\Gamma^C\ge\delta>0$, the upper boundary violates the
KKT condition and $e^*<\bar e$. When both hold, $0<e^*<\bar e$. For persistence, let \(\vartheta_n\to\vartheta\) and suppose that,
on every \(O'\Subset O\),
\(f^{\vartheta_n}\to f^\vartheta\) in the \(C^{1,2}(O')\) topology containing
the derivatives entering \(\Gamma^C\).  The coefficients in the KKT index are
continuous in the primitives, so
\(\|\Gamma^{C,\vartheta_n}-\Gamma^{C,\vartheta}\|_{L^\infty(O')}\to0\).
For all large \(n\) this norm is below \(\delta/2\), and the two strict
margins remain at least \(\delta/2\) on \(O'\).  Local-uniform convergence of
the values alone would not imply this derivative convergence.  In the
explicit Gaussian--CARA branch the finite-dimensional coefficient ODEs and
their continuous dependence on parameters supply the required jet stability.
Hence the conditional strict-KKT persistence statement used in the Main
Paper holds without assuming a uniformly elliptic covariance on the full
calendar-time state. No viscosity-comparison conclusion is used here.

\subsection{Abstract causal--noncausal separator}\label{app:causal-noncausal-separator}

\paragraph{Temporal-support decomposition.}
Write the equilibrium derivative as
$D\mathcal G=N+\mathcal V$, where $\mathcal V$ is strictly causal in time and
$N$ is the bounded residual noncausal channel.  Resolving the causal part first
gives
\[
 D_E\mathcal F=(I-\mathcal V)
 [I-(I-\mathcal V)^{-1}N].
\]
Thus strict causal propagation enters through
$R=(I-\mathcal V)^{-1}$ and the master operator is $K=RN$.  In a pure pointwise same-date finite grid, $N=M$ is block diagonal.
For arbitrary Banach date blocks the causal factor gives a triangular-spectrum
inclusion; equality holds for finite-dimensional date blocks and can fail in
infinite dimension.  This boundary is proved sharply below.
For general noncommuting or temporally nonlocal $N$, the spectrum of $RN$ need
not equal that of $N$.  When $N$ is finite rank, however, part~\textup{(iii)}
shows that the nonzero master spectrum is captured exactly by the finite matrix
$L(I-\mathcal V)^{-1}U_0$; this is a factor compression, not a time
truncation.  The Volterra resolvent/iteration framework for weakly singular
kernels is classical; see \citet{GripenbergEtAl1990}.  The concrete
Gamma/Mittag--Leffler estimate needed here is proved next.

\begin{lemma}[Fractional Volterra iteration]\label{lem:fractional-volterra-iteration}
Let \(\mathcal V\) be a bounded operator on
\(\mathbb X=L^\infty(0,S;H)\) satisfying
\[
 \|(\mathcal Vx)(s)\|_H
 \le C\int_0^s(s-r)^{-\alpha}\|x(r)\|_H\,\dd r,
 \qquad 0\le\alpha<1.
\]
With \(\eta=1-\alpha\),
\begin{equation}\label{eq:fractional-volterra-power-app}
 \|\mathcal V^n\|
 \le\frac{[C\Gamma(\eta)S^\eta]^n}{\Gamma(n\eta+1)}.
\end{equation}
Consequently \(r(\mathcal V)=0\),
\((I-\mathcal V)^{-1}=\sum_{n\ge0}\mathcal V^n\) in operator norm, and
\begin{equation}\label{eq:mittag-leffler-resolvent-app}
 \|(I-\mathcal V)^{-1}\|
 \le E_\eta(C\Gamma(\eta)S^\eta).
\end{equation}
\end{lemma}

\begin{proof}
Let \(I^\eta\) be the Riemann--Liouville fractional integral
\[
 (I^\eta f)(s)=\frac1{\Gamma(\eta)}
 \int_0^s(s-r)^{\eta-1}f(r)\,\dd r.
\]
The hypothesis gives
\(\|\mathcal Vx\|\le C\Gamma(\eta)I^\eta\|x\|\) pointwise.  Since
\(I^\eta I^\gamma=I^{\eta+\gamma}\),
\[
 \|\mathcal V^nx(s)\|
 \le[C\Gamma(\eta)]^n(I^{n\eta}\|x\|)(s)
 \le\frac{[C\Gamma(\eta)]^ns^{n\eta}}
 {\Gamma(n\eta+1)}\|x\|_\infty.
\]
Stirling's formula yields \(\|\mathcal V^n\|^{1/n}\to0\) and summability.
Summing the displayed power bound proves the resolvent estimate.
\end{proof}

\paragraph{Function-space boundary for the weak kernel.}
For the Gaussian exponent $\alpha=1/2$, the time kernel
$(s-r)^{-1/2}$ is integrable on each time slice but is not square-integrable
on the time triangle.  No Hilbert--Schmidt argument is used: boundedness,
quasinilpotence, and the resolvent estimate follow directly from the
fractional-Volterra iteration.  Compactness of the master operator $K=RN$ in the statewise spectral
expansion below is a separate hypothesis, not a consequence of the weak
Volterra kernel.

\begin{proof}[Proof of \cref{thm:causal-noncausal-separator}]
If \(E_1,E_2\) are fixed points, put
\(d(s)=\|E_1(s)-E_2(s)\|_H\).  From
\eqref{eq:abstract-causal-lipschitz-main},
\[
 d(s)\le\frac C{1-\mu}\int_0^s(s-r)^{-\alpha}d(r)\,\dd r.
\]
Iterating and applying \cref{lem:fractional-volterra-iteration} with
$C/(1-\mu)$ gives
\[
 \|d\|_\infty
 \le
 \frac{[C\Gamma(\eta)S^\eta/(1-\mu)]^n}
 {\Gamma(n\eta+1)}\|d\|_\infty.
\]
The coefficient tends to zero, so $d=0$.

For the smooth factorization, the lemma makes
$R=(I-\mathcal V)^{-1}$ bounded and
\[
 (I-\mathcal V)[I-RN]=I-\mathcal V-N.
\]
Hence
\[
 D_E\mathcal F=(I-\mathcal V)(I-K),\qquad K=RN,
\]
and the invertible first factor shows that $D_E\mathcal F$ is singular exactly
when $1\in\sigma(K)$.  The Neumann criterion $\|RN\|<1$, together with
\eqref{eq:mittag-leffler-resolvent-app}, gives the quantitative sufficient
condition in the theorem.  When $N=0$, $K=0$ and the linearization is
invertible.  This establishes the separator itself; no equilibrium-fold claim
is inferred from singularity alone.
\end{proof}

\begin{proposition}[Pure same-date finite-grid spectrum boundary]
\label{prop:pure-samedate-grid-boundary}
Let \(\mathbb X_N=H_0\times\cdots\times H_{N-1}\), where the \(H_n\) are
complex Banach spaces (or complexifications of real Banach spaces).  If
\(\mathcal V_N\) is strictly block lower triangular,
\(M_N=\operatorname{diag}(M_0,\ldots,M_{N-1})\), and
\(K_N=(I-\mathcal V_N)^{-1}M_N\), then
\[
 \sigma(K_N)\subseteq\bigcup_{n=0}^{N-1}\sigma(M_n).
\]
Equality holds when every \(H_n\) is finite-dimensional, and can fail for
infinite-dimensional Banach blocks.
\end{proposition}

\begin{proof}
Because $\mathcal V_N$ is strictly block lower triangular,
$\mathcal V_N^N=0$ and
\[
 R_N=(I-\mathcal V_N)^{-1}
 =I+\mathcal V_N+\cdots+\mathcal V_N^{N-1}.
\]
Thus $K_N=R_NM_N$ is block lower triangular with diagonal blocks $M_n$.
All spectra are taken in the complex date spaces, or after complexifying real
ones.  If $z\notin\bigcup_n\sigma(M_n)$, every $zI-M_n$ is invertible.  Solving
$(zI-K_N)x=y$ successively from the first date to the last gives a bounded
triangular inverse; the finiteness of $N$ makes the recursive bound finite.
Hence
\[
 \sigma(K_N)\subseteq\bigcup_{n=0}^{N-1}\sigma(M_n).
\]
If all date spaces are finite-dimensional, block-triangular determinant
factorization gives
\[
 \det(zI-K_N)=\prod_{n=0}^{N-1}\det(zI-M_n),
\]
so equality holds.

The conclusion can fail in infinite dimensions. Let
$H=\ell^2(\mathbb N_0;\mathbb C)$, let $S$ be the unilateral right shift,
let $S^*$ be the left shift, and let
$Px=\langle x,e_0\rangle e_0$.  Set
\[
 M_0=I+S^*,\qquad M_1=I+S,\qquad
 \mathcal V_N=\begin{pmatrix}0&0\\P&0\end{pmatrix},\qquad
 R_N=\begin{pmatrix}I&0\\P&I\end{pmatrix}.
\]
Then
\[
 K_N-I=
 \begin{pmatrix}S^*&0\\P(I+S^*)&S\end{pmatrix}.
\]
Given $(y,z)\in H\oplus H$, define
\[
 x=(z_0-y_0)e_0+Sy,\qquad w=S^*z.
\]
Since $S^*x=y$, $P(I+S^*)x=z_0e_0$, and
$Sw=z-z_0e_0$, one has $(K_N-I)(x,w)=(y,z)$.  This inverse formula is bounded;
its kernel is trivial by the same two equations, so $K_N-I$ is invertible and
$1\notin\sigma(K_N)$.  Yet
$1\in\sigma(M_0)\cap\sigma(M_1)$ because neither $S^*$ nor $S$ is invertible.
Thus the inclusion may be strict in infinite dimension.  Off-diagonal causal
blocks may remove diagonal spectral points even though they cannot create
spectral points outside the union.
\end{proof}

\begin{proof}[Proof of part~\textup{(iii)} of
\cref{thm:causal-noncausal-separator}]
Write $R=(I-\mathcal V)^{-1}$, $U=RU_0$, so $K=UL$ and $B=LU$.
For every $\lambda\in\mathbb C$ for which $I_r-\lambda B$ is invertible,
straight multiplication gives the Woodbury identity
\begin{equation}\label{eq:finite-rank-woodbury-app}
 (I-\lambda UL)^{-1}
 =I+\lambda U(I_r-\lambda LU)^{-1}L.
\end{equation}
Conversely, interchanging $U$ and $L$ gives
\[
 (I_r-\lambda LU)^{-1}
 =I_r+\lambda L(I-\lambda UL)^{-1}U
\]
whenever $I-\lambda UL$ is invertible.  This proves
\eqref{eq:finite-rank-invertibility-main}.  For $\mu\ne0$, taking
$\lambda=\mu^{-1}$ yields
$\mu\in\sigma(UL)$ if and only if $\mu\in\sigma(LU)$, which is
\eqref{eq:finite-rank-spectrum-main}; the unit-spectrum/determinant statement is
the case $\mu=1$.  No temporal discretization is used: all causal history is
contained in $U=RU_0$, while only the noncausal factor dimension is reduced.
\end{proof}

\paragraph{Exact two-date mixed realization.}
The finite-factor reduction includes a minimal system with both causal and
noncausal channels.  Let
\[
 \mathcal V_\nu=\begin{pmatrix}0&0\\ \nu&0\end{pmatrix},\qquad
 R_\nu=(I-\mathcal V_\nu)^{-1}
 =\begin{pmatrix}1&0\\ \nu&1\end{pmatrix},\qquad
 N=u\ell^\top.
\]
The compression matrix is scalar and the only possible nonzero master
eigenvalue is
\begin{equation}\label{eq:two-date-mixed-eigenvalue-app}
 \boxed{\rho_\nu=\ell^\top R_\nu u
 =\ell_0u_0+\ell_1(u_1+\nu u_0).}
\end{equation}
When $\rho_\nu>0$, the fixed-state timing threshold is
$\lambda_{\mathrm{spec}}(\nu)=\rho_\nu^{-1}$.  Hence
\[
 \frac{\dd\lambda_{\mathrm{spec}}}{\dd\nu}
 =-\frac{\ell_1u_0}{\rho_\nu^2}.
\]
If $\ell_1u_0>0$, stronger causal transmission lowers the threshold.  This
does not contradict the same-date finite-grid result: $u\ell^\top$ is
temporally nonlocal rather than block diagonal.  The calculation realizes an
exact mixed $N+\mathcal V$ system without truncating causal propagation.

\paragraph{Settlement-deadline interpretation.}
If an independent completion time \(D\) makes the same-cycle path
\(\lambda=F_D(\Delta)\), a fixed-state spectral deadline solves
\(F_D(\Delta_{\mathrm{spec}})=r(K)^{-1}\) whenever \(r(K)>1\) and the cdf
attains that value. For exponential latency with rate \(\nu_D\),
\[
 \Delta_{\mathrm{spec}}=\nu_D^{-1}\log\!\frac{r(K)}{r(K)-1}.
\]
This is a structural fixed-state interpretation; the Main Paper's
value-weighted timing theorem supplies the financial multiplier and
identification statements.

\begin{proposition}[Spectral loading and fold certification]
\label{prop:spectral-fold-certification-app}
Let \(I-\mathcal V\) be invertible on a Banach space, put
\(R=(I-\mathcal V)^{-1}\) and \(K=RN\), and suppose \(K\) is compact with an algebraically simple positive eigenvalue
\(\rho=r(K)\), normalized right and left modes
\(\phi,\chi\), and \(\chi(\phi)=1\).

\textup{(i)} Along \(D_E\mathcal G=\mathcal V+\lambda N\), the formal
statewise threshold is \(\lambda_{\mathrm{spec}}=\rho^{-1}\). It belongs to
the admissible path \(\lambda\in[0,1]\) exactly when \(\rho\ge1\). At the
threshold, \(D_E\mathcal F\) is Fredholm of index zero, with one-dimensional
kernel and cokernel. For subcritical \(\lambda\), the operator resolvent decomposes as
\[
 (I-\lambda K)^{-1}x
 =\frac{\phi\,\chi(x)}{1-\lambda\rho}+H_\lambda x,
\]
where \(H_\lambda\) is uniformly bounded as
\(\lambda\uparrow\rho^{-1}\). If, in addition, for a Banach parameter space \(\mathbb Z\) the residual
\(\mathcal F:\mathbb X\times\mathbb Z\times[0,\rho^{-1})\to\mathbb X\)
is \(C^1\) in \((E,\xi)\) along subcritical local branches
\(\xi\mapsto E_\lambda(\xi)\) with \(E_\lambda(0)=E_*\), and
\[
 x_\lambda:=R\mathcal F_\xi(E_*,0,\lambda)
 \longrightarrow x_c
 \quad\text{in }\mathcal L(\mathbb Z,\mathbb X),
\]
then the sensitivity limits
\eqref{eq:critical-mode-limit-app}--\eqref{eq:critical-observable-sensitivity-app}
hold. These limits identify the inverse-linear residue: it is nonzero precisely
when \(\chi\circ x_c\ne0\); for a bounded financial observable \(A\), the
observable residue is nonzero precisely when \(\chi\circ x_c\ne0\) and
\(A\phi\ne0\).

\textup{(ii)} Suppose in addition that the equilibrium residual is \(C^2\),
the critical point lies on an equilibrium branch, and \(\vartheta\) is a scalar
unfolding parameter. If
\[
 \langle\psi,\mathcal F_\vartheta\rangle\ne0,
 \qquad
 \frac12\langle\psi,D^2_{EE}\mathcal F[\phi,\phi]\rangle\ne0,
 \qquad \psi=R^*\chi,
\]
then the critical point is a classical local saddle-node. The Fredholm
index-zero condition is automatic here because \(I-\lambda K\) is identity
minus compact and the Volterra factor is invertible.
\end{proposition}

\begin{proof}[Proof of \cref{prop:spectral-fold-certification-app}\textup{(i)}]
Along the linear contemporaneity path,
\[
 D_E\mathcal F_\lambda=(I-\mathcal V)(I-\lambda K),
 \qquad K=(I-\mathcal V)^{-1}N.
\]
Normalize right and left eigenvectors by
$K\phi=\rho\phi$, $K^*\chi=\rho\chi$, and $\chi(\phi)=1$.  Put
$Y=\ker\chi$.  Then $Y$ is $K$-invariant and
$\mathbb X=\operatorname{span}\{\phi\}\oplus Y$.  Compactness of $K$ and
algebraic simplicity of $\rho$ isolate $\rho$ from the rest of the spectrum, so
$I-\lambda K|_Y$ is invertible with uniformly bounded inverse for $\lambda$
near $\rho^{-1}$.  Therefore
$I-\lambda K$ is invertible for $0\le\lambda<\rho^{-1}$ and at
$\lambda=\rho^{-1}$ has one-dimensional kernel $\operatorname{span}\{\phi\}$
and cokernel $\operatorname{span}\{\chi\}$.  The invertible Volterra factor
preserves the kernel dimension; the cokernel of $D_E\mathcal F$ is spanned by
$R^*\chi$, since $(I-\mathcal V)^*R^*=I$.

For each subcritical $\lambda$ in the admissible range, the $C^1$ hypothesis and invertibility just
proved give the local branch $\xi\mapsto E_\lambda(\xi)$ by the
implicit-function theorem. Differentiating
$\mathcal F(E_\lambda(\xi),\xi,\lambda)=0$ at $\xi=0$ and setting
$x_\lambda=R\mathcal F_\xi(E_*,0,\lambda)$ gives
\begin{equation}\label{eq:critical-sensitivity-general-app}
 D_\xi E_\lambda(0)=-(I-\lambda K)^{-1}x_\lambda.
\end{equation}
On the direct sum above,
\[
 (I-\lambda K)^{-1}x
 =\frac{\phi\,\chi(x)}{1-\lambda\rho}+H_\lambda x,
\]
where $H_\lambda$ is uniformly bounded as
$\lambda\uparrow\rho^{-1}$. If $x_\lambda\to x_c$ in
$\mathcal L(\mathbb Z,\mathbb X)$, then
\begin{equation}\label{eq:critical-mode-limit-app}
 (1-\lambda\rho)D_\xi E_\lambda(0)
 \longrightarrow-\phi\otimes(\chi\circ x_c)
\end{equation}
in operator norm. For any bounded financial observable $A$,
\begin{equation}\label{eq:critical-observable-sensitivity-app}
 (1-\lambda\rho)A D_\xi E_\lambda(0)
 \longrightarrow-(A\phi)\otimes(\chi\circ x_c).
\end{equation}
Thus a nonzero inverse-linear leading residue is present when the limiting
primitive parameter operator loads on the left critical mode; a bounded financial
observable inherits that leading residue when it also loads on the right critical
mode. If \(\chi\circ x_c=0\), the limits above imply only
\((1-\lambda\rho)D_\xi E_\lambda(0)\to0\) and do not exclude slower divergence.
For example, boundedness follows under the stronger rate condition
\(\chi\circ x_\lambda=O(1-\lambda\rho)\), together with the maintained
operator-norm convergence (hence boundedness) of \(x_\lambda\) and the uniform
boundedness of \(H_\lambda\). These are statewise spectral and sensitivity statements.
\end{proof}

\begin{proposition}[Exact rank-one financial resolvent]
\label{prop:rank-one-financial-resolvent-app}
Let \(I-\mathcal V\) be invertible on a Banach space \(\mathbb X\), put
\(R=(I-\mathcal V)^{-1}\), and let the noncausal channel be
\(N=a\otimes b\), with \(a\in\mathbb X\) and \(b\in\mathbb X^*\). Define
\(\chi=b(Ra)\). Whenever \(1-\lambda\chi\ne0\),
\begin{equation}\label{eq:rank-one-financial-resolvent-app}
 \boxed{
 (I-\mathcal V-\lambda a\otimes b)^{-1}
 =R+\frac{\lambda}{1-\lambda\chi}
 (Ra)\otimes(bR).}
\end{equation}
At a fixed equilibrium, let a scalar primitive \(z\) enter
\(\mathcal G\) with derivative \(g_z=D_z\mathcal G\), and let
\(\mathcal Y(E,z)\) be a differentiable financial outcome. Its exact local
equilibrium sensitivity is
\begin{align}
 \frac{\dd\mathcal Y}{\dd z}
 &=\mathcal Y_z+D_E\mathcal Y(Rg_z)
 +\frac{\lambda}{1-\lambda\chi}
   D_E\mathcal Y(Ra)\,b(Rg_z).\label{eq:rank-one-financial-outcome-app}
\end{align}
Thus the finite-distance same-cycle wedge is governed by the product
\begin{equation}\label{eq:rank-one-financial-loading-app}
 \boxed{\omega_{\mathcal Y,z}
 :=D_E\mathcal Y(Ra)\,b(Rg_z).}
\end{equation}
The rank-one singularity is financially silent at leading order for this
primitive--outcome pair if either the primitive misses the resolved left
channel, \(b(Rg_z)=0\), or the outcome misses the resolved right response,
\(D_E\mathcal Y(Ra)=0\).
\end{proposition}

\begin{proof}
The Sherman--Morrison identity gives
\eqref{eq:rank-one-financial-resolvent-app}; direct multiplication by
\(I-\mathcal V-\lambda a\otimes b\) gives the same verification without a
matrix representation. Differentiating the equilibrium equation gives
\(E_z=(I-D_E\mathcal G)^{-1}g_z\). Applying the chain rule for
\(\mathcal Y(E(z),z)\) yields
\eqref{eq:rank-one-financial-outcome-app}. The final statement follows from
the factorization in \eqref{eq:rank-one-financial-loading-app}.
\end{proof}

\begin{proof}[Proof of \cref{prop:spectral-fold-certification-app}\textup{(ii)}]
At the stated equilibrium, Fredholm index and defect dimensions pass through
the invertible factor $I-\mathcal V$.  Thus
$\ker D_E\mathcal F=\operatorname{span}\{\phi\}$, and with
$\psi=R^*\chi$ the cokernel is $\operatorname{span}\{\psi\}$.  Apply the
standard one-dimensional Lyapunov--Schmidt reduction for a simple defect
(cf.\ \citet{CrandallRabinowitz1971}).  Write
$E=E_*+s\phi+w(s,\mu)$ with $\mu=\vartheta-\vartheta_*$ and $w$ in a fixed
complement of the kernel.  Solving the range equation and pairing the remaining
scalar equation with $\psi$ gives
\[
 g(s,\mu)=a\mu+bs^2+o(|\mu|+s^2),
\]
where the derivatives of $w$ do not contribute to the displayed first
parameter and second kernel-direction coefficients because of the cokernel
pairing.  The coefficients are
\[
 a=\langle\psi,\mathcal F_\vartheta(E_*,\vartheta_*)\rangle,
 \qquad
 b=\frac12\langle\psi,
 D^2_{EE}\mathcal F(E_*,\vartheta_*)[\phi,\phi]\rangle.
\]
When $a\ne0$ and $b\ne0$, the implicit-function theorem gives
\[
 \mu(s)=-\frac ba s^2+o(s^2).
\]
Thus the equilibrium condition, transversality, and quadratic nondegeneracy
are all required; a fixed-state spectral singularity alone is not a
saddle-node theorem.
\end{proof}

\paragraph{Volterra estimate for a regularized Gaussian response.}
For a fixed stopped Gaussian regularization, the following sufficient
conditions rule out a hidden same-date atom and place the response in the
causal class used by the separator.

\begin{proposition}[Parabolic sufficient conditions for a strictly causal Gaussian response]
\label{prop:gaussian-volterra-realization}
Let \(\mathcal P\) be a bounded \(C^2\) stopped precision domain.  Assume
\[
 \mathscr K\in W^{3,\infty}(\mathcal P),\qquad
 \operatorname*{ess\,inf}_{p\in\mathcal P}\mathscr K(p)\ge \mathscr K_*>0,
\]
so that $\lambda_{\mathscr K}=\mathscr K'/\mathscr K\in W^{2,\infty}(\mathcal P)$.  Put
\[
 Z_\gamma=C^\gamma(\overline{\mathcal P}),\qquad
 \mathbb X_\gamma=C([0,S];Z_\gamma),\qquad
 Y=L^2(\mathcal P),\qquad
 \mathbb X_0=L^\infty(0,S;Y),
\]
where \(0<\gamma<1\).  The classical feedbacks used in the regularized
construction form a bounded subset of
\(C^{\gamma/2,\gamma}([0,S]\times\overline{\mathcal P})\), continuously
embedded in \(\mathbb X_\gamma\).  Fix a response-regularization vector $\boldsymbol\zeta$ consisting of positive viscosity and the stated smooth response-side saturation and barrier parameters.  Let
\(\mathcal U_\gamma\subset\mathbb X_\gamma\) be an open set containing this
feedback set and every line segment between feedbacks to be compared.  Suppose
the frozen regularized system is
\begin{equation}\label{eq:gaussian-frozen-evolution-app}
 \partial_sU=\mathcal A_{E(s)}(s)U+Q_{E(s)}(s),\qquad U(0)=U_0,
\end{equation}
with feedback-independent initial and boundary data, and assume:
\begin{enumerate}[label=\textup{(\roman*)},leftmargin=2.1em]
\item the operators \(\mathcal A_E(s)\) have a common dense parabolic domain
\(D\subset Y\), uniform sectorial bounds, and a uniform time-H\"older estimate
sufficient to generate the evolution family \(S_E(s,r)\);
\item the coefficient map is \(C^1\) from \(\mathcal U_\gamma\) into
\(C^{\gamma/2}([0,S];\mathcal L(D,Y))\times
C^{\gamma/2}([0,S];Y)\), and its linearized forcing has the pointwise form
\(\mathcal B_E(r)h(r)\);
\item for the response observation \(J\), the composed kernel
\[
 K_E(s,r):=JS_E(s,r)\mathcal B_E(r),\qquad 0\le r<s\le S,
\]
obeys both scale-compatible estimates
\begin{equation}\label{eq:gaussian-two-scale-kernel-v170}
 \|K_E(s,r)\|_{Z_\gamma\to Z_\gamma}
 +\|K_E(s,r)\|_{Y\to Y}
 \le C_{\boldsymbol\zeta}(s-r)^{-1/2},
\end{equation}
uniformly on the relevant bounded feedback set;
\item the parameter-to-solution map \(E\mapsto U[E]\) is Fr\'echet
differentiable in the stated \(\mathbb X_\gamma\) topology, its derivative is
the unique solution of the linearized nonautonomous equation with zero initial
and boundary data, and both the solution and first variation take values
continuously in \(\mathbb X_\gamma\).
\end{enumerate}
Then
\(\mathcal R^{\boldsymbol\zeta}:\mathcal U_\gamma\to\mathbb X_\gamma\) is
Fr\'echet differentiable in one and the same Banach feedback space, with
\begin{equation}\label{eq:gaussian-response-derivative-app}
 D\mathcal R^{\boldsymbol\zeta}(E)h(s)
 =\int_0^sK_E(s,r)h(r)\,\dd r,
 \qquad h\in\mathbb X_\gamma.
\end{equation}
The same integral formula defines the bounded kernel extension
\(\widetilde{D\mathcal R^{\boldsymbol\zeta}(E)}\in
\mathcal L(\mathbb X_0,\mathbb X_0)\).  On both scales it is strict causal
Volterra and quasinilpotent; hence
\(I-D\mathcal R^{\boldsymbol\zeta}(E)\) on \(\mathbb X_\gamma\) and
\(I-\widetilde{D\mathcal R^{\boldsymbol\zeta}(E)}\) on \(\mathbb X_0\) are invertible.
The $\mathbb X_0$ object is used only as the bounded linear kernel extension.
Feedback-dependent boundary data, simultaneous market-clearing equations, and
direct same-date observation terms belong to the bounded noncausal operator
$N$ in \cref{thm:causal-noncausal-separator}.
\end{proposition}

\begin{proof}
For \(h\in\mathbb X_\gamma\) and sufficiently small \(\eta\), openness of
\(\mathcal U_\gamma\) makes \(E_\eta=E+\eta h\) admissible.  The explicit
coefficient topology in part~\textup{(ii)} and the parameter-to-solution
differentiability assumed in part~\textup{(iv)} identify the limit of the
difference quotient as the unique linearized solution in the parabolic solution space:
\[
 \partial_s\dot U=\mathcal A_{E(s)}(s)\dot U
 +\mathcal B_E(s)h(s),\qquad \dot U(0)=0,
\]
with zero boundary variation.  The same estimates applied to the Taylor
remainder are uniform for \(h\) in the \(\mathbb X_\gamma\) unit ball, so the
response is Fr\'echet, not merely G\^ateaux, differentiable
\(\mathbb X_\gamma\to\mathbb X_\gamma\).  Variation of constants gives
\eqref{eq:gaussian-response-derivative-app}.

The \(Z_\gamma\)-bound in
\eqref{eq:gaussian-two-scale-kernel-v170} makes the derivative a bounded
strict Volterra operator on \(\mathbb X_\gamma\).  The \(Y\)-bound gives, for
\(h\in\mathbb X_0\),
\begin{equation}\label{eq:gaussian-extended-derivative-bound-v170}
 \|\widetilde{D\mathcal R^{\boldsymbol\zeta}(E)}h(s)\|_Y
 \le C_{\boldsymbol\zeta}\int_0^s(s-r)^{-1/2}\|h(r)\|_Y\,\dd r,
\end{equation}
which defines the claimed kernel extension without invoking density of H\"older
functions in the \(L^\infty\) norm.  On either scale, the \(n\)-fold kernel is
bounded by
\[
 \frac{(C_{\boldsymbol\zeta}\sqrt{\pi S})^n}{\Gamma(n/2+1)},
\]
so the spectral radius is zero and the Neumann resolvent converges in operator
norm.  Because the initial and boundary variations are zero, Duhamel's formula
contains only \(r<s\); no boundary lift or algebraic atom at \(r=s\) is hidden.

\paragraph{Gaussian coefficient check under the stated regularization hypotheses.}
On a bounded stopped cylinder, if the chosen positive-viscosity construction
uses a common stopped boundary domain and fixed smooth saturation/barrier maps
with the displayed uniform bounds, then the coefficient implications are
standard.  Positive viscosity gives the three equations the same uniformly
elliptic principal part; \(\mathscr K\ge \mathscr K_*>0\) gives
\(\mathscr K^{-1}\in W^{3,\infty}\) and
\(\lambda_{\mathscr K}=\mathscr K'/\mathscr K\in W^{2,\infty}\); and the remaining coefficients are smooth
Nemytskii compositions of the frozen feedback and normalized-vintage
coefficients.  Under those imposed choices, uniformly parabolic Schauder
estimates \citep{Lieberman1996} yield the
\(Z_\gamma\to Z_\gamma\) estimate in
\eqref{eq:gaussian-two-scale-kernel-v170}, while the energy estimate plus
one-derivative parabolic smoothing yields the \(Y\to Y\) bound with exponent
$1/2$.  Finally,
\(J(u,a_1,a_2)=u_p-\lambda_{\mathscr K}(a_1+2a_2)\) is bounded on the corresponding
one-derivative solution space.  These estimates establish the proposition for each fixed regularization
satisfying the stated stopped-domain bounds. A hard limit is governed by the
uniform causal-closure theorem below.
\end{proof}

\paragraph{Hard selector without smoothing.}
The hard selector is already sufficient for global well-posedness on a closed
invariant response class once the response has a nonlinear causal modulus; no
derivative of the projection is needed.

\begin{lemma}[Nonsmooth Volterra well-posedness]\label{lem:hard-volterra-closure}
Let \(Y\) be a Banach lattice, let
\(\mathfrak e:Y\to Y\) be pointwise \(L_{\mathfrak e}\)-Lipschitz, and let
\(\mathcal R\) satisfy, for some \(0<\beta\le1\),
\begin{equation}\label{eq:hard-volterra-closure-app}
 \|\mathcal R(E_1)(s)-\mathcal R(E_2)(s)\|_Y
 \le C\int_0^s(s-r)^{\beta-1}
 \|E_1(r)-E_2(r)\|_Y\,\dd r.
\end{equation}
Then \(E=\mathfrak e(\mathcal R(E))\) has at most one solution in
\(L^\infty(0,S;Y)\).  More precisely, if \(d(s)=\|E_1(s)-E_2(s)\|_Y\),
then for every \(n\ge1\),
\begin{equation}\label{eq:hard-volterra-iterate-app}
 d(s)\le
 \frac{[L_{\mathfrak e}C\Gamma(\beta)]^n}
 {\Gamma(n\beta)}
 \int_0^s(s-r)^{n\beta-1}d(r)\,\dd r,
\end{equation}
and hence
\begin{equation}\label{eq:hard-volterra-sup-iterate-app}
 \|d\|_{L^\infty(0,s)}
 \le
 \frac{[L_{\mathfrak e}C\Gamma(\beta)s^\beta]^n}
 {\Gamma(n\beta+1)}
 \|d\|_{L^\infty(0,s)}.
\end{equation}
The coefficient on the right tends to zero as \(n\to\infty\).

Moreover, let \(\mathcal C\subset L^\infty(0,S;Y)\) be nonempty and closed and
suppose \(\mathcal T:=\mathfrak e\circ\mathcal R\) maps \(\mathcal C\) into
itself.  Then \(\mathcal T\) has exactly one fixed point in \(\mathcal C\),
and the Picard sequence \(E^{m+1}=\mathcal T(E^m)\) converges in
\(L^\infty\) from every \(E^0\in\mathcal C\).  Thus a closed invariant
response class upgrades the uniqueness statement to global well-posedness
without a horizon-wide small-gain restriction.

If responses \(\mathcal R_j\) obey \eqref{eq:hard-volterra-closure-app}
with the same \(C,\beta\) and satisfy \(\mathcal R_j(E)\to\mathcal R(E)\) in
\(L^\infty(0,S;Y)\) for every admissible feedback \(E\), then
\(\mathcal R\) inherits the same estimate.  If,
in addition, \(\mathcal R\) is Fr\'echet differentiable on the same Banach
space, then
\begin{equation}\label{eq:hard-derivative-volterra-app}
 \|D\mathcal R(E)h(s)\|_Y
 \le C\int_0^s(s-r)^{\beta-1}\|h(r)\|_Y\,\dd r.
\end{equation}
Let $M_A$ be a bounded pointwise multiplication operator on
$L^\infty(0,S;Y)$, $(M_Ax)(s)=A(s)x(s)$, with
$\operatorname*{ess\,sup}_s\|A(s)\|<\infty$.  Then
$M_A D\mathcal R(E)$ satisfies the same Volterra estimate with $C$ multiplied
by $\|M_A\|$, and is therefore quasinilpotent.  No such claim is made for
composition with an arbitrary bounded linear operator, which need not preserve
temporal support.
\end{lemma}

\begin{proof}
The selector inequality and \eqref{eq:hard-volterra-closure-app} give the
case \(n=1\).  The fractional-integral semigroup identity
\[
 \int_r^s(s-q)^{\beta-1}(q-r)^{n\beta-1}\,\dd q
 =\frac{\Gamma(\beta)\Gamma(n\beta)}
 {\Gamma((n+1)\beta)}(s-r)^{(n+1)\beta-1}
\]
proves \eqref{eq:hard-volterra-iterate-app} by induction; integration of the
kernel gives \eqref{eq:hard-volterra-sup-iterate-app}.  Stirling's formula
makes its coefficient vanish.  For existence on an invariant class, choose
\(E^0\in\mathcal C\) and write
\(d_m(s)=\|E^{m+1}(s)-E^m(s)\|_Y\).  Repeated application of the same
fractional estimate gives
\[
 \|d_m\|_{L^\infty(0,S)}
 \le
 \frac{[L_{\mathfrak e}C\Gamma(\beta)S^\beta]^m}
 {\Gamma(m\beta+1)}\,
 \|d_0\|_{L^\infty(0,S)}.
\]
The series of coefficients is a Mittag--Leffler series and is finite.
Hence \(\sum_m\|E^{m+1}-E^m\|_\infty<\infty\), so the Picard sequence is
Cauchy.  Closedness of \(\mathcal C\) puts its limit \(E^*\) in
\(\mathcal C\), and the Volterra Lipschitz bound makes \(\mathcal T\)
continuous; consequently \(E^*=\mathcal T(E^*)\).  The preceding uniqueness
argument makes this fixed point independent of the starting point.  To pass the response estimate to the limit, choose a subsequence of
representatives converging almost everywhere.  The right-hand Volterra
integrals converge in \(L^\infty\) under the integrable-kernel bound, so the
a.e. inequality is closed under strong \(L^\infty\) convergence.  Apply the
same closed-set argument to a subsequence of difference quotients at
\(E+th,E\); Fr\'echet differentiability then gives
\eqref{eq:hard-derivative-volterra-app}.  If $M_A$ is pointwise as stated,
\[
 \|M_A D\mathcal R(E)h(s)\|_Y
 \le \|M_A\|C\int_0^s(s-r)^{\beta-1}\|h(r)\|_Y\,\dd r,
\]
so the same fractional iteration gives
$r(M_A D\mathcal R(E))=0$.  Pointwise multiplication preserves temporal support; an arbitrary bounded
operator can mix future and past times and destroy it.
\end{proof}

\begin{theorem}[Uniform causal closure under vanishing regularization]
\label{thm:uniform-causal-closure-app}
Let \(\Theta\) be a parameter set,
\(\mathbb X=L^\infty(0,S;Y)\), and let
\(\mathcal C\subset\mathbb X\) be nonempty and closed. Let
\(\mathfrak e:Y\to Y\) be pointwise \(L_{\mathfrak e}\)-Lipschitz. For
\(\varepsilon>0\) and \(\theta\in\Theta\), suppose
\(\mathcal R_{\varepsilon,\theta}:\mathcal C\to\mathbb X\) satisfies the
common causal modulus
\begin{equation}\label{eq:uniform-regularized-causal-modulus-app}
 \|\mathcal R_{\varepsilon,\theta}(E_1)(s)
 -\mathcal R_{\varepsilon,\theta}(E_2)(s)\|_Y
 \le C_0\int_0^s(s-r)^{\beta-1}
 \|E_1(r)-E_2(r)\|_Y\,\dd r
\end{equation}
for one \(C_0<\infty\), \(0<\beta\le1\), uniformly in
\((\varepsilon,\theta)\), and suppose
\(\mathfrak e\circ\mathcal R_{\varepsilon,\theta}\) maps \(\mathcal C\)
into itself. Assume the response family is uniformly Cauchy as the
regularization vanishes:
\begin{equation}\label{eq:uniform-response-cauchy-app}
 \lim_{\varepsilon,\varepsilon'\downarrow0}
 \sup_{\theta\in\Theta}\sup_{E\in\mathcal C}
 \|\mathcal R_{\varepsilon,\theta}(E)-
 \mathcal R_{\varepsilon',\theta}(E)\|_\infty=0.
\end{equation}
Completeness defines the uniform limit
\(\mathcal R_{0,\theta}:\mathcal C\to\mathbb X\); write
\begin{equation}\label{eq:uniform-response-limit-app}
 \delta_\varepsilon
 :=\sup_{\theta\in\Theta}\sup_{E\in\mathcal C}
 \|\mathcal R_{\varepsilon,\theta}(E)-
 \mathcal R_{0,\theta}(E)\|_\infty.
\end{equation}
Then \(\delta_\varepsilon\to0\). Put
\[
 \mathfrak M_S:=E_\beta
 \!\left(L_{\mathfrak e}C_0\Gamma(\beta)S^\beta\right).
\]
Then:

\textup{(i)} Each limit response inherits
\eqref{eq:uniform-regularized-causal-modulus-app}. For every
\(\theta\), both the regularized and limit fixed-point problems have unique
solutions \(E_{\varepsilon,\theta}\) and \(E_{0,\theta}\) in \(\mathcal C\),
and
\begin{equation}\label{eq:uniform-fixed-point-error-app}
 \boxed{\|E_{\varepsilon,\theta}-E_{0,\theta}\|_\infty
 \le L_{\mathfrak e}\,\delta_\varepsilon\,\mathfrak M_S.}
\end{equation}
The bound is uniform in \(\theta\).

\textup{(ii)} If
\begin{equation}\label{eq:limit-response-parameter-modulus-app}
 \sup_{E\in\mathcal C}
 \|\mathcal R_{0,\theta_1}(E)-\mathcal R_{0,\theta_2}(E)\|_\infty
 \le L_\Theta d(\theta_1,\theta_2),
\end{equation}
then the limit equilibrium branch is globally single-valued and
\begin{equation}\label{eq:uniform-branch-lipschitz-app}
 \boxed{\|E_{0,\theta_1}-E_{0,\theta_2}\|_\infty
 \le L_{\mathfrak e}L_\Theta\mathfrak M_S
 d(\theta_1,\theta_2).}
\end{equation}
In particular, a connected parameter family satisfying these hypotheses has no
saddle-node generated by pure causal feedback.

\textup{(iii)} Suppose \(\mathcal R_{0,\theta}\) is Fr\'echet differentiable
on the same space and its derivative satisfies the inherited Volterra bound
with constant \(C_0\). If \(M_A\) is pointwise multiplication with
\(\|M_A\|\le L_A\), then
\begin{equation}\label{eq:uniform-inverse-bound-app}
 \boxed{\|(I-M_A D\mathcal R_{0,\theta})^{-1}\|
 \le E_\beta\!\left(L_A C_0\Gamma(\beta)S^\beta\right).}
\end{equation}
For a pointwise scalar (Nemytskii) hard-cost residual
\(\mathcal F_0(E)=c'(E)-\mathcal R_{0,\theta}(E)\), assume
\(c''\ge\chi_0>0\) on the admissible range and that
\(c''(E(\cdot))\) induces a bounded pointwise multiplication operator
(for example, \(0<\chi_0\le c''\le\chi_1<\infty\) on a bounded effort interval). Then
\begin{equation}\label{eq:uniform-hard-cost-inverse-bound-app}
 \boxed{\|D\mathcal F_0(E)^{-1}\|
 \le \chi_0^{-1}
 E_\beta\!\left(
 \frac{C_0\Gamma(\beta)S^\beta}{\chi_0}\right).}
\end{equation}
These inverse bounds contain no regularization parameter.
\end{theorem}

\begin{proof}
The uniform Cauchy condition and completeness of \(\mathbb X\) construct
\(\mathcal R_{0,\theta}\) uniformly in \((\theta,E)\). The Lipschitz property
of \(\mathfrak e\) implies
\(\mathfrak e(\mathcal R_{\varepsilon,\theta}(E))\to
\mathfrak e(\mathcal R_{0,\theta}(E))\); closedness of \(\mathcal C\) therefore
makes the limit map invariant. Passing \(\varepsilon\downarrow0\) in
\eqref{eq:uniform-regularized-causal-modulus-app} gives the same inequality for
\(\mathcal R_{0,\theta}\). Lemma~\ref{lem:hard-volterra-closure} then gives
existence and uniqueness on \(\mathcal C\). For the fixed points, set
\(d(s)=\|E_{\varepsilon,\theta}(s)-E_{0,\theta}(s)\|_Y\). The common selector
bound and \eqref{eq:uniform-response-limit-app} imply
\[
 d(s)\le L_{\mathfrak e}\delta_\varepsilon
 +L_{\mathfrak e}C_0\int_0^s(s-r)^{\beta-1}d(r)\,\dd r.
\]
Iterating the fractional integral gives
\(d(s)\le L_{\mathfrak e}\delta_\varepsilon
 E_\beta(L_{\mathfrak e}C_0\Gamma(\beta)s^\beta)\), proving part~\textup{(i)}.
The same argument with the first term replaced by
\(L_{\mathfrak e}L_\Theta d(\theta_1,\theta_2)\) proves part~\textup{(ii)}.

For part~\textup{(iii)}, the power estimate for the Volterra operator
\(M_A D\mathcal R_{0,\theta}\) yields the Neumann--Mittag--Leffler bound
\eqref{eq:uniform-inverse-bound-app}. Finally,
\[
 D\mathcal F_0=M_{c''}
 \left[I-M_{1/c''}D\mathcal R_{0,\theta}\right],
 \qquad \|M_{1/c''}\|\le\chi_0^{-1},
\]
and another application of the preceding bound gives
\eqref{eq:uniform-hard-cost-inverse-bound-app}.
\end{proof}

\paragraph{Gaussian application.}
Suppose a family of the stopped Gaussian regularizations in
\cref{prop:gaussian-volterra-realization} is defined on a common invariant
class, the constants \(C_{\boldsymbol\zeta}\) in
\eqref{eq:gaussian-two-scale-kernel-v170} remain bounded as
\(\boldsymbol\zeta\to0\), and the responses are uniformly Cauchy as in
\eqref{eq:uniform-response-cauchy-app}. Then
\cref{thm:uniform-causal-closure-app} constructs the selected unregularized
response, gives a unique hard equilibrium on \([0,S]\), and transfers the
regularization-independent bounds above. If the common estimates hold with
\(S=T\), the conclusion is whole-horizon. The fixed-regularization parabolic
proposition supplies the causal structure; a regularization-uniform kernel
bound and uniform Cauchy control of the response are the additional
requirements for the hard limit.

\paragraph{Supplementary local analytic certificate.}
The following calculation concerns the displayed reduced coefficient system and
is included as a local robustness check.

Set $s=T-t$ and suppose $\mathscr K=\mathscr K(p)$ is analytic and uniformly
positive on the compact precision neighborhood under consideration.  Put
$\lambda_{\mathscr K}=\mathscr K'/\mathscr K$, $L=\sigma^2p+s$, and
\[
 d=\lambda_{\mathscr K}(a_1+2a_2),\qquad
 \Gamma=u_p-d+B,\qquad E=\mathfrak e(\Gamma),\qquad a=\sigma^{-2}+E,
\]
where
\begin{align*}
 B&=\frac{s}{2\gamma pL}
 +\frac{\mathscr K^2s(2\sigma^2p+s)}{2\gamma p^2L^2},
 &S_1&=\frac{\mathscr K^2\sigma^2}{\gamma L^2},\\
 H_2&=\frac{\mathscr K^2s(2\sigma^2p+s)}{2\gamma p^2L^2},
 &F_0&=\frac{-\mathscr K^2\sigma^4+sL/p}{2\gamma\sigma^2L^2}
       +S_1+\sigma^{-2}H_2.
\end{align*}
The displayed pre-Legendre reduced system is
\begin{align}
 u_s&=a(u_p-d)+F_0+EB-c(E),
 \label{eq:hard-prelegendre-u-app}\\
 (a_1)_s&=a[(a_1)_p-\lambda_{\mathscr K}a_1]+S_1,
 \label{eq:hard-prelegendre-a1-app}\\
 (a_2)_s&=a[(a_2)_p-2\lambda_{\mathscr K}a_2+H_2],
 \label{eq:hard-prelegendre-a2-app}\\
 \Gamma&=u_p-d+B,\qquad
 \Gamma\in c'(E)+N_{[0,\bar e]}(E).
 \label{eq:hard-prelegendre-selector-app}
\end{align}
Define the constrained convex conjugate
$g(H)=\sup_{0\le e\le\bar e}\{eH-c(e)\}$.  On each separated selector regime,
$E=\mathfrak e(\Gamma)$ and $g(\Gamma)=E\Gamma-c(E)$, so
\[
 a(u_p-d)+EB-c(E)
 =\sigma^{-2}(u_p-d)+E(u_p-d+B)-c(E)
 =\sigma^{-2}(u_p-d)+g(\Gamma).
\]
Hence the displayed system is algebraically equivalent, term by term, to
\begin{align}
 u_s&=\sigma^{-2}(u_p-d)+F_0+g(\Gamma),
 \label{eq:hard-analytic-u-app}\\
 (a_1)_s&=a[(a_1)_p-\lambda_{\mathscr K}a_1]+S_1,
 \label{eq:hard-analytic-a1-app}\\
 (a_2)_s&=a[(a_2)_p-2\lambda_{\mathscr K}a_2+H_2].
 \label{eq:hard-analytic-a2-app}
\end{align}
On every separated selector regime $g'=\mathfrak e$, so
$a=\sigma^{-2}+g'(\Gamma)$.

\begin{proposition}[Local analytic no-fold certificate for the displayed reduced Gaussian system]
\label{prop:reduced-gaussian-analytic-no-fold-app}
Fix $\mathcal P_1\Subset\mathcal P_0$.  Suppose the incoming trace on
$s=s_*$, the displayed coefficient functions, $\mathscr K$, and $c$ extend
real-analytically to a neighborhood of $\overline{\mathcal P_0}$, with
$\inf_{\overline{\mathcal P_0}}\mathscr K>0$ and $c''\ge\chi_0>0$.
Assume the incoming selector index is uniformly separated by $\delta>0$ from
$c'(0)$ and $c'(\bar e)$ and lies in one strict selector regime.  If a finite-
dimensional parameter $\vartheta$ is present, assume joint real analyticity in
$\vartheta$.  Then, for some $h>0$, the displayed reduced hard system has
exactly one real-analytic solution on
$[s_*,s_*+h]\times\mathcal P_1$, remains in the same selector regime with
margin at least $\delta/2$, and depends real-analytically on $\vartheta$.
Hence its local analytic solution set is a single graph and cannot exhibit a
saddle-node within this reduced class.
\end{proposition}

\begin{proof}[Proof of \cref{prop:reduced-gaussian-analytic-no-fold-app}]
On the no-research and capacity regimes,
\[
 (g(\Gamma),g'(\Gamma))=(0,0)
 \quad\text{and}\quad
 (g(\Gamma),g'(\Gamma))=(\bar e\Gamma-c(\bar e),\bar e),
\]
respectively.  On a strict interior regime, the analytic inverse-function
theorem and $c''\ge\chi_0>0$ give
\[
 g'(\Gamma)=(c')^{-1}(\Gamma),\qquad
 g(\Gamma)=\Gamma(c')^{-1}(\Gamma)-c((c')^{-1}(\Gamma)),
\]
so $g$ is analytic on each separated regime.  Under the proposition's
hypotheses, the right side of
\eqref{eq:hard-analytic-u-app}--\eqref{eq:hard-analytic-a2-app} is analytic in
$(s,p,U,U_p)$ and jointly analytic in any stated finite-dimensional parameter.
The surface $s=s_*$ is noncharacteristic because the coefficient of $U_s$ is
the identity.  The Cauchy--Kowalevski theorem \citep{John1982} therefore gives
a unique analytic solution near every point of
$\overline{\mathcal P_1}$.  The incoming data extend to a neighborhood of the
larger compact set $\overline{\mathcal P_0}$; a finite cover and uniqueness on
overlaps give one common positive slab over $\mathcal P_1$.

The index
$\Gamma=u_p-\lambda_{\mathscr K}(a_1+2a_2)+B$ is continuous in the analytic
solution norm.  Shrinking the slab preserves at least half of the incoming
$\delta$-margin, so the same analytic formula for $g$ remains valid.  Parametric
uniqueness makes the local solution set a single analytic graph, excluding a
saddle-node within this reduced analytic class.  This conclusion is local to the displayed reduced system and remains valid
while the selector margin, analytic control, and stopped-domain interior
persist.
\end{proof}

\paragraph{Conditional finite-grid hard recursion.}
Assume that a chosen reverse-time grid response has the triangular form
\[
 \Gamma^n=\mathcal R_n(E^0,\ldots,E^{n-1}),\qquad
 E^n=\mathfrak e(\Gamma^n).
\]
Then $E^0$ is determined from the incoming or terminal data, followed by
$E^1$, and so on.  Induction proves existence and uniqueness \emph{for this
displayed finite-grid recursion}, including the no-research and capacity faces.
Wherever the projected selector is differentiable,
$D\mathcal R_N$ is strictly block lower triangular and the fixed-point residual
\[
 \mathcal G_N(E)=E-\mathfrak e(\mathcal R_N(E))
\]
has derivative
\[
 D\mathcal G_N(E)
 =I-M_{\mathfrak e'(\mathcal R_N(E))}D\mathcal R_N.
\]
It is block lower triangular with identity diagonal blocks and is therefore
invertible.  On a strict interior active set, the equivalent residual
$c'(E)-\mathcal R_N(E)$ has positive diagonal blocks $c''(E^n)$.  At a
projection kink a classical derivative may fail to exist, but recursive
uniqueness remains.  This proves the finite-grid claim summarized in the Main text.  The statement is conditional on the
displayed triangular response; it does not derive that recursion from an
unspecified primitive continuum model.

\paragraph{Continuum response at fixed parabolic response regularization.}
Let \(E_1,E_2\in\mathcal U_\gamma\) be admissible H\"older feedbacks whose
line segment remains in \(\mathcal U_\gamma\).  The Banach-space fundamental
theorem of calculus and the two-scale kernel estimate give
\[
 \mathcal R^{\boldsymbol\zeta}(E_1)-\mathcal R^{\boldsymbol\zeta}(E_2)
 =\int_0^1D\mathcal R^{\boldsymbol\zeta}
 (E_2+\lambda(E_1-E_2))(E_1-E_2)\,\dd\lambda
\]
and hence \eqref{eq:abstract-causal-lipschitz-main}.  Apply
\cref{lem:hard-volterra-closure} with
\(L_{\mathfrak e}=\chi_0^{-1}\) and \(\beta=1/2\): the \emph{exact projected
hard-selector equation}
\[
 E=\mathfrak e(\mathcal R^{\boldsymbol\zeta}(E))
\]
has at most one solution, without a horizon-wide small-gain condition.  If a
nonempty closed response class is invariant under
$\mathfrak e\circ\mathcal R^{\boldsymbol\zeta}$, the existence part of
\cref{lem:hard-volterra-closure} makes that solution unique and Picard-reachable
from every starting point in the class.

At an interior equilibrium define
\[
 \mathcal V_E=M_{1/c''(E)}D\mathcal R^{\boldsymbol\zeta}(E).
\]
Then
\[
 \|\mathcal V_E^n\|_{\mathbb X_\gamma\to\mathbb X_\gamma}
 \le
 \frac{(C_{\boldsymbol\zeta}\sqrt{\pi S}/\chi_0)^n}
 {\Gamma(n/2+1)},
\]
so
\[
 D_E[c'-\mathcal R^{\boldsymbol\zeta}]
 =M_{c''(E)}(I-\mathcal V_E)
\]
is invertible.  This is the actual linearization when the PDE response is
regularized but the selector is hard.

If the research selector is itself regularized, write it as
\(\mathcal S_{\delta,M}\).  The actual fixed-point residual is then
\[
 \mathcal F_{\boldsymbol\zeta,\delta,M}(E)
 =E-\mathcal S_{\delta,M}(\mathcal R^{\boldsymbol\zeta}(E)),
\]
not \(c'-\mathcal R^{\boldsymbol\zeta}\), and
\begin{equation}\label{eq:fully-regularized-linearization-app}
 D\mathcal F_{\boldsymbol\zeta,\delta,M}(E)
 =I-M_{\mathcal S'_{\delta,M}(\mathcal R^{\boldsymbol\zeta}(E))}
 D\mathcal R^{\boldsymbol\zeta}(E).
\end{equation}
The multiplier is bounded and the second factor is strict Volterra, so
\eqref{eq:fully-regularized-linearization-app} is invertible.  This proves the regularized continuum claim summarized in the Main text and removes the selector/response
regularization ambiguity.

\paragraph{Whole-horizon scope.}
The hard whole-horizon conclusion is supplied by
\cref{thm:uniform-causal-closure-app} when its common estimates hold with
\(S=T\). Local-uniform value convergence without the response convergence and
uniform causal modulus in that theorem is insufficient.

\subsection{Versioned quant-model release: global box reduction and criticality}
\label{app:model-release}

This subsection proves
\cref{thm:model-release-reduction,cor:gaussian-protocol-criticality,thm:financial-release-multiplier}.
Validation determines when a model version enters portfolio use; the validation
objective and internal organization are outside the release primitive.

\paragraph{Validated capital and the triangular response.}
Iterating \eqref{eq:model-release-capital} gives
\[
 q_n=\left(\prod_{j=0}^{n-1}\delta_j\right)q_0
 +\sum_{r<n}\kappa_r
  \left(\prod_{j=r+1}^{n-1}\delta_j\right)e_r.
\]
Multiplication by \(a_n^{\mathrm{val}}\) yields
\eqref{eq:model-release-Q}.  The first term is predetermined and nonnegative
because \(q_0\ge0\); together with the primitive baseline it is
\(\widetilde b_n^0\ge0\).  For a fixed scalar release price \(u\), strict
concavity of each supplier problem gives the projected best response
\eqref{eq:model-release-stagewise-map}.  Since stage \(n\) depends only on
earlier responses, this recursion defines a unique vector
\(\mathscr E_\vartheta(u)\) for every \(u\).

\begin{proof}[Proof of \cref{thm:model-release-reduction}]
Let \(u_\lambda(y)=(1-\lambda)\bar\Gamma+\lambda W_R'(y)\).

\emph{Global reduction and existence.}
If \(e\) is a capacity-constrained equilibrium and \(y=\ell^\top e\), every
coordinate satisfies the stagewise projected first-order condition, hence
\(e=\mathscr E_\vartheta(u_\lambda(y))\) and
\(F^{\mathrm{box}}(y,\vartheta;\lambda)=0\).  Conversely, a root \(y\) of
\eqref{eq:model-release-global-scalar} reconstructs a vector
\(e=\mathscr E_\vartheta(u_\lambda(y))\) whose precision is
\(\ell^\top e=\mathscr Y_\vartheta(u_\lambda(y))=y\); the projected first-order
conditions are therefore internally consistent and sufficient for every
supplier optimum.  Uniqueness of the stagewise recursion gives the claimed
one-to-one correspondence.

The projection maps are continuous, so \(F^{\mathrm{box}}\) is continuous.
Moreover, \(\mathscr Y_\vartheta(u)\in[0,\bar y]\).  Hence
\[
 F^{\mathrm{box}}(0,\vartheta;\lambda)\le0,\qquad
 F^{\mathrm{box}}(\bar y,\vartheta;\lambda)\ge0.
\]
The intermediate value theorem gives at least one equilibrium.

\emph{Monotonicity and the global Lipschitz bound.}
Take \(u_2\ge u_1\), write
\(e^j=\mathscr E_\vartheta(u_j)\), and \(\Delta e=e^2-e^1\),
\(\Delta u=u_2-u_1\).  Monotonicity of scalar projection and induction over
the triangular recursion give \(\Delta e\ge0\).  Its one-Lipschitz property
also gives componentwise
\[
 C\Delta e\le Q\Delta e+\ell\,\Delta u.
\]
Since
\[
 (C-Q)^{-1}
 =\sum_{k=0}^{J-1}(C^{-1}Q)^kC^{-1}\ge0,
\]
we obtain
\[
 0\le\Delta e\le(C-Q)^{-1}\ell\,\Delta u.
\]
Multiplication by \(\ell^\top\) proves
\eqref{eq:model-release-global-lipschitz} and
\eqref{eq:model-release-path-sum}.

Put \(L_+=\sup_{[0,\bar y]}[W_R'']_+\).  For \(y_2>y_1\), if
\(u_\lambda(y_2)\le u_\lambda(y_1)\), monotonicity of \(\mathscr Y_\vartheta\)
gives
\[
 F^{\mathrm{box}}(y_2)-F^{\mathrm{box}}(y_1)\ge y_2-y_1.
\]
Otherwise,
\[
 \mathscr Y_\vartheta(u_\lambda(y_2))
 -\mathscr Y_\vartheta(u_\lambda(y_1))
 \le\lambda\beta_R L_+(y_2-y_1).
\]
Under \eqref{eq:model-release-global-uniqueness}, the residual is strictly
increasing, proving uniqueness.

\emph{Fixed active faces.}
Let \(I\) be the free coordinates of one relative-interior face and fix every
coordinate outside \(I\) at zero or capacity.  Their triangular contributions
enter an affine baseline.  The free first-order conditions are
\[
 (C_I-Q_{II})e_I=b_I^{\mathrm{face}}(\vartheta)+\ell_Iu.
\]
The principal block remains lower triangular with positive diagonal, so
\[
 e_I=(C_I-Q_{II})^{-1}
 [b_I^{\mathrm{face}}(\vartheta)+\ell_Iu].
\]
Thus \(\mathscr Y_\vartheta\) is affine on the face with slope
\(\beta_{R,I}\) in \eqref{eq:model-release-face-beta}; if
\(\ell_I=0\), this slope is zero.  Differentiation gives
\eqref{eq:model-release-face-derivative}.  The free-coordinate response
derivative is
\[
 C_I^{-1}Q_{II}
 +\lambda W_R''(y)C_I^{-1}\ell_I\ell_I^\top.
\]
Resolving the nilpotent first term produces the same-cycle master block
\[
 \lambda W_R''(y)(C_I-Q_{II})^{-1}\ell_I\ell_I^\top,
\]
whose single possibly nonzero eigenvalue is
\(\lambda\beta_{R,I}W_R''(y)\).  The scalar and vector residuals are related by invertible elimination on
that cell.  Relative interior of the geometric face is not sufficient when a
bound multiplier vanishes.  If every bound inequality is strict in the sense
of \eqref{eq:model-release-raw-score}, the projected-response cell is locally
constant, the elimination is smooth, and
\(\partial_\vartheta F_I=\kappa_{\vartheta,I}\).  A scalar saddle-node therefore
lifts to a classical fold when strict complementarity holds,
\(\kappa_{\vartheta,I}>0\), and the scalar quadratic coefficient is nonzero.  At
loss of strict complementarity the global residual is only piecewise \(C^1\),
so a classical fold requires additional matching conditions.

\emph{Protocol comparative statics.}
Applying the face derivative to \(W_R^P\) gives
\(\partial_yF_I^P=1-\lambda\beta_{R,I}(W_R^P)''(y)\).
Subtracting the two gains proves
\eqref{eq:protocol-critical-gain-difference-main}. If \(\beta_{R,I}>0\) and
\((W_R^P)''(y)>0\), solving \(\partial_yF_I^P=0\) for \(\lambda\) gives
\eqref{eq:protocol-statewise-threshold-main}. No ordering follows without an
ordering of protocol-indexed optimized curvatures.
\end{proof}

\begin{proposition}[Strictly convex supplier costs]
\label{prop:model-release-convex-costs-app}
For each stage \(n\), replace the quadratic cost \(c_nx^2/2\) by
\(\mathcal C_n(x)\), where
\(\mathcal C_n\in C^2([0,\bar e_n])\) and
\(\mathcal C_n''(x)\ge\underline c_n>0\). For fixed \(u\), define the
stagewise response recursively as the unique maximizer on
\([0,\bar e_n]\) of
\[
 x\left[\widetilde b_n^0-\vartheta d_n
 +\sum_{r<n}Q_{nr}\mathscr E^{\mathcal C}_{\vartheta,r}(u)
 +\ell_nu\right]-\mathcal C_n(x),
\]
and put
\(\mathscr Y^{\mathcal C}_\vartheta(u)
 =\ell^\top\mathscr E^{\mathcal C}_\vartheta(u)\). For a reconstructed vector
\(e\), define its unnormalized stage score by
\begin{equation}\label{eq:model-release-general-score-app}
 s_n(e,u,\vartheta)
 :=\widetilde b_n^0-\vartheta d_n+\sum_{r<n}Q_{nr}e_r+\ell_nu.
\end{equation}
A general-cost projected-response cell is \emph{strictly complementary} at
\((e,u,\vartheta)\) when every free coordinate satisfies
\[
 0<e_n<\bar e_n,\qquad s_n(e,u,\vartheta)=\mathcal C_n'(e_n),
\]
every lower-bound coordinate satisfies
\[
 e_n=0,\qquad s_n(e,u,\vartheta)<\mathcal C_n'(0),
\]
and every upper-bound coordinate satisfies
\[
 e_n=\bar e_n,\qquad
 s_n(e,u,\vartheta)>\mathcal C_n'(\bar e_n).
\]

\textup{(i)} The full box-constrained equilibrium is in one-to-one
correspondence with roots of
\[
 F^{\mathrm{box}}_{\mathcal C}(y,\vartheta;\lambda)
 :=y-\mathscr Y^{\mathcal C}_\vartheta(u_\lambda(y))=0
\]
on \([0,\bar y]\), and at least one root exists.

\textup{(ii)} Let
\(\underline C=\diag(\underline c_0,\ldots,\underline c_{J-1})\) and
\[
 \overline\beta_R
 :=\ell^\top(\underline C-Q)^{-1}\ell.
\]
Then, for \(u_2\ge u_1\),
\[
 0\le
 \mathscr Y^{\mathcal C}_\vartheta(u_2)
 -\mathscr Y^{\mathcal C}_\vartheta(u_1)
 \le\overline\beta_R(u_2-u_1).
\]
Consequently the equilibrium is unique if
\[
 \lambda\overline\beta_R
 \sup_{0\le y\le\bar y}[W_R''(y)]_+<1.
\]

\textup{(iii)} On a strictly complementary smooth cell with free coordinates
\(I\), define
\[
 D_{\mathcal C,I}(e)
 :=\diag(\mathcal C_n''(e_n))_{n\in I},\qquad
 B_I(e):=D_{\mathcal C,I}(e)-Q_{II}.
\]
Then
\[
 \frac{\partial e_I}{\partial u}=B_I(e)^{-1}\ell_I,
 \qquad
 \frac{\partial e_I}{\partial\vartheta}=-B_I(e)^{-1}d_I,
\]
so
\[
 \partial_yF_I=1-\lambda\beta_{R,I}(e)W_R''(y),
 \qquad
 \partial_\vartheta F_I=\kappa_{\vartheta,I}(e),
\]
where
\[
 \beta_{R,I}(e)=\ell_I^\top B_I(e)^{-1}\ell_I,
 \qquad
 \kappa_{\vartheta,I}(e)=\ell_I^\top B_I(e)^{-1}d_I.
\]
If the costs and \(W_R\) are \(C^3\), put
\[
 v_I=B_I(e)^{-1}\ell_I,\qquad
 w_I=B_I(e)^{-\top}\ell_I,
\]
and
\[
 \Xi_I(e):=\sum_{n\in I}
 (w_I)_n\,\mathcal C_n'''(e_n)\,(v_I)_n^2.
\]
The full scalar quadratic coefficient is determined by
\begin{equation}\label{eq:model-release-general-Fyy-app}
 \boxed{
 \partial_{yy}F_I
 =-\lambda\beta_{R,I}(e)W_R'''(y)
 +\lambda^2[W_R''(y)]^2\Xi_I(e).}
\end{equation}
A scalar turning point lifts to a classical fold when it lies on the full
equilibrium, the cell remains strictly complementary,
\(\kappa_{\vartheta,I}(e)\ne0\), and
\(\partial_{yy}F_I\ne0\).
\end{proposition}

\begin{proof}
Strong convexity makes each constrained scalar best response single-valued,
continuous, and nondecreasing in its score. Strong monotonicity of
\(\mathcal C_n'+N_{[0,\bar e_n]}\) gives
\[
 0\le x_n(s_2)-x_n(s_1)
 \le\frac{s_2-s_1}{\underline c_n}
 \qquad(s_2\ge s_1).
\]
Induction through the triangular recursion yields a unique response vector for
each \(u\), continuity of the scalar residual, and the same endpoint signs as
in \cref{thm:model-release-reduction}; this proves part~\textup{(i)}. For two
release prices, the componentwise strong-monotonicity bound gives
\[
 \underline C\,\Delta e\le Q\Delta e+\ell\,\Delta u.
\]
The lower-triangular inverse \((\underline C-Q)^{-1}\) is nonnegative, proving
part~\textup{(ii)}.

On a strictly complementary cell, the free first-order conditions are smooth:
\[
 \mathcal C_I'(e_I)
 =b_I^{\mathrm{face}}-\vartheta d_I+Q_{II}e_I+\ell_Iu.
\]
Implicit differentiation proves the first-order formulas in
part~\textup{(iii)}. For the second derivative, write
\(v_I=B_I^{-1}\ell_I\), \(w_I=B_I^{-\top}\ell_I\), and
\(\beta=\ell_I^\top B_I^{-1}\ell_I\). Along the scalar residual,
\[
 \frac{\partial e_I}{\partial y}
 =\lambda W_R''(y)v_I,
 \qquad
 \frac{\partial\beta}{\partial y}
 =-\lambda W_R''(y)\Xi_I(e).
\]
Differentiating
\(\partial_yF_I=1-\lambda\beta W_R''(y)\) gives
\eqref{eq:model-release-general-Fyy-app}. Smooth scalar/vector elimination and
the saddle-node theorem give the fold statement.
\end{proof}

Quadratic costs are used for piecewise-affine global responses, constant
closed-form cell exposures, and the simplified coefficient
\(-\lambda\beta_{R,I}W_R'''\); they are not required for the scalar reduction
or the temporal-support mechanism.
The financial materiality decomposition also survives pointwise: on a
general-cost smooth cell,
\[
 \frac{\dd\mathcal O}{\dd z}
 =\frac{\mathcal O'(y)\,\ell_I^\top B_I(e)^{-1}h_I}
 {1-\lambda\beta_{R,I}(e)W_R''(y)}.
\]
Only the exposure and primitive loading become state dependent.

The marginal-cost thresholds cannot be replaced by the quadratic raw-score
thresholds. For example, let
\(\mathcal C(x)=x+x^2/2\) on \([0,1]\) and let the score be \(s=1/2\).
Since \(s<\mathcal C'(0)=1\), the unique response is \(x^*=0\), although the
quadratic normalization \(s/\mathcal C''=1/2\) lies in \((0,1)\). Strictness is
also essential at the correct threshold: the projection
\(x(u)=\Pi_{[0,1]}u\) has left and right derivatives zero and one at \(u=0\).
Thus relative interior of a geometric face alone does not make the response
smooth.

\paragraph{Gaussian curvature, global uniqueness, and the lift test.}
For \eqref{eq:model-release-gaussian-value}, put \(r=\sqrt{y+1/2}\).
The derivatives are the expressions in
\eqref{eq:timing-Vp-app}--\eqref{eq:timing-Vppp-app}, with \(e\) replaced by
\(y\).  In particular, \(W_R''\) is maximized at zero, is strictly decreasing
on its positive-curvature interval, and
\[
 W_R''(0)=\frac{39\sqrt2-20}{54},\qquad W_R'''<0
\]
throughout that interval.

\begin{proof}[Proof of \cref{cor:model-release-gaussian-lift}]
If
\(\lambda\beta_RW_R''(0)<1\), global uniqueness follows directly from
\cref{thm:model-release-reduction}.  At equality, for any \(y_2>y_1\),
the secant slope of \(W_R'\) is strictly smaller than \(W_R''(0)\), because
\(W_R''(y)<W_R''(0)\) for every \(y>0\).  The same comparison used in the
global Lipschitz proof is therefore strict for distinct points, so the
capacity-constrained equilibrium remains unique.  The equality contact of the
unconstrained reduced equation occurs at \(y=0\).

Now suppose
\(\lambda>\lambda_{\mathrm{on}}^{\fresh}\).
Continuity and strict decrease of \(W_R''\) give a unique \(y_f>0\) satisfying
\eqref{eq:model-release-fold-condition}.  On the full-interior face the
equilibrium equation is
\[
 f(y,\vartheta;\lambda)
 =y-a_0(\lambda)+\kappa_\vartheta\vartheta
  -\lambda\beta_RW_R'(y)=0.
\]
Under the stated transversality assumption \(\kappa_\vartheta>0\), at
\(\vartheta_f\) in \eqref{eq:model-release-fold-cost}, \(f=f_y=0\).  Positivity of \(\vartheta_f\) follows from
\[
 H(y):=W_R'(y)-yW_R''(y),\qquad
 H'(y)=-yW_R'''(y)>0,\qquad H(0)=W_R'(0)>0,
\]
and \(\widetilde b^0,\bar\Gamma\ge0\).
Furthermore
\[
 f_\vartheta=\kappa_\vartheta>0,\qquad
 f_{yy}=-\lambda\beta_RW_R'''(y_f)>0.
\]
The scalar saddle-node lemma gives
\eqref{eq:model-release-branch-expansion}.  Differentiating the reconstruction
formula gives
\[
 \frac{\dd e_\pm}{\dd\vartheta}
 =-A^{-1}d+\lambda A^{-1}\ell W_R''(y_\pm)
   \frac{\dd y_\pm}{\dd\vartheta}.
\]
The singular component has direction \(A^{-1}\ell\).  Continuity keeps both
branches on the full-interior face exactly when the critical reconstruction
\eqref{eq:model-release-fold-lift} lies strictly inside the box.  The same
argument applies on a fixed face with its affine baseline and
\(\beta_{R,I}\).
\end{proof}

The lift qualification cannot be omitted.  Consider
\[
 J=2,\quad C=I,\quad Q=0,\quad
 \ell=(1,1)^\top,\quad d=(100,1)^\top,\quad
 \widetilde b^0=0,\quad\bar\Gamma=0,\quad\lambda=1.
\]
Then \(\beta_R=2\), and the reduced turning point satisfies
\[
 y_f=0.050727793389\ldots,\qquad
 W_R'(y_f)=0.071088864797\ldots,\qquad
 \vartheta_f=0.000905444913\ldots.
\]
But its reconstructed effort is
\[
 e_f=(-0.01945563\ldots,\ 0.07018342\ldots)^\top.
\]
The reduced scalar equation has a turning point, while the full-interior
equilibrium does not.  Increasing upper capacities cannot repair the negative
first component.

The general monotonicity statement is equally precise.  If \(C,\ell\) are
fixed and \(0\le Q_1\le Q_2\) entrywise, then the finite path expansions give
\[
 \ell^\top(C-Q_1)^{-1}\ell
 \le\ell^\top(C-Q_2)^{-1}\ell.
\]
No ordering follows merely from changing the number of stages without a
nested embedding that preserves the existing primitives.

\paragraph{Validation-time and financial-materiality certificates.}
\begin{proof}[Proof of \cref{thm:financial-release-multiplier}]
Taking expectations in \eqref{eq:model-release-realized-settlement}, with
\(y\) fixed under the exogenous-timing assumption, gives
\eqref{eq:model-release-expected-settlement} and identifies the deterministic
release share with \(\lambda_\rho(\Delta)\). On \(\{D\le\Delta\}\),
\(e^{-\rho\Delta}\le e^{-\rho D}\le1\), which proves
\eqref{eq:model-release-distribution-free-bounds}. Conditional completion at
\(D=\Delta\) or \(D=0\) attains the lower or upper endpoint. Mixing these two
conditional laws attains every intermediate value, so the share interval is
sharp. Multiplication by \(\beta_{R,I}[W_R''(y)]_+\) gives
\eqref{eq:model-release-partial-id-interval} and its subcritical and
supercritical certificates.

For the financial statement, strict complementarity fixes the quadratic-cost
cell locally. After adding \(zh_I\) to its free scores, the cell equation is
\[
 A_Ie_I=b_I^{\mathrm{face}}-\vartheta d_I+zh_I
 +\ell_I\{(1-\lambda)\bar\Gamma+\lambda W_R'(y)\},
 \qquad A_I=C_I-Q_{II}.
\]
Bound coordinates are locally constant. Differentiating with respect to \(z\)
and multiplying by \(\ell_I^\top A_I^{-1}\) gives
\[
 \frac{\dd y}{\dd z}
 =\gamma_{h,I}+\lambda\chi_I(y)\frac{\dd y}{\dd z},
 \qquad
 \gamma_{h,I}=\ell_I^\top A_I^{-1}h_I,
 \qquad
 \chi_I(y)=\beta_{R,I}W_R''(y).
\]
The implicit-function theorem gives the multiplier
\((1-\lambda\chi_I)^{-1}\). Setting
\(\lambda=\lambda_\rho(\Delta)\) yields
\eqref{eq:model-release-financial-multiplier}. This is also the release-cell
specialization of \cref{prop:rank-one-financial-resolvent-app}: the open-loop
loading \(\omega_{\mathcal O,h,I}\) is multiplied by the feedback resolvent
\((1-\lambda\chi_I)^{-1}\).

When \(\chi_I(y)>0\) and
\(\overline\lambda_\Delta\chi_I(y)<1\), the multiplier is increasing and
continuous in \(\lambda\), so sharpness of the timing-share interval gives
\eqref{eq:model-release-multiplier-id}. Multiplying by the signed loading gives the corresponding sharp
conditional timing-sensitivity set. If a nondegenerate timing interval contains
\(1/\chi_I(y)\), it contains nonsingular shares arbitrarily close to that
value. Every share in the interval is attained by a compatible completion
law; hence \(|\dd\mathcal O/\dd z|\) has no finite uniform bound when the
loading is nonzero.

For \(D\sim\operatorname{Exp}(\nu)\), direct integration gives
\[
 \lambda_\rho(\Delta)
 =\int_0^\Delta\nu e^{-(\nu+\rho)s}\,\dd s
 =\frac{\nu}{\nu+\rho}
  \left(1-e^{-(\nu+\rho)\Delta}\right).
\]
For the Gaussian value, setting \(G=\beta_RW_R''(0)\),
\(r=\rho/\nu\), and \(d=\nu\Delta\) gives
\eqref{eq:model-release-dimensionless-index}. Its ceiling is \(G/(1+r)\), so
a finite unit-gain deadline exists exactly when \(G>1+r\), and inversion gives
\eqref{eq:model-release-critical-deadline}. Direct differentiation gives
\[
 \frac{\partial d_*}{\partial G}=-\frac{1}{G(G-1-r)}<0,
\]
which proves the protocol-deadline comparison in the Main Paper after setting
\(G=G_P\). At the critical deadline,
\[
 \left.\frac{\partial}{\partial d}
 [1-\mathfrak C(G,r,d)]\right|_{d=d_*}
 =-G e^{-(1+r)d_*}=-(G-1-r).
\]
A first-order expansion proves the pole
\eqref{eq:model-release-deadline-pole}. Strict decrease of the Gaussian curvature then gives the unique positive
unconstrained unit-gain location and the fold qualifications stated in the
corollary.

The positive-location conclusion uses the Gaussian qualification.
For \(W_R(y)=ay^2/2\), the reduced derivative
\(1-\lambda_\rho(\Delta)\beta_Ra\) is independent of \(y\); a zero-state gain
above one need not generate a positive turning point.
\end{proof}

\paragraph{Endogenous validation latency.}
Suppose instead that the value-weighted same-cycle share is
\(s(y)\in[0,1]\), with \(s\in C^1\), and that its dependence on the release
state is only through \(y\).  The global reduction remains exact after
replacing \(u_\lambda(y)\) by
\[
 u_s(y)=\bar\Gamma+s(y)[W_R'(y)-\bar\Gamma].
\]
On a smooth face,
\[
 \partial_yF_I
 =1-\beta_{R,I}
 \{s(y)W_R''(y)+s'(y)[W_R'(y)-\bar\Gamma]\}.
\]
The first term is the exogenous settlement channel.  The second is a release
elasticity: value-based fast tracking amplifies current feedback when
\(W_R'(y)>\bar\Gamma\), whereas complexity-induced delay can damp it.  If
validation time depends on the full vector \(e\), rather than on \(y\) alone,
the one-dimensional reduction need not survive and the additional derivative
must be retained in the finite-factor system.

\begin{proposition}[Exact multi-factor release compression]
\label{prop:model-release-multifactor-app}
Let \(A\in\mathbb R^{J\times J}\) be invertible,
\(L:\mathbb R^J\to\mathbb R^q\), \(U:\mathbb R^q\to\mathbb R^J\), and
\(\Psi\in C^2(\mathbb R^q)\).  The nonlinear release equilibrium
\[
 Ae=b+\lambda U\nabla\Psi(Le)
\]
is in one-to-one correspondence with
\[
 z=a+\lambda B\nabla\Psi(z),\qquad
 a:=LA^{-1}b,\qquad B:=LA^{-1}U.
\]
Every reduced root reconstructs uniquely as
\(e=A^{-1}[b+\lambda U\nabla\Psi(z)]\), and the full derivative is singular
if and only if
\[
 \det[I_q-\lambda B\nabla^2\Psi(z)]=0.
\]
Thus the complete reduced equilibrium and its critical locus depend on the
pair \((a,B)\).  Conditional on a fixed reduced root \(z\), the local
singularity test depends on causal memory through
\(B\nabla^2\Psi(z)\).
\end{proposition}

\begin{proof}
Multiplication by \(LA^{-1}\) gives the reduced equation; substitution proves
the converse.  The full residual derivative is
\[
 A-\lambda U\nabla^2\Psi(z)L
 =A[I_J-\lambda A^{-1}U\nabla^2\Psi(z)L].
\]
Sylvester's determinant identity reduces the second determinant to
\(\det[I_q-\lambda LA^{-1}U\nabla^2\Psi(z)]\).
\end{proof}

The input direction \(a\) and output loading \(B\) play different roles.  In the scalar case
\(q=1\), take \(\Psi(z)=z^4/4\), with \(a\ne0\) and \(B>0\) so that the
displayed fold parameter is positive.  The reduced equation and fold condition
are
\[
 z=a+\lambda Bz^3,\qquad 1=3\lambda Bz^2.
\]
They imply \(z_f=3a/2\) and
\[
 \boxed{\lambda_f=\frac{4}{27Ba^2}.}
\]
Thus two systems with the same exposure matrix \(B\) but different affine
input \(a\) have different fold loci.  Only after a reduced root has been
fixed is \(B\nabla^2\Psi(z)\) the complete local singularity statistic.

\paragraph{Quantitative frontier.}
The following panels summarize the deadline and completion-rate regions implied
by the release formulas.

\begin{figure}[ht]
\centering
\begin{minipage}[t]{.43\textwidth}
\centering
\includegraphics[width=\linewidth]{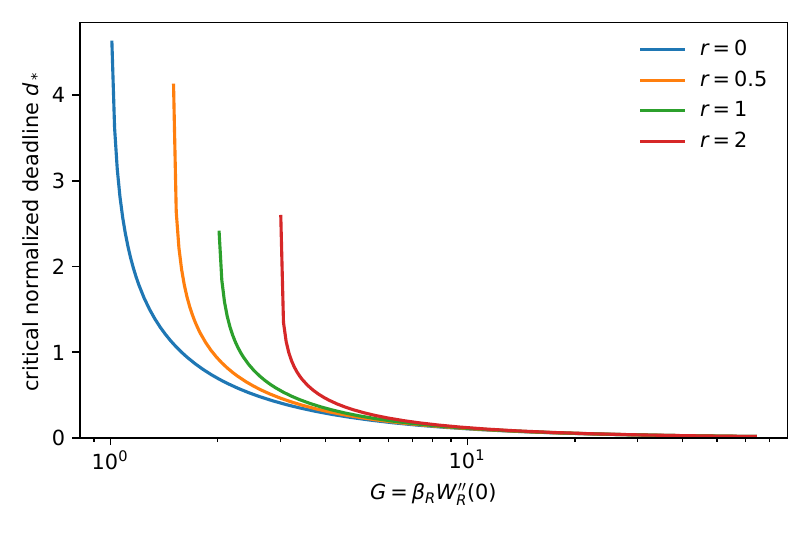}
\end{minipage}\hspace{.04\textwidth}
\begin{minipage}[t]{.43\textwidth}
\centering
\includegraphics[width=\linewidth]{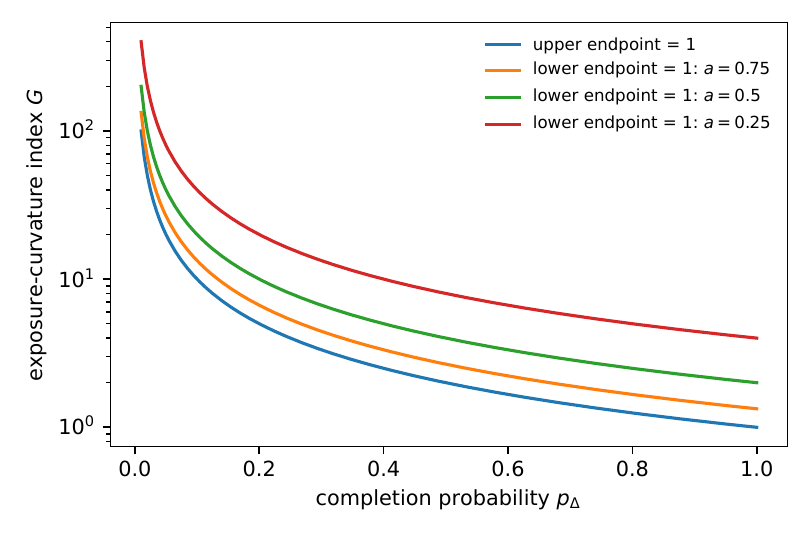}
\end{minipage}
\caption{Sharp validation-time certificates. The left panel shows the
critical deadline at \(y=0\), \(d_*\), which exists only for \(G>1+r\).
The right panel uses only completion probability \(p_\Delta\): the upper gain
endpoint reaches one at \(G=1/p_\Delta\), while the lower endpoint reaches one
at \(G=1/(a p_\Delta)\), with \(a=e^{-\rho\Delta}\). Between these frontiers, completion probability alone does not identify a
strict side of unit gain.}
\label{fig:model-release-deadline-app}
\end{figure}

\begin{corollary}[Scalar Gaussian release family]
\label{prop:gaussian-timing-family}\label{prop:gaussian-batch-fold}
\label{cor:timing-financial-criticality}
Take \(J=1\), \(C=1/32\), \(Q=0\), \(\ell=d=1\), and
\(\widetilde b^0=0\). Then \(\beta_R=\kappa_\vartheta=32\). Let
\[
 p=e+\frac12,\qquad m=\frac32,\qquad z=\sigma^2=\gamma=1,
\]
and define
\begin{equation}\label{eq:timing-family-value}
 V^{\fresh}(e)
 =\max_{\pi\ge0}\left\{
 \pi(m-p^{-1/2})-\frac12(1+p^{-1})\pi^2\right\}
 =\frac{(\frac32\sqrt p-1)^2}{2(p+1)}.
\end{equation}
With predetermined causal marginal value \(\bar\Gamma\ge0\), the projected
residual and potential are
\begin{align}
 F_\lambda(e,\vartheta)
 &=\vartheta+\frac e{32}-(1-\lambda)\bar\Gamma
 -\lambda(V^{\fresh})'(e),
 \label{eq:timing-family-residual}\\
 \mathcal W_\lambda(e,\vartheta)
 &=\lambda V^{\fresh}(e)+(1-\lambda)\bar\Gamma e
 -\vartheta e-\frac{e^2}{64},
 \qquad \partial_e\mathcal W_\lambda=-F_\lambda.
 \label{eq:timing-family-potential}
\end{align}
The global reduced onset is
\begin{equation}\label{eq:timing-family-lambda-c}
 \boxed{\lambda_{\mathrm{on}}
 =\frac{27}{16(39\sqrt2-20)}\approx0.0480026.}
\end{equation}
For \(0\le\lambda\le\lambda_{\mathrm{on}}\), the projected equilibrium is
unique for every \(\vartheta\); equality first produces a boundary tangency at
\(e=0\). If \(\lambda>\lambda_{\mathrm{on}}\), the unconstrained positive
branch has a unique turning point \(e_f(\lambda)\in(0,1/2]\) satisfying
\(32\lambda(V^{\fresh})''(e_f)=1\). At
\begin{equation}\label{eq:timing-family-theta-fold}
 \vartheta_f(\lambda)
 =(1-\lambda)\bar\Gamma+\lambda(V^{\fresh})'(e_f)-\frac{e_f}{32},
\end{equation}
the point is an interior saddle-node precisely when \(\bar e>e_f(\lambda)\).
At full same-cycle settlement,
\begin{equation}\label{eq:timing-family-batch-recovery}
 \boxed{e_f(1)=\frac12,\qquad
 \vartheta_f(1)=\frac9{64},\qquad
 \frac12F_{1,ee}(e_f,\vartheta_f)=\frac{21}{128}.}
\end{equation}
Whenever the fold is interior, the two positive branches satisfy
\begin{equation}\label{eq:timing-fold-branches}
 e_\pm(\vartheta)
 =e_f\pm
 \sqrt{\frac{\vartheta_f-\vartheta}
 {-\lambda(V^{\fresh})'''(e_f)/2}}
 +O(\vartheta_f-\vartheta),
\end{equation}
and hence
\begin{equation}\label{eq:timing-research-sensitivity}
 \left|\frac{\dd e_\pm}{\dd\vartheta}\right|
 \asymp(\vartheta_f-\vartheta)^{-1/2}.
\end{equation}
The positive-position portfolio and optimized value have nonzero derivatives at
the fold and therefore inherit the same pole. At \(\lambda=1\),
\begin{align}
 \left|\frac{\dd\pi(e_\pm)}{\dd\vartheta}\right|
 &\sim\frac32\sqrt{\frac2{21}}
  (\vartheta_f-\vartheta)^{-1/2},
 \label{eq:batch-portfolio-criticality}\\
 \left|\frac{\dd V^{\fresh}(e_\pm)}{\dd\vartheta}\right|
 &\sim\frac58\sqrt{\frac2{21}}
  (\vartheta_f-\vartheta)^{-1/2}.
 \label{eq:batch-value-criticality}
\end{align}
\end{corollary}

\paragraph{Nested Gaussian timing family.}
We prove \cref{prop:gaussian-timing-family}.
For fixed anticipated aggregate $e$, supplier $i$ solves a strictly concave
quadratic problem, hence
\begin{equation}\label{eq:timing-best-response-app}
 \operatorname{BR}_{\lambda,\vartheta}(e)
 =\Pi_{[0,\bar e]}
 \left(32\bigl[(1-\lambda)\bar\Gamma
 +\lambda(V^{\fresh})'(e)-\vartheta\bigr]\right).
\end{equation}
Consequently the projected equilibrium conditions are $F_\lambda(0,\vartheta)\ge0$
at the no-research face, $F_\lambda(\bar e,\vartheta)\le0$ at capacity, and
$F_\lambda(e,\vartheta)=0$ in the interior.

Put $r=\sqrt{e+1/2}\ge1/\sqrt2$.  Direct differentiation of
\eqref{eq:timing-family-value} gives
\begin{align}
 (V^{\fresh})'(e)
 &=\frac{(2r+3)(3r-2)}{8r(r^2+1)^2},
 \label{eq:timing-Vp-app}\\
 (V^{\fresh})''(e)
 &=-\frac{Q(r)}{8r^3(r^2+1)^3},
 \qquad
 Q(r)=9r^4+10r^3-18r^2-3,
 \label{eq:timing-Vpp-app}\\
 (V^{\fresh})'''(e)
 &=\frac{3P(r)}{16r^5(r^2+1)^4},
 \qquad
 P(r)=15r^6+20r^5-45r^4-15r^2-3.
 \label{eq:timing-Vppp-app}
\end{align}
The sign structure is elementary and global.  First,
\begin{equation}\label{eq:timing-Qprime-app}
 Q'(r)=6r(2r+3)(3r-2)>0
 \qquad(r\ge1/\sqrt2),
\end{equation}
while
$Q(1/\sqrt2)=-39/4+5\sqrt2/2<0$.  Hence $Q$ crosses zero at most once, so
$(V^{\fresh})''$ is positive only on one initial interval.  Second,
\begin{equation}\label{eq:timing-Pprime-app}
 P'(r)=10rQ(r),
 \qquad
 P(1/\sqrt2)=-159/8+5\sqrt2/2<0.
\end{equation}
Therefore $P<0$ throughout the positive-curvature interval, and
$(V^{\fresh})'''<0$ wherever $(V^{\fresh})''>0$.  Thus
$(V^{\fresh})''$ decreases strictly from
\begin{equation}\label{eq:timing-Vpp0-app}
 (V^{\fresh})''(0)=\frac{39\sqrt2-20}{54}
\end{equation}
to zero.  For every interior $e>0$ on the positive-curvature branch, define
\[
 \lambda_{\mathrm{spec}}(e)
 =\frac{1}{32(V^{\fresh})''(e)}.
\]
Because $(V^{\fresh})''$ is strictly decreasing there,
\[
 \lambda_{\mathrm{on}}
 =\inf_{\substack{e>0:\,(V^{\fresh})''(e)>0}}\lambda_{\mathrm{spec}}(e)
 =\lim_{e\downarrow0}\lambda_{\mathrm{spec}}(e)
 =\frac{27}{16(39\sqrt2-20)}.
\]
This equality between the lower envelope of local thresholds and the global
onset uses the scalar monotonicity just proved; it is not asserted by the
abstract separator for a general Banach-space equilibrium family.

The literal Bayesian inheritance comparator gives the protocol comparison used in the Main Paper. Preserve the
initial natural mandate
\[
 p_0=\frac12,\qquad k_0=\sqrt{p_0}=\frac1{\sqrt2}.
\]
Since posterior variance is \(p^{-1}\), the inherited worst-case drift penalty
is \(k_0/p\). On the same positive-position Gaussian quadratic,
\begin{equation}\label{eq:literal-inherited-objective-app}
 J^{\inh}(\pi,e)
 =\pi\left(\frac32-\frac{k_0}{p}\right)
 -\frac12(1+p^{-1})\pi^2,
 \qquad p=e+\frac12,
\end{equation}
with optimized value
\begin{equation}\label{eq:literal-inherited-value-app}
 V^{\inh}(p)=\frac{(3p-\sqrt2)^2}{8p(p+1)}.
\end{equation}
Direct differentiation gives
\begin{align}
 (V^{\inh})''(p)
 &=-\frac{(9+6\sqrt2)p^3-6p^2-6p-2}
 {4p^3(p+1)^3},
 \label{eq:literal-inherited-second-global-app}\\
 (V^{\inh})'''(p)
 &=\frac{3[(9+6\sqrt2)p^4-8p^3-12p^2-8p-2]}
 {4p^4(p+1)^4}.
 \label{eq:literal-inherited-third-global-app}
\end{align}
Let
\[
 N(p)=(9+6\sqrt2)p^3-6p^2-6p-2,
 \quad
 M(p)=(9+6\sqrt2)p^4-8p^3-12p^2-8p-2.
\]
The identity
\begin{equation}\label{eq:inherited-curvature-identity-app}
 \boxed{M(p)=pN(p)-2(p+1)^3}
\end{equation}
shows that \(M(p)<0\) whenever \((V^{\inh})''(p)>0\), because positive
curvature is equivalent to \(N(p)<0\). Hence inherited curvature decreases
strictly throughout its positive-curvature region and is maximized at the
admissible boundary \(p=1/2\). At that point
\begin{equation}\label{eq:literal-inherited-curvature-app}
 (V^{\inh})''(1/2)
 =\frac{86}{27}-\frac{4\sqrt2}{9}.
\end{equation}
Under \(\beta_R=32\), the global reduced onset is therefore
\begin{equation}\label{eq:literal-inherited-threshold-app}
 \boxed{\lambda_{\mathrm{on}}^{\inh}
 =\frac{27}{64(43-6\sqrt2)}
 \approx0.0122230462.}
\end{equation}
Together with the fresh calculation above, this proves
\cref{cor:gaussian-protocol-criticality}. In this benchmark
\(\lambda_{\mathrm{on}}^{\fresh}/\lambda_{\mathrm{on}}^{\inh}
\approx3.9272\): inheritance reaches reduced unit gain earlier under the same
research-supply normalization. The comparison concerns the reduced unit-gain
boundary; a feasible fold still requires the lift conditions in the Main
Paper. There is no protocol sign theorem across Gaussian primitives. More
generally, let the common baseline precision be \(p_0=q^2\), let \(z>0\) be
the confidence scale, and retain predictive variance \(1+p^{-1}\). Comparing
a fresh penalty \(z/\sqrt p\) with inheritance of the baseline natural budget
\(zq\), direct differentiation at \(p_0\) gives
\begin{equation}\label{eq:protocol-curvature-sign-app}
 \left[(V^{\fresh})''-(V^{\inh})''\right]_{p=q^2}
 =\frac{z(5mq^3+mq-8q^2z-4z)}
 {4q^4(q^2+1)^2}.
\end{equation}
The sign has a sharp signal-to-ambiguity representation. Put
\(x=mq/z\); the common positive-position face requires \(x>1\). The numerator
in \eqref{eq:protocol-curvature-sign-app} has the sign of
\[
 (5x-8)q^2+x-4=(5q^2+1)\{x-x_c(q)\},
 \qquad
 \boxed{x_c(q)=\frac{8q^2+4}{5q^2+1}>1.}
\]
Consequently
\[
 (V^{\fresh})''>(V^{\inh})''\iff x>x_c(q),\qquad
 (V^{\fresh})''<(V^{\inh})''\iff1<x<x_c(q).
\]
Thus signal-to-ambiguity configuration, not provenance by itself, orders the
protocol curvatures even while both protocols remain on the same positive
portfolio face. Part~\textup{(iv)} of \cref{thm:model-release-reduction}
transmits this curvature ordering to the unit-gain boundary.

Now
\begin{equation}\label{eq:timing-Fe-app}
 F_{\lambda,e}(e,\vartheta)
 =\frac1{32}-\lambda(V^{\fresh})''(e).
\end{equation}
If $\lambda\le\lambda_{\mathrm{on}}$, this derivative is nonnegative everywhere and
strictly positive for $e>0$.  At $\lambda=\lambda_{\mathrm{on}}$ its only zero is the boundary point $e=0$, so this is a boundary tangency rather than an interior fold.  Hence $F_\lambda(\cdot,\vartheta)$ is strictly
increasing on the admissible interval.  Existence is immediate: if
$F_\lambda(0,\vartheta)\ge0$, then $e=0$ satisfies the projected KKT
condition; if $F_\lambda(\bar e,\vartheta)\le0$, then $e=\bar e$ does;
otherwise the intermediate-value theorem gives an interior zero.  Strict
monotonicity makes these three cases mutually exclusive, proving global
projected uniqueness.

Let $\lambda\in(\lambda_{\mathrm{on}},1]$.  Since
$(V^{\fresh})''(1/2)=1/32$ and the positive-curvature branch is strictly
decreasing, the equation
$32\lambda(V^{\fresh})''(e)=1$ has exactly one solution
$e_f(\lambda)\in(0,1/2]$.  At the cost level
\eqref{eq:timing-family-theta-fold}, economic admissibility is automatic under
$\bar\Gamma\ge0$.  Indeed, the fold equation gives
\[
 \lambda(V^{\fresh})'(e_f)-\frac{e_f}{32}
 =\lambda\big[(V^{\fresh})'(e_f)-e_f(V^{\fresh})''(e_f)\big].
\]
With $H(e)=(V^{\fresh})'(e)-e(V^{\fresh})''(e)$,
$H(0)=(V^{\fresh})'(0)>0$ and
$H'(e)=-e\,\partial_e^3V^{\fresh}(e)>0$ on the positive-curvature fold
interval.  Thus $H(e_f)>0$ and
$\vartheta_f=(1-\lambda)\bar\Gamma+\lambda H(e_f)>0$.  At this cost,
\[
 F_\lambda=F_{\lambda,e}=0,
 \qquad F_{\lambda,\vartheta}=1,
 \qquad
 \frac12F_{\lambda,ee}
 =-\frac\lambda2(V^{\fresh})'''(e_f)>0.
\]
The scalar Lyapunov--Schmidt/implicit-function argument therefore gives a
nondegenerate saddle-node.  The projected boundary condition also matters for
the equilibrium count.  Define
\[
 \vartheta_0(\lambda)=(1-\lambda)\bar\Gamma
 +\lambda(V^{\fresh})'(0).
\]
Then $e=0$ is a projected equilibrium exactly when
$\vartheta\ge\vartheta_0(\lambda)$.  At the fold, using
$1/32=\lambda(V^{\fresh})''(e_f)$,
\begin{align*}
 \vartheta_f-\vartheta_0
 &=\lambda\left[(V^{\fresh})'(e_f)-(V^{\fresh})'(0)
    -e_f(V^{\fresh})''(e_f)\right]\\
 &=\lambda\int_0^{e_f}
   \left[(V^{\fresh})''(s)-(V^{\fresh})''(e_f)\right]\,\dd s>0,
\end{align*}
because $(V^{\fresh})''$ is strictly decreasing on the positive-curvature
interval containing $[0,e_f]$.  Hence zero research is already an equilibrium
at the fold.  Moreover $F_{\lambda,e}<0$ on $(0,e_f)$ and
$F_{\lambda,e}>0$ immediately to the right of $e_f$; after the
positive-curvature interval, $(V^{\fresh})''\le0$ makes the latter inequality
global.  Thus, for
$\vartheta\in(\vartheta_0,\vartheta_f)$ sufficiently close to
$\vartheta_f$, $F_\lambda(0,\vartheta)>0$ and
$F_\lambda(e_f,\vartheta)<0$, producing one positive root on each side of
$e_f$.  If $\bar e>e_f$, continuity of the two local roots ensures that both
remain below capacity for all sufficiently close costs.  This proves the
zero-plus-two-positive projected phase diagram in the capacity-accessible
case.  If $\bar e=e_f$, the critical contact lies on the capacity boundary;
if $\bar e<e_f$, the feasible interval ends before the unconstrained fold, so
no interior-fold or three-equilibrium conclusion is asserted.

Moreover
\[
 \lambda_{\mathrm{spec}}(e)=\frac{1}{32(V^{\fresh})''(e)}
\]
is strictly increasing on $[0,1/2]$, so the fold locus satisfies
$e_f(\lambda)\downarrow0$ as $\lambda\downarrow\lambda_{\mathrm{on}}$
and reaches $1/2$ at $\lambda=1$.  At $e=1/2$,
\[
 (V^{\fresh})'=\frac5{32},\qquad
 (V^{\fresh})''=\frac1{32},\qquad
 (V^{\fresh})'''=-\frac{21}{64},
\]
which gives \eqref{eq:timing-family-batch-recovery}.  Notice that
$\bar\Gamma$ enters only the level of
$\vartheta_f(\lambda)$; the Gaussian fold-onset share and fold curvature are
controlled by the contemporaneous derivative.

For the financial criticality statement, Taylor expansion at the fold gives
\[
 \vartheta
 =\vartheta_f-b_\lambda(e-e_f)^2+O(|e-e_f|^3),
 \qquad
 b_\lambda=-\frac\lambda2(V^{\fresh})'''(e_f)>0.
\]
Inverting this quadratic turning point proves
\eqref{eq:timing-fold-branches}--\eqref{eq:timing-research-sensitivity}.  The
positive-position portfolio extracted from the same Gaussian quadratic is
\[
 \pi(e)=\frac{\frac32p-\sqrt p}{p+1},
 \qquad
 \pi'(e)=\frac{p+3\sqrt p-1}{2\sqrt p\,(p+1)^2}>0
 \quad(p\ge1/2).
\]
The formula \eqref{eq:timing-Vp-app} likewise gives
$(V^{\fresh})'(e)>0$ for every $e\ge0$.  Chain rule therefore transfers the
same inverse-square-root pole to portfolio exposure and continuation value on
both local branches.  Finally, with
$\widehat{\mathcal W}_{\lambda,\pm}(\vartheta)
=\mathcal W_\lambda(e_\pm(\vartheta),\vartheta)$,
\eqref{eq:timing-family-potential} gives $\partial_e\mathcal W_\lambda=-F_\lambda$.
Hence along each interior equilibrium branch
$\dd\widehat{\mathcal W}_{\lambda,\pm}/\dd\vartheta=-e_\pm\to-e_f$,
whereas
$\dd^2\widehat{\mathcal W}_{\lambda,\pm}/\dd\vartheta^2
=-\dd e_\pm/\dd\vartheta$ has the inverse-square-root pole in
\eqref{eq:timing-research-sensitivity}.
At $\lambda=1$, $b_1=21/128$, $\pi'(1/2)=3/8$, and
$(V^{\fresh})'(1/2)=5/32$, which gives the exact constants in
\eqref{eq:batch-portfolio-criticality}--\eqref{eq:batch-value-criticality}.


\begin{thebibliography}{99}

\bibitem[Akey et~al.(2026)Akey, Robertson, and Simutin]{AkeyRobertsonSimutin2026}
Akey, P., Robertson, A. Z., \& Simutin, M. (2026). Noisy factors? The retroactive impact of methodological changes on the Fama--French factors. \emph{Review of Finance}, rfag002. doi: 10.1093/rof/rfag002.

\bibitem[Bain and Crisan(2009)]{BainCrisan2009}
Bain, A., \& Crisan, D. (2009). \emph{Fundamentals of Stochastic Filtering} (Stochastic Modelling and Applied Probability, Vol. 60). New York, NY: Springer. doi: 10.1007/978-0-387-76896-0.

\bibitem[Board of Governors et~al.(2026)]{FederalReserve2026ModelRisk}
Board of Governors of the Federal Reserve System, Federal Deposit Insurance Corporation, \& Office of the Comptroller of the Currency. (2026). \emph{Revised Guidance on Model Risk Management}, SR 26-2 / OCC Bulletin 2026-13 / FIL-15-2026, April 17. Available at \url{https://www.federalreserve.gov/supervisionreg/srletters/SR2602.htm}.

\bibitem[Bonnans and Shapiro(2000)]{BonnansShapiro2000}
Bonnans, J. F., \& Shapiro, A. (2000). \emph{Perturbation Analysis of Optimization Problems} (Springer Series in Operations Research and Financial Engineering). New York, NY: Springer. doi: 10.1007/978-1-4612-1394-9.

\bibitem[Bregman(1967)]{Bregman1967}
Bregman, L. M. (1967). The relaxation method of finding the common point of convex sets and its application to the solution of problems in convex programming. \emph{USSR Computational Mathematics and Mathematical Physics, 7}, 200--217. doi: 10.1016/0041-5553(67)90040-7.

\bibitem[Crandall et~al.(1992)Crandall, Ishii, and Lions]{CrandallIshiiLions1992}
Crandall, M. G., Ishii, H., \& Lions, P.-L. (1992). User's guide to viscosity solutions of second order partial differential equations. \emph{Bulletin of the American Mathematical Society, 27}, 1--67. doi: 10.1090/S0273-0979-1992-00266-5.

\bibitem[Crandall and Rabinowitz(1971)]{CrandallRabinowitz1971}
Crandall, M. G., \& Rabinowitz, P. H. (1971). Bifurcation from simple eigenvalues. \emph{Journal of Functional Analysis, 8}, 321--340. doi: 10.1016/0022-1236(71)90015-2.

\bibitem[Cummins et~al.(2023)Cummins, Gogolin, Kearney, Kiely, and Murphy]{CumminsEtAl2023}
Cummins, M., Gogolin, F., Kearney, F., Kiely, G., \& Murphy, B. (2023). Practice-relevant model validation: Distributional parameter risk analysis in financial model risk management. \emph{Annals of Operations Research, 330}, 431--455. doi: 10.1007/s10479-022-04574-x.

\bibitem[Da Lio and Ley(2006)]{DaLioLey2006}
Da Lio, F., \& Ley, O. (2006). Uniqueness results for second-order Bellman--Isaacs equations under quadratic growth assumptions and applications. \emph{SIAM Journal on Control and Optimization, 45}, 74--106. doi: 10.1137/S0363012904440897.

\bibitem[Epstein and Schneider(2003)]{EpsteinSchneider2003}
Epstein, L. G., \& Schneider, M. (2003). Recursive multiple-priors. \emph{Journal of Economic Theory, 113}, 1--31. doi: 10.1016/S0022-0531(03)00097-8.

\bibitem[Fleming and Soner(2006)]{FlemingSoner2006}
Fleming, W. H., \& Soner, H. M. (2006). \emph{Controlled Markov Processes and Viscosity Solutions} (2nd ed.; Stochastic Modelling and Applied Probability, Vol. 25). New York, NY: Springer. doi: 10.1007/0-387-31071-1.

\bibitem[Ganguli and Yang(2009)]{GanguliYang2009}
Ganguli, J. V., \& Yang, L. (2009). Complementarities, multiplicity, and supply information. \emph{Journal of the European Economic Association, 7}, 90--115. doi: 10.1162/JEEA.2009.7.1.90.

\bibitem[Gilboa and Schmeidler(1989)]{GilboaSchmeidler1989}
Gilboa, I., \& Schmeidler, D. (1989). Maxmin expected utility with non-unique prior. \emph{Journal of Mathematical Economics, 18}, 141--153. doi: 10.1016/0304-4068(89)90018-9.

\bibitem[Gripenberg et~al.(1990)Gripenberg, Londen, and Staffans]{GripenbergEtAl1990}
Gripenberg, G., Londen, S.-O., \& Staffans, O. (1990). \emph{Volterra Integral and Functional Equations} (Encyclopedia of Mathematics and its Applications, Vol. 34). Cambridge, UK: Cambridge University Press. doi: 10.1017/CBO9780511662805.

\bibitem[Gu et~al.(2020)Gu, Kelly, and Xiu]{GuKellyXiu2020}
Gu, S., Kelly, B., \& Xiu, D. (2020). Empirical asset pricing via machine learning. \emph{The Review of Financial Studies, 33}, 2223--2273. doi: 10.1093/rfs/hhaa009.

\bibitem[Guan, Jia, et~al.(2024)Guan, Jia, and Liang]{GuanJiaLiang2024}
Guan, G., Jia, Y., \& Liang, Z. (2024). Robust portfolio selection under state-dependent confidence set. \emph{arXiv:2409.19571}. doi: 10.48550/arXiv.2409.19571.

\bibitem[Guan, Liang, et~al.(2024)Guan, Liang, and Xia]{GuanLiangXia2024}
Guan, G., Liang, Z., \& Xia, J. (2024). Equilibrium portfolio selection for smooth ambiguity preferences. \emph{Mathematics of Operations Research, 50}, 1042--1071. doi: 10.1287/moor.2023.0112.

\bibitem[Hellwig and Veldkamp(2009)]{HellwigVeldkamp2009}
Hellwig, C., \& Veldkamp, L. (2009). Knowing what others know: Coordination motives in information acquisition. \emph{The Review of Economic Studies, 76}, 223--251. doi: 10.1111/j.1467-937X.2008.00515.x.

\bibitem[Jensen et~al.(2026)Jensen, Kelly, Malamud, and Pedersen]{JensenKellyMalamudPedersen2026}
Jensen, T. I., Kelly, B., Malamud, S., \& Pedersen, L. H. (2026). Machine learning and the implementable efficient frontier. \emph{The Review of Financial Studies}, hhag022. doi: 10.1093/rfs/hhag022.

\bibitem[John(1982)]{John1982}
John, F. (1982). \emph{Partial Differential Equations} (4th ed.; Applied Mathematical Sciences, Vol. 1). New York, NY: Springer.

\bibitem[Lakner(1998)]{Lakner1998}
Lakner, P. (1998). Optimal trading strategy for an investor: The case of partial information. \emph{Stochastic Processes and their Applications, 76}, 77--97. doi: 10.1016/S0304-4149(98)00032-5.

\bibitem[Lei and Pun(2024)]{LeiPun2024}
Lei, Q., \& Pun, C. S. (2024). Nonlocality, nonlinearity, and time inconsistency in stochastic differential games. \emph{Mathematical Finance, 34}, 190--256. doi: 10.1111/mafi.12420.

\bibitem[Liang et~al.(2026)Liang, Liu, and Ma]{LiangLiuMa2026}
Liang, Z., Liu, Y., \& Ma, X. (2026). Robust Bayesian portfolio optimization with discrepancy-based posterior ambiguity. \emph{arXiv:2606.17643}. doi: 10.48550/arXiv.2606.17643.

\bibitem[Liang et~al.(2025)Liang, Wang, and Xia]{LiangWangXia2025}
Liang, Z., Wang, S., \& Xia, J. (2025). Portfolio selection with costly information acquisition. \emph{arXiv:2508.12373}. doi: 10.48550/arXiv.2508.12373.

\bibitem[Liang and Ye(2024)]{LiangYe2024}
Liang, Z., \& Ye, Q. (2024). Optimal information acquisition for eliminating estimation risk. \emph{arXiv:2405.09339}. doi: 10.48550/arXiv.2405.09339.

\bibitem[Lieberman(1996)]{Lieberman1996}
Lieberman, G. M. (1996). \emph{Second Order Parabolic Differential Equations}. Singapore: World Scientific. doi: 10.1142/3302.

\bibitem[McLean and Pontiff(2016)]{McLeanPontiff2016}
McLean, R. D., \& Pontiff, J. (2016). Does academic research destroy stock return predictability? \emph{The Journal of Finance, 71}, 5--32. doi: 10.1111/jofi.12365.

\bibitem[Myatt and Wallace(2012)]{MyattWallace2012}
Myatt, D. P., \& Wallace, C. (2012). Endogenous information acquisition in coordination games. \emph{The Review of Economic Studies, 79}, 340--374. doi: 10.1093/restud/rdr018.

\bibitem[Revuz and Yor(1999)]{RevuzYor1999}
Revuz, D., \& Yor, M. (1999). \emph{Continuous Martingales and Brownian Motion} (3rd ed.; Grundlehren der mathematischen Wissenschaften, Vol. 293). Berlin, Germany: Springer. doi: 10.1007/978-3-662-06400-9.

\bibitem[Schneider(2013)]{Schneider2013}
Schneider, R. (2013). \emph{Convex Bodies: The Brunn--Minkowski Theory} (2nd ed.; Encyclopedia of Mathematics and its Applications, Vol. 151). Cambridge, UK: Cambridge University Press. doi: 10.1017/CBO9781139003858.

\bibitem[Wang et~al.(2026)Wang, Yu, Zhang, and Zhou]{WangYuZhangZhou2026}
Wang, Z., Yu, X., Zhang, J., \& Zhou, Z. (2026). Equilibrium under time-inconsistency: A new existence theory by vanishing entropy regularization. \emph{arXiv:2603.10321v2}. doi: 10.48550/arXiv.2603.10321.

\bibitem[Wu(2026)]{Wu2026FeedbackCycles}
Wu, C.-H. (2026). Feedback cycles in exploratory equilibria. \emph{arXiv:2607.18128}. doi: 10.48550/arXiv.2607.18128.

\bibitem[Yan\c{c}(2026)]{YancPriceRelearning2026}
Yan\c{c}, H. (2026). The Price of Relearning: Ambiguity Provenance in Dynamic Decisions. \emph{arXiv:2608.16255v2}, revised August 26, 2026. doi: 10.48550/arXiv.2608.16255.
\end{thebibliography}
\end{document}